%% file: main.tex
\documentclass[leqno,centertags,reqno,cmex10]{LMCS}
\def\dOi{9(4:10)2013}
\lmcsheading%
{\dOi}
{1--51}
{}
{}
{Jun.~21, 2012}
{Nov.~12, 2013}
{}


 \ACMCCS{[{\bf Theory of computation}]: Models of computation; Formal
   languages and automata theory; [{\bf Software and its
       engineering}]: Software organization and properties---Software
   system structures---Software system models---Petri nets; Software
   organization and properties---Software functional
   properties---Formal methods}

\usepackage{enumerate}
\usepackage{hyperref}

\usepackage{epsfig,psfrag}
\usepackage{amsthm,graphics,amscd,amssymb,enumerate,xspace,latexsym,amsbsy,smallmath,MnSymbol}

\usepackage{pgf}
\usepackage{tikz}
\usetikzlibrary{backgrounds,automata,arrows,calc,positioning,decorations.pathmorphing,petri}

\tikzset{background rectangle/.style={rounded corners,inner frame sep=0mm,draw}} 
\tikzstyle{every state}=[inner sep=4pt, minimum size=2mm,draw=blue,fill=blue!20,thick]
\tikzstyle{every node}=[inner sep=4pt, minimum size=4mm]

\tikzstyle{every place}=[inner sep=0pt, minimum size=3mm,draw=blue!20,fill=blue!20,thick]
\tikzstyle{every transition}=[inner sep=0pt, minimum size=3mm,fill=red,draw=red,thick]

\input comm

\makeatletter
\DeclareRobustCommand*\cal{\@fontswitch\relax\mathcal}
\makeatother

\begin{document}

\title[Priced Timed Petri Nets]{Priced Timed Petri Nets\rsuper*}

\author[P.~A.~Abdulla]{Parosh Aziz Abdulla\rsuper a}
\address{{\lsuper a}Uppsala University, Department of Information Technology, 
Box 337, SE-751 05 Uppsala, Sweden}

\author[R.~Mayr]{Richard Mayr\rsuper b}
\address{{\lsuper b}University of Edinburgh, School of Informatics,
10 Crichton Street, Edinburgh EH89AB, UK}

\thanks{{\lsuper b}Supported by Royal Society grant JP080268.} 

\titlecomment{{\lsuper*}Extended abstracts of parts of this work have appeared in 
the proceedings of FOSSACS 2009 \cite{AM:FOSSACS2009} and LICS 2011 \cite{AM:LICS2011}.}

\begin{abstract}
\noindent
We consider priced timed Petri nets, i.e., 
unbounded Petri nets where each token carries a
real-valued clock. Transition arcs are labeled with time intervals,
which specify constraints on the ages of tokens.
Furthermore, our cost model assigns token storage costs per time unit to places, 
and firing costs to transitions. This general model strictly subsumes
both priced timed automata and unbounded priced Petri nets. 

We study the cost of computations that reach a given control-state.
In general, a computation with minimal cost may not exist, due to strict inequalities in
the time constraints. 
However, we show that the infimum of the costs to reach a given control-state 
is computable in the case where all place and transition costs are non-negative.

On the other hand, if negative costs are allowed, then the question whether a
given control-state is reachable with zero overall cost becomes undecidable.
In fact, this negative result holds even in the simpler case of discrete time
(i.e., integer-valued clocks).
\end{abstract}

\maketitle

\section{Introduction}
\noindent
Petri nets \cite{Petri:thesis,Peterson:survey}
are a widely used model for the study and analysis of concurrent systems.
Many different formalisms have been proposed which extend 
Petri nets with clocks and real-time constraints, leading to 
various definitions of {\it Timed Petri nets (TPNs)}.
A complete discussion of all these formalisms is beyond the scope of this 
paper and the interested reader is referred to the surveys in
\cite{Srba:FORMATS2008,Bowden:TPN:Survey,BCHLR-atva2005}.

An important distinction is whether the time model is discrete 
or continuous. In discrete-time nets, time is interpreted as
being incremented in discrete steps and thus the ages of tokens are in a countable domain,
commonly the natural numbers. Such discrete-time nets have been studied 
in, e.g., \cite{PCVC99:TPN:Nondecidability,Escrig:etal:TPN}.
In continuous-time nets, time is interpreted as continuous, and the ages of
tokens are real numbers. Some problems for continuous-time nets have been
studied in \cite{Parosh:Aletta:bqoTPN,Parosh:Aletta:infinity,ADMN:dlg,AMM:LMCS2007}.

In parallel, there have been several works on extending the
model of timed automata \cite{AD:timedautomata} with {\it prices}
({\it weights}) (see e.g., \cite{AlurTP01:WTA,Kim:etal:priced,BouyerBBR:PTA,Behrman:HSCC2001,Bouyer:CACM2011,Jurdzinski:FORMATS2008,Bouyer:FORMATS2008}).
Weighted timed automata are suitable models for embedded systems, where 
one has to take into consideration the fact that the behavior of the system
may be constrained by the consumption of different types of resources.
Concretely, weighted timed automata
extend classical timed automata 
with a cost function $\costfun$ 
that maps every location and every transition to a 
nonnegative integer (or rational) number.
For a transition, $\costfun$ gives the cost of 
performing the transition. 
For a location, $\costfun$ gives the cost per time unit for 
staying in the location. 
In this manner, we can define, for each computation of the system, the accumulated 
cost of staying in locations and performing transitions 
along the computation.

Here we consider a very expressive model that subsumes all models
mentioned above. {\em Priced Timed Petri Nets} ({\it PTPN})
are a generalization of classic Petri nets \cite{Petri:thesis} with
real-valued (i.e., continuous-time) clocks, real-time constraints, 
and prices for computations.

Each token is equipped with a real-valued clock, representing the age of the token.
The firing conditions of a transition include the usual ones for Petri nets.
Additionally, each arc between a place and a transition is labeled with a 
time-interval whose bounds are natural numbers (or possibly $\infty$ as upper
bound). These intervals can be open, closed or half
open. Like in timed automata, this is used to encode strict or non-strict
inequalities that describe constraints on the real-valued clocks.
When firing a transition, tokens which 
are removed/added from/to places must have ages lying in the 
intervals of the corresponding transition arcs.
Furthermore, we add special {\em read-arcs} to our model. These affect the
enabledness of transitions, but, unlike normal arcs, they do not remove the
token from the input place. Read arcs preserve the exact age of the input
token, unlike the scenario where a token is first removed and then replaced.
Read arcs are necessary in order to make PTPN subsume the classic
priced timed automata of \cite{BouyerBBR:PTA}.

We assign a cost to computations via a cost function $\costfun$ that
maps transitions and places of the Petri net to natural numbers.
For a transition $t$, $\costfun(t)$ gives the cost of
performing the transition, while for a place $p$,
$\costfun(p)$ gives the cost per time unit per token in the place.
The total cost of a computation is given by the sum of all costs of
fired transitions plus the storage costs for storing certain numbers of tokens
on certain places for certain times during the computation.
Like in priced timed automata, having integers as costs and time bounds is not a restriction, because
the case of rational numbers can be reduced to the integer case.
In most of the paper we consider non-negative costs for transitions and
places. In the last section we show that allowing negative costs makes even
very basic questions undecidable.

Apart from the cost model, our PTPN are very close to 
the Timed-Arc Petri Nets of \cite{Srba:FORMATS2008}
(except for some extensions; see below) in the sense that time is
recorded by clocks in the individual tokens, of which there 
can be unboundedly many.
This model differs significantly from the 
{\em Time Petri Nets} (also described in \cite{Srba:FORMATS2008}) where time is measured by 
a bounded number of clocks in the transitions.
In addition to the cost model, our PTPN also extend the models
of \cite{Srba:FORMATS2008}, \cite{Parosh:Aletta:bqoTPN,Parosh:Aletta:infinity,ADMN:dlg,AMM:LMCS2007}
and \cite{PCVC99:TPN:Nondecidability,Escrig:etal:TPN} in two other ways.
First, our PTPN model includes read-arcs, which is necessary to subsume
(priced) timed automata. Secondly, in PTPN, newly generated tokens 
do not necessarily have age zero, but instead their age is chosen
nondeterministically from specified intervals (which can be point intervals). 
Our PTPN model uses a continuous-time semantics (with real-valued clocks) like
the models considered in
\cite{Parosh:Aletta:bqoTPN,Parosh:Aletta:infinity,ADMN:dlg,AMM:LMCS2007}
and unlike the simpler discrete-time semantics (with integer-valued clocks)
considered in \cite{PCVC99:TPN:Nondecidability,Escrig:etal:TPN}.
We use a non-urgent semantics where tokens can grow older even
if this disables the firing of certain transitions (sometimes for ever), like 
in the Timed-Arc Petri Nets of \cite{Srba:FORMATS2008} and unlike
in the Time Petri Nets of \cite{Srba:FORMATS2008}.

Thus, our PTPN are a continuous-time model which subsumes the continuous-time
TPN of
\cite{Parosh:Aletta:bqoTPN,Parosh:Aletta:infinity,ADMN:dlg,AMM:LMCS2007},
the Timed-Arc Petri Nets of \cite{Srba:FORMATS2008} 
and the priced timed automata of \cite{AlurTP01:WTA,Kim:etal:priced,BouyerBBR:PTA}.
It should be noted that PTPN are infinite-state in several different ways.
First, the Petri net itself is unbounded. So the number of tokens (and thus
the number of clocks) can grow beyond any bound, i.e., the PTPN can create and
destroy arbitrarily many clocks. In that PTPN differ from
the priced timed automata of \cite{AlurTP01:WTA,Kim:etal:priced,BouyerBBR:PTA}, which 
have only a finite number of control-states and only a fixed
finite number of clocks.
Secondly, every single clock value is a real number of which there are 
uncountably many.

\paragraph{\bf Our contribution.}
We study the cost to reach a given control-state 
in a PTPN. In Petri net terminology, this is called a control-state
reachability problem or a coverability problem. 
The related reachability problem (i.e., reaching a particular
configuration) is undecidable for (continuous-time and discrete-time) 
TPN \cite{PCVC99:TPN:Nondecidability}, even without taking costs into
account.
In general, a cost-optimal computation may not exist
(e.g., even in priced timed automata it can happen that there is no computation
of cost $0$, but there exist computations of cost $\le \epsilon$ for every 
$\epsilon > 0$).

Our main contribution is to show that the {\em infimum} of the costs to reach a given control-state 
is computable, provided that all transition and place costs are non-negative.

This cost problem had been shown to be decidable for the much simpler 
model of {\em discrete-time} PTPN in \cite{AM:FOSSACS2009}.
However, discrete-time PTPN do not subsume the
priced timed automata of \cite{BouyerBBR:PTA}. Moreover, the techniques
from \cite{AM:FOSSACS2009} do not carry over to the continuous-time domain
(e.g., arbitrarily many delays of length $2^{-n}$ for $n=1,2,\dots$
can can happen in $\le 1$ time).

On the other hand, if negative costs are allowed, then even very basic
questions become undecidable. In Section~\ref{sec:undecidability} 
we show that the question whether a
given control-state is reachable with zero overall cost is undecidable if
negative transition costs are allowed.
This negative result does not need real-valued clocks; 
it even holds in the simpler case of discrete time (i.e., integer-valued) clocks.

\paragraph{\bf Outline of Used Techniques.}
Since the PTPN model is very expressive, several powerful new techniques are
developed in this paper to analyze them. These techniques are
interesting in their own right and can be instantiated 
to solve other problems. 

In Section~\ref{tpn:section} we define PTPN and the priced coverability
problem, and describe its relationship with priced timed automata and Petri
nets. 
Then, in Sections~\ref{delta:section}--\ref{sec:cost-abstraction},
we reduce the priced coverability problem for PTPN to a 
coverability problem in an abstracted untimed model called
AC-PTPN. This abstraction is done by an argument similar to a 
construction in \cite{BouyerBBR:PTA}, where parameters indicating 
a feasible computation are contained in a polyhedron, which is described by a
totally unimodular matrix. However, our class of matrices is more general than
in \cite{BouyerBBR:PTA}, because PTPN allow the creation of new clocks
with a nonzero value.
The resulting AC-PTPN are still much more expressive than Petri nets,
because their configurations are arbitrarily long sequences of multisets
(instead of a single multiset in the case of normal Petri nets).
Moreover, the transitions of AC-PTPN are not monotone, because larger
configurations cost more and might thus exceed the cost limit.

In order to solve coverability for AC-PTPN, we develop a very general method to solve
reachability/coverability problems in infinite-state transition systems
which are more general than the 
well-quasi-ordered/well-structured transition systems of
\cite{Parosh:Bengt:Karlis:Tsay:general:IC,Finkel:Schnoebelen:everywhere:TCS}.
We call this method the {\em abstract phase construction}, and it is 
described in abstract terms in Section~\ref{sec:phase}.
In particular, it includes a generalization of the Valk-Jantzen construction \cite{ValkJantzen}
to arbitrary well-quasi-ordered domains.

In Section~\ref{sec:main}, we instantiate this abstract method with AC-PTPN
and prove the main result. This instantiation is nontrivial and requires
several auxiliary lemmas, which ultimately use the decidability 
of the reachability problem for Petri nets 
with one inhibitor arc \cite{Reinhardt:inhibitor,Bonnet-mfcs11}.
There exist close connections between timed Petri nets, Petri nets with one inhibitor arc, and
transfer nets.
In Section~\ref{sec:sdtn} we establish a connection between a subclass of
transfer nets called {\em simultaneous-disjoint-transfer nets} (SD-TN) and Petri nets
with one inhibitor arc, and in Section~\ref{sec:oracle} we show the
decidability of a crucial property by reducing it to a reachability problem
for SD-TN. By combining all parts, we show the main result, i.e.,
the computability of the infimum of the costs for PTPN.

The undecidability result for general costs of Section~\ref{sec:undecidability} is shown 
by a direct reduction from Minsky 2-counter machines, where a tradeoff 
between positive and negative costs is used to ensure a faithful simulation.

\section{Priced Timed Petri Nets}
\label{tpn:section}

\subsection{Preliminaries}
We use $\nat,\nnreals, \preals$ to denote the sets
of natural numbers (including 0), nonnegative reals, 
and strictly positive reals, respectively. 
For a natural number $k$, we use $\nat^k$ and $\nat^k_\omega$ 
to denote the set of vectors of size $k$ over $\nat$ and  $\nat
\cup\set{\omega}$, respectively
($\omega$ represents the first limit ordinal).
For $n\in\nat$, we use $\zeroto{n}$ to denote the set $\set{0,\ldots,n}$.
For $x\in\nnreals$, we use $\fractof{x}$ to denote the fractional part of $x$.
We use a set $\Intervals$ of intervals. An open interval is 
written as $\oointrvl  w z$ 
 where $\w\in\nat$ and $\z\in\nat\cup\set{\infty}$. Intervals can also 
be closed in one or both directions, 
 e.g. $[w:z]$ is closed in both directions  and
  $[w:z)$ is closed to the left and open to the right.

For a set $\aset$, we use $\wordsover\aset$ and $\msetsover\aset$
to denote the set of finite words and finite multisets over $\aset$, respectively.
We view a multiset $\mset$ over $\aset$ as
a mapping $\mset: \aset \mapsto \nat$.
Sometimes, we write finite multisets as lists (possibly with multiple occurrences), so 
both $\lst{2.4,2.4,2.4,5.1,5.1}$ and
$\lst{2.4^3\;,\;5.1^2}$ represent a multiset $\mset$ over $\nnreals$
where $\mset(2.4)=3$, $\mset(5.1)=2$ and $\mset(x)=0$ for $x\neq
2.4,5.1$.
%
%
For multisets $\mset_1$ and $\mset_2$ over $\aset$, we say that
$\mset_1\mleq\mset_2$ if $\mset_1(a)\leq\mset_2(a)$ for each $a\in A$.
We define $\mset_1+\mset_2$ to be the multiset $\mset$ where
$\mset(a)=\mset_1(a)+\mset_2(a)$, and (assuming $\mset_1\mleq\mset_2$) we
define $\mset_2 - \mset_1$ to be the multiset $\mset$ where
$\mset(a)=\mset_2(a)-\mset_1(a)$, for each $a\in \aset$.
We use $a\in\mset$ to denote that $\mset(a)>0$.
We use $\emptyset$ or $[]$ to denote the empty multiset and $\emptystring$ 
to denote the empty word.

Let $(\aset, \leq)$ be a poset. We define a partial order
$\wleq$ on $\wordsover{\aset}$ as follows. Let $a_1\dots a_n \wleq b_1\dots b_m$ iff
there is a subsequence $b_{j_1}\dots b_{j_n}$ of 
$b_1\dots b_m$ s.t. $\forall k \in \{1,\dots,n\}.\, a_k \le b_{j_k}$.
A subset $B\subseteq \aset$ is  said to be {\it upward closed}
in $\aset$
if $a_1\in B, a_2 \in \aset$ and $a_1\leq a_2$ implies
$a_2\in B$.
If $\aset$ is known from the context, then we say simply that 
$B$ is {\it upward closed}.
For $B\subseteq \aset$
we define the {\it upward closure} $B\!\uparrow$ to be the set
$\setcomp{a\in \aset}{\exists a'\in B:\; a'\leq a}$.
A {\it downward closed} set $B$  and the {\it downward closure} 
$B\!\downarrow$ are defined in a similar manner.
We use $a\!\uparrow$, $a\!\downarrow$, $a$ instead of 
$\set{a}\!\uparrow$, $\set{a}\!\downarrow$, $\set a$, respectively.

Given a transition relation $\trans{}$, we denote its transitive closure
by $\trans{+}$ and its reflexive-transitive closure by $\strans$.
Given a set of configurations $\confs$, let 
$\pre_{\trans{}}(\confs) = \{\conf' \,|\, \exists \conf \in \confs.\, \conf' \trans{} \conf\}$
and
$\pre^*_{\trans{}}(\confs) = \{\conf' \,|\, \exists \conf \in \confs.\, \conf' \strans \conf\}$.

\subsection{Priced Timed Petri Nets}
A {\em Priced Timed Petri Net (PTPN)} is a tuple $\ptpn=\ptpntuple$ where
$\states$ is a finite set of control-states and
$\places$ is a finite set of places.
Though control-states can in principle be encoded into extra places,
it is conceptually useful to distinguish them. From a modeling point of view,
this distinguishes the control-flow from the data. In technical proofs,
it is useful to distinguish the finite memory of the control
from the infinite memory of the (timed) tokens in the Petri net.
$\transitions$ is a finite set of transitions,
where each transition $\transition \in \transitions$ is of the form
$\transition = (q_1,q_2,\inputs,\reads,\outputs)$.
We have that $q_1,q_2 \in \states$ are the source control-state and target control-state,
respectively, and 
$\inputs,\reads,\outputs \in \msetsover{(\places \times \Intervals)}$
are finite multisets over $\places \times \Intervals$ which define the
input-arcs, read-arcs and output-arcs of $t$, respectively.
$\costfun:\places \cup \transitions \to \nat$ 
is the cost function assigning
firing costs to transitions and storage costs to places.
Note that it is not a restriction to use integers for time bounds and costs
in PTPN. By the same standard technique as in timed automata, 
the problem for rational numbers can be reduced to the integer case
(by multiplying all numbers with the least common multiple of the divisors).
To simplify the presentation we use a one-dimensional cost.
This can be generalized to multidimensional costs; see
Section~\ref{sec:conclusion}.
We let $\maxval$ denote the maximum integer appearing on the arcs of a given
PTPN.
A {\it configuration} of $\ptpn$ is a tuple 
$\tuple{q,\marking}$ where $q \in\states$ is a control-state and
$\marking$ is a {\it marking} of $\ptpn$.
A marking is
a multiset over $\places\times\nnreals$, i.e., 
$\marking \in \msetsover{(\places\times\nnreals)}$.
The marking $\marking$ defines the numbers  and ages
of tokens in each place in the net.
We identify a token in a marking
$\marking$ by the pair $(\place,x)$ representing its place and age
in $\marking$.
Then, $\marking(\place,x)$ 
defines the number of tokens with age $x$ in place $\place$.
Abusing notation, we define, for each place $\place$, a multiset $\marking(\place)$ 
over $\nnreals$,
where $\marking(\place)(x)=\marking(p,x)$.

For a marking $\marking$ of the form
$\lst{(\place_1,x_1)\;,\;\ldots\;,\;(\place_n,x_n)}$
 and $x\in\preals$, we use $\marking^{+x}$
to denote the marking
$\lst{(\place_1,x_1+x)\;,\;\ldots\;,\;(\place_n,x_n+x)}$.

\subsection{Computations}\label{subsec:computations} 
We define two transition relations on the 
set of configurations: timed transition
and discrete transition.
 A {\em timed transition} increases the age of each token
by the same real number.
Formally, for $x\in\preals, q \in \states$, we have 
$(q,\marking_1)\xtimedmovesto{x}(q,\marking_2)$ 
if $\marking_2=\marking_1^{+x}$.
We use $(q,\marking_1)\timedmovesto{}(q,\marking_2)$ to denote that
$(q,\marking_1)\xtimedmovesto{x}(q,\marking_2)$ for some $x\in\preals$.

We define the set of {\em discrete transitions $\discmovesto$} as
$\bigcup_{\transition\in\transitions}\ttrans$, where
$\ttrans$ represents the effect of {\it firing} the discrete 
transition $\transition$.
To define $\ttrans$ formally, we need the auxiliary predicate $\match$
that relates markings with the inputs/reads/outputs of transitions.
Let $\marking \in \msetsover{(\places\times\nnreals)}$ and
$\alpha \in \msetsover{(\places \times \Intervals)}$.
Then $\match(\marking,\alpha)$ holds iff there exists a bijection 
$f: \marking \mapsto \alpha$ s.t. for every $(p,x) \in \marking$ we have 
$f((p,x))=(p',\interval)$ with $p'=p$ and $x \in \interval$.
Let $\transition = (q_1,q_2,\inputs,\reads,\outputs) \in \transitions$.
Then we have a discrete transition
$(q_1,\marking_1) \ttrans (q_2,\marking_2)$ iff 
there exist $I,O,R,\marking_1^{\it rest} \in \msetsover{(\places\times\nnreals)}$ s.t. the 
following conditions are satisfied:
\begin{iteMize}{$\bullet$}
\item
$\marking_1 = I + R + \marking_1^{\it rest}$
\item
$\match(I, \inputs)$, $\match(R, \reads)$ and 
$\match(O, \outputs)$.
\item
$\marking_2 = O + R + \marking_1^{\it rest}$
\end{iteMize}
We say that $\transition$ is  {\it enabled} in $(q_1,\marking_1)$ if 
the first two conditions are satisfied.
A transition $\transition$ may be fired iff for each
input-arc and each read-arc, there is a token with the right age
in the corresponding input place.
These tokens in $I$ matched to the input arcs will be removed when the transition is
fired, while the tokens in $R$ matched to the read-arcs are kept.
The newly produced tokens in $O$ have ages which are chosen nondeterministically from the
relevant intervals on the output arcs of the transitions.
This semantics is lazy, i.e., enabled transitions do not need to
fire and can be disabled again.

We write $\movesto{} \, =\, \timedmovesto\,\cup\,\discmovesto$ to denote all 
transitions.
For sets $\confs,\confs'$ of configurations,
we write $\confs\movesto{*}\confs'$ to denote that
$\conf\movesto{*}\conf'$ for some $\conf\in\confs$ and $\conf'\in\confs'$.
%
A {\it computation} $\comp$ (from $\conf$ to $\conf'$) is a sequence of transitions
$\conf_0\movesto{} \conf_1 \movesto{} \dots \movesto{} \conf_n$ 
such that $\conf_0=\conf$ and $\conf_n=\conf'$.
We write $\conf \compto{\comp} \conf'$ to denote that
$\comp$ is a computation from $\conf$ to $\conf'$.
Similarly, we write $\confs \compto{\comp} \confs'$ to denote that
$\exists \conf_1 \in \confs,\conf_n \in \confs'.\, \conf_1 \compto{\comp}
\conf_n$.
%
\subsection{Costs}
%
The cost of a computation consisting of
one discrete transition $\transition \in \transitions$ 
is defined as $\costof{(q_1,\marking_1) \ttrans (q_2,\marking_2)} := \costof{t}$.
The cost of a computation consisting of
one timed transition is defined by
$\costof{(q,\marking) \movesto{x} (q,\marking^{+x})} := 
x*\sum_{\place \in \places} |\marking(\place)|*\costof{\place}$.
The cost of a computation is the sum of all transition costs in it, i.e., 
\[
\costof{(q_1,\marking_1) \movesto{} (q_2,\marking_2) \movesto{} \dots \movesto{} (q_n,\marking_n)}:=
\sum_{1\leq i<n}\costof{(q_i,\marking_i) \movesto{} (q_{i+1},\marking_{i+1})}
\]
We write $\confs\movesto{\costnum} \confs'$ to denote that there is
a computation $\comp$ such that  $\confs\movesto{\comp} \confs'$ 
and $\costof{\comp}\leq\costnum$.
We define 
$\optcostof{\confs,\confs'}$ to be the infimum 
of the set $\setcomp{\costnum}{\confs\movesto{\costnum} \confs'}$,
i.e., the infimum of the costs of all computations leading
from $\confs$ to $\confs'$. 
We use the infimum, because the minimum does not exist in general.
%
We partition the set of places $\places = \costplaces \cup \freeplaces$
where $\costof{p} > 0$ for $p \in \costplaces$ and
$\costof{p} = 0$ for $p \in \freeplaces$.
The places in $\costplaces$ are called cost-places and the places in
$\freeplaces$ are called free-places.

\subsection{Relation of PTPN to Other Models}
PTPN subsume the priced timed automata of 
\cite{AlurTP01:WTA,Kim:etal:priced,BouyerBBR:PTA}
via the following simple encoding.
For every one of the finitely many clocks of the automaton
we have one place in the PTPN with exactly one token on it whose age 
encodes the clock value. We assign cost zero to these places.
For every control-state $s$ of the automaton we have one place $p_s$ in the
PTPN. Place $p_s$ contains exactly one token iff the automaton is in state
$s$, and it is empty otherwise. An automaton transition from state $s$ to
state $s'$ is encoded by a PTPN transition consuming the token from $p_s$
and creating a token on $p_{s'}$. The transition guards referring to clocks
are encoded as read-arcs to the places which encode clocks, labeled with the
required time intervals. Note that open and half-open time intervals
are needed to encode the strict inequalities used in timed automata.
Clock resets are encoded by consuming the timed token (by an input-arc)
and replacing it (by an output-arc) with a new token on the same place with age $0$.
The cost of staying in state $s$ is encoded by assigning a cost to place
$p_s$, and the cost of performing a transition is encoded as the cost of the
corresponding PTPN transition.
Also PTPN subsume fully general unbounded (i.e.,
infinite-state) Petri nets (by setting all time intervals to $[0:\infty)$ and 
thus ignoring the clock values).

Note that (just like for timed automata) 
the problems for continuous-time PTPN cannot be reduced to 
(or approximated by) the discrete-time case. 
Replacing strict inequalities with non-strict ones
might make the final control-state reachable, when it originally was
unreachable
(i.e., it would cause qualitative differences, and not only quantitative ones).

\subsection{The Priced Coverability Problem}
We will consider two variants of the cost problem,
the {\it Cost-Threshold} problem and the 
{\it Cost-Optimality} problem.
They are both characterized by an \emph{initial}
control state $\initstate$ and a \emph{final} control state $\finalstate$.

Let $\initconf = \tuple{\initstate,\lst{}}$ be the initial configuration
and $\finalconfs = \{\tuple{\finalstate,\marking}\,|\, 
\marking \in \msetsover{(\places\times\nnreals)}\}$
the set of final configurations defined by the control-state $\finalstate$.
I.e., we start from a configuration where the control state
is $\initstate$ and where all the places are empty, and then consider
the cost of computations that takes us to
$\finalstate$. (If $\initconf$ contained tokens with a non-integer age
then the optimal cost might not be an integer.)

In the {\it Cost-Threshold} problem we ask the question whether
$\optcostof{\initconf, \finalconfs} \le \threshold$
for a given threshold $\threshold \in \nat$.


In the {\it Cost-Optimality} problem, we want to compute
$\optcostof{\initconf, \finalconfs}$.

\subsection{A Running Example}\label{subsec:example1}

\begin{figure}[htbp]
\begin{center}
\begin{tikzpicture}[show background rectangle,label distance=-1.2mm]

\node[transition,label=above:{\tiny $\transition_1$},label=right:{\tiny $1$}](t1){}; %
\node[place,,minimum size=10mm,below  = 9mm of t1,label=right:{\tiny $\place_2$},label=left:{\tiny $2$}] (p2){}
[children are tokens, token distance=5mm]
child {node [token] {$6.5$}};

\node[place,draw=blue,fill=blue,left  = 7mm of p2,label=left:{\tiny $\state_1$}] (q1){};
\node[place,minimum size=10mm,left  = 7mm of q1,label=left:{\tiny $\place_1$},label=above:{\tiny $3$}] (p1){}
[children are tokens, token distance=5mm]
child {node [token] {$3.1$}}
child {node [token] {$3.1$}}
child {node [token] {$2.5$}};

\node[place,draw=blue,fill=blue,right  = 7mm of p2,label=right:{\tiny $\state_2$}] (q2){};
\node[place,minimum size=10mm,right  = 7mm of q2,label=right:{\tiny $\place_3$},label=below:{\tiny $0$}] (p3){}
[children are tokens, token distance=5mm]
child {node [token] {$0.1$}}
child {node [token] {$0.1$}};
;

\node[transition,below  =  9mm of p2,label=below:{\tiny $\transition_2$},label=left:{\tiny $3$}](t2){};

\draw[->,thick] (t1) -- (p2);
\draw[->,thick] (q1) -- (t1);
\draw[->,thick] (t1) -- (q2);
\draw[->,thick] (p1) -- node[sloped,above=-1mm] {\tiny $(0,3]$} (t1);
\draw[->,thick] (t1) --  node[sloped,above=-1mm] {\tiny $[1,5)$} (p2);
\draw[->,thick] (t1) -- node[sloped,above=-1mm] {\tiny $(2,\infty)$} (p3);
\draw[->,thick] (t2) -- node[sloped,below=-1mm] {\tiny $[0,\infty)$}(p1);
\draw[<->,thick] (t2) -- node[sloped,below=-1mm] {\tiny $[2,2]$}(p2);
\draw[->,thick] (p3) -- node[sloped,below=-1mm] {\tiny $[1,4)$} (t2);
\draw[->,thick] (q2) -- (t2);
\draw[->,thick] (t2) -- (q1);

\end{tikzpicture}
\end{center}
\caption{A simple example of a PTPN.
}
\label{ptpn:fig}
\end{figure}

Figure~\ref{ptpn:fig} shows a simple PTPN.
We will use this PTPN to give examples of some of the concepts that we introduce
in this paper.

\paragraph{\bf Places and Transitions}
The PTPN has
two control states ($\state_1$ and $\state_2$) depicted as dark-colored circles,
three places ($\place_1$, $\place_2$, $\place_3$) depicted as light-colored circles,
and two transitions ($\transition_1$ and $\transition_2$) depicted as rectangles.
Source/target control states, input/output places are indicated by arrows
to the relevant transition.
Read places are indicated by double headed arrows.
The source and target control states of $\transition_1$ are $\state_1$, resp.\ $\state_2$.
The input, read resp.\ output arcs of $\transition_1$ are given by the multisets
$\lst{\tuple{\place_1,(0,3]}}$,  $\lst{}$
resp.\ $\lst{\tuple{\place_2,[1,5)},\tuple{\place_3,(2,\infty)}}$.
In a similar manner, $\transition_2$  is defined by the tuple
$\tuple{
\state_2,\state_1,
\lst{\tuple{\place_3,[1,4)}},
\lst{\tuple{\place_2,[2,2]}},\lst{\tuple{\place_1,[0,\infty)}}}$.
The prices of $\transition_1,\transition_2,\place_1,\place_2,\place_3$
are $1,3,3,2,0$ respectively.
The value of $\maxval$ is $5$.

\paragraph{\bf Markings}
Figure~\ref{ptpn:fig} shows a marking 
$\lst{\tuple{\place_1,3.1}^2,\tuple{\place_1,2.5},\tuple{\place_2,6.5},\tuple{\place_3,0.1}^2}$.

\paragraph{\bf Computations and Prices}
An example of a computation $\comp$ is:
\[
\begin{array}{c}
\tuple{\state_1,\lst{\tuple{\place_1,3.1}^2,\tuple{\place_1,2.5},\tuple{\place_2,6.5},\tuple{\place_3,0.1}^2}}
\\\longrightarrow_{\transition_1}\\
\tuple{\state_2,\lst{\tuple{\place_1,3.1}^2,\tuple{\place_2,6.5},\tuple{\place_2,1.3},
\tuple{\place_3,0.1}^2,\tuple{\place_3,2.2}}}
\\\xtimedmovesto{0.7}\\
\tuple{\state_2,\lst{\tuple{\place_1,3.8}^2,\tuple{\place_2,7.2},\tuple{\place_2,2.0},
\tuple{\place_3,0.8}^2,\tuple{\place_3,2.9}}}
\\\longrightarrow_{\transition_2}\\
\tuple{\state_1,\lst{\tuple{\place_1,3.8}^2,\tuple{\place_1,9.2},\tuple{\place_2,7.2},\tuple{\place_2,2.0},\tuple{\place_3,0.8}^2}}
\\\xtimedmovesto{1.3}\\
\tuple{\state_1,\lst{\tuple{\place_1,5.1}^2,\tuple{\place_1,10.5},\tuple{\place_2,8.5},\tuple{\place_2,3.3},\tuple{\place_3,2.1}^2}}
\end{array}
\]

The cost $\costof\comp$ is given by
\[
\begin{array}{l}
1+ 2*3*0.7+2*2*0.7+3*0*0.7+\\
3+ 3*3*1.3+2*2*1.3+1*0*1.3=27.9
\end{array}
\]

The transition $\transition_2$ is not enabled from any of the following configurations:
\begin{iteMize}{$\bullet$}
\item
The marking
$\tuple{\state_1,\lst{\tuple{\place_1,3.8},\tuple{\place_2,2.0},
\tuple{\place_3,2.9}}}$
since it does not have the correct control state.
\item
The marking
$\tuple{\state_2,\lst{\tuple{\place_1,3.1}^2,\tuple{\place_2,2.0},\tuple{\place_3,0.1}^2}}$ 
since it is missing input tokens with the correct ages in $\place_3$.
\item
The marking
$\tuple{\state_2,\lst{\tuple{\place_1,3.1}^2,\tuple{\place_2,1.0},\tuple{\place_3,1.1}^2}}$ 
since it is missing read tokens with the correct ages in $\place_2$.
\end{iteMize}

\section{Computations in $\delta$-Form}
\label{delta:section}
\noindent
We solve the Cost-Threshold and Cost-Optimality problems for PTPN via a 
reduction that goes through a series of abstraction steps.
First we show that, in order to solve the cost problems, it is sufficient
to consider computations of a certain form where the ages of all the tokens are
arbitrarily close to an integer.
This technique is similar to the corner-point abstraction used in the analysis
of priced timed automata
\cite{Behrman:HSCC2001,Bouyer:CACM2011,BouyerBBR:PTA} where sets of possible 
clock valuations are described by polyhedra, and where the infimum cost is
achieved in the corner-points of these polyhedra.
For PTPN, extra difficulties arise from the unbounded growth
of configurations by newly created clocks in new tokens and the potential
disappearance of old clocks in consumed tokens.
Moreover, newly created clocks in PTPN do not necessarily have value zero.
This leads to a more complex structure of the matrices that describe the
polyhedra of clock valuations than in the case of priced timed automata; see
Def.~\ref{def:ptpn-matrix}.
However, we show in Lemma~\ref{lem:PTPN_totally_unimodular} that these more
general matrices are still totally unimodular, which makes the corner-point
abstraction possible.

We decompose PTPN markings $\marking$ into submarkings such that in every submarking the 
fractional parts (but not necessarily the integer parts) of the token ages are
the same. We then arrange these submarkings in a sequence  
$
\marking_{-m}, \dots , \marking_{-1}, \marking_0,
\marking_1,\dots,\marking_n
$
such that $\marking_{-m}, \dots , \marking_{-1}$ contain tokens with fractional 
parts $\ge 1/2$ in increasing order,
$\marking_0$ contains the tokens with fractional part zero,
and $\marking_1,\dots,\marking_n$ contain 
tokens with fractional 
parts $< 1/2$ in increasing order.

For example, the marking
$\lst{\tuple{\place_1,3.1}^2,\tuple{\place_1,2.5},\tuple{\place_2,6.5},\tuple{\place_3,0.1}^2}$
of Figure~\ref{ptpn:fig} is decomposed as
$M_{-1} = \lst{\tuple{\place_1,2.5},\tuple{\place_2,6.5}}$,
$M_0 = \lst{}$
and
$M_1 = \lst{\tuple{\place_1,3.1}^2, \tuple{\place_3,0.1}^2}$.

Formally, the decomposition of a
PTPN marking $\marking$ into its fractional parts
\[
\marking_{-m}, \dots , \marking_{-1}, \marking_0,
\marking_1,\dots,\marking_n
\] 
is uniquely defined by the following properties:
\begin{iteMize}{$\bullet$}
\item
$\marking=\marking_{-m}+ \dots + \marking_{-1}+ \marking_0+
\marking_1+\dots+\marking_n$.
\item
If $\tuple{\place,x}\in\marking_i$ and $i<0$ then $\fractof{x} \ge 1/2$.
If $\tuple{\place,x}\in\marking_0$ then $\fractof{x}=0$.
If $\tuple{\place,x}\in\marking_i$ and $i>0$ then $\fractof{x} < 1/2$.
\item
Let $\tuple{\place_i,x_i}\in\marking_i$ and $\tuple{\place_j,x_j}\in\marking_j$.
Then $\fractof{x_i} = \fractof{x_j}$ iff $i=j$,
and if $-m\leq i < j < 0$ or $0 \le i < j \le n$ then
$\fractof{x_i} < \fractof{x_j}$.
\item
$\marking_i\neq\emptyset$ if $i\neq 0$ 
($\marking_0$ can be empty, but the other $\marking_i$ must be non-empty in order to get a
unique representation.)
\end{iteMize}
We say that a timed transition $(q,\marking)\movesto{x}(q,\marking')$
is {\em detailed} iff at most one fractional part of any token in $\marking$
changes its status about reaching or exceeding the next higher integer value.
Formally, let $\epsilon$ be the fractional part of the token ages in
$\marking_{-1}$, or $\epsilon = 1/2$ if $M_{-1}$ does not exist.
Then $(q,\marking)\movesto{x}(q,\marking')$
is {\em detailed} iff either $0 < x < 1-\epsilon$ (i.e., no tokens reach the
next integer),
or $\marking_0 = \emptyset$ and $x = \epsilon$ (no tokens had integer age,
but those in $M_{-1}$ reach integer age).
Every computation of a PTPN can be transformed into an equivalent one 
(w.r.t. reachability and cost) where
all timed transitions are detailed, by replacing some long timed transitions with
several detailed shorter ones where necessary. 
Thus we may assume w.l.o.g.
that timed transitions are detailed.
This property is needed to obtain a one-to-one correspondence between PTPN
steps and the steps of A-PTPN, defined in the next section.

For $\delta\in(0:1/5]$ the marking
$\lst{\tuple{\place_1,x_1},\ldots,\tuple{\place_n,x_n}}$ 
is in {\it $\delta$-form} if, for all $i:1\leq i\leq n$, 
it is the case that
either (i) $\fract(x_i) < \delta$ (low fractional part), 
or (ii) $\fract(x_i) > 1-\delta$ (high fractional part).
I.e., the age of each token is close to (within $< \delta$) 
an integer.
We choose $\delta \le 1/5$ to ensure that the cases 
(i) and (ii) do not overlap, and that they still do not overlap for a new 
$\delta' \le 2/5$ after a delay of $\le 1/5$ time units.

The occurrence of a discrete transition $t$ is said to be in $\delta$-form if 
its output $O$ is in $\delta$-form, i.e., 
the ages of the newly generated tokens are  
close to an integer. This is not a property
of the transition $t$ as such, but a property of its occurrence,
because it depends on the particular choice of $O$ (which is not fixed
but possibly nondeterministic within certain constraints; 
see Subsection~\ref{subsec:computations}).

Let $\ptpn=\ptpntuple$ be a PTPN and
$\initconf = \tuple{\initstate,\lst{}}$
and $\finalconfs = \{\tuple{\finalstate,\marking}\,|\, 
\marking \in \msetsover{(\places\times\nnreals)}\}$ 
as in the last section.

For $0 < \delta \le 1/5$, the computation $\comp$ is in $\delta$-form
iff 
\begin{enumerate}
\item 
Every occurrence of a discrete transition 
$\conf_i \ttrans \conf_{i+1}$ is in $\delta$-form, and 
\item
For every timed transition $\conf_i \movesto{x} \conf_{i+1}$ we have
either $x \in (0:\delta)$ or $x \in (1-\delta:1)$.
\end{enumerate}
We show that, in order to find the infimum of the possible costs,
it suffices to consider computations in $\delta$-form, for arbitrarily small
values of $\delta > 0$.

\begin{lem}\label{lem:deltaform}
Let $\initconf \compto{\comp} \finalconfs$,
where $\comp$ is $\initconf = \conf_0\movesto{} 
\dots \movesto{} \conf_{\it length} \in \finalconfs$.
Then for every $\delta > 0$ there exists a computation
$\comp'$ in $\delta$-form where
$\initconf \compto{\comp'} \finalconfs$,
where $\comp'$ is $\initconf = \conf'_0\movesto{}
\dots \movesto{} \conf'_{\it length} \in \finalconfs$ s.t.
$\costof{\comp'} \le \costof{\comp}$, $\comp$ and $\comp'$ have the same
length and $\forall i:0 \le i \le {\it length}.\, |\conf_i| = |\conf'_i|$.
Furthermore, if $\comp$ is detailed then $\comp'$ is detailed.
\end{lem}

\begin{proof}
{\bf Outline of the proof.}
We construct $\comp'$ by fixing the structure of the computation $\comp$ and
varying the finitely many real numbers describing the delays of timed transitions
and the ages of newly created tokens. The tuples of numbers corresponding to a
possible computation are contained in a polyhedron, which is described by 
inequations via a totally
unimodular matrix, and whose vertices thus have integer coordinates.
Since the cost function is linear in these numbers, the infimum of the costs
can be approximated arbitrarily closely by computations $\comp'$ whose numbers are
arbitrarily close to integers, i.e., computations $\comp'$ in $\delta$-form 
for arbitrarily small $\delta > 0$.\\
{\bf Detailed proof.}
The computation $\comp$ with $\initconf \compto{\comp} \finalconfs$
consists of a sequence of discrete transitions and timed transitions.
Let $n$ be the number of timed transitions in $\comp$ and $x_i >0$
(for $1\le i \le n$) be the delay of the $i$-th timed transition in $\comp$.
Let $m$ be the number of newly created tokens in $\comp$. We fix some
arbitrary order on these tokens (it does not need to agree with the order of token
creation) and call them $t_1,\dots,t_m$.
Let $y_i$ be the age of token $t_i$ when it is created in $\comp$.
(Recall that the age of new tokens is not always zero, but chosen
nondeterministically out of given intervals.)

We now consider the set of all computations $\comp'$ that have the same
structure, i.e., the same transitions, as $\comp$, but with modified
values of $y_1,\dots,y_m$ and $x_1,\dots,x_n$.
Such computations $\comp'$ have the same length as $\comp$
and the sizes of the visited configurations match.
Also if $\comp$ is detailed then $\comp'$ is detailed.

It remains to show that one such computation $\comp'$ is
in $\delta$-form and $\costof{\comp'} \le \costof{\comp}$.

The set of tuples $(y_1,\dots,y_m,x_1,\dots,x_n)$ for which such a computation
$\comp'$ is feasible is described by a set of inequations that depend on %
the transition guards. (The initial configuration, and the set of final 
configurations do not introduce any constraints on
$(y_1,\dots,y_m,x_1,\dots,x_n)$,
because they are closed under changes to token ages.)
The inequations are derived from the following conditions.
\begin{iteMize}{$\bullet$}
\item
The time always advances, i.e., $x_i >0$.
\item
When the token $t_j$ is created by an output arc with interval
$[a:b]$ we have $a \le y_j \le b$, and similarly with strict
inequalities if the interval is (half) open.
Note that the bounds $a$ and $b$ are integers (except where 
$b=\infty$ in which case there is no upper bound constraint).
\item
Consider a token $t_j$ that is an input of some discrete transition $t$
via an input arc or a read arc labeled with interval $[a:b]$.
Note that the bounds $a$ and $b$ are integers (or $\infty$).
Let $x_k, x_{k+1}, \dots, x_{k+l}$ be the delays of the timed transitions
that happened between the creation of token $t_j$ and the transition $t$.
Then we must have $a \le y_j + x_k + x_{k+1} + \dots + x_{k+l} \le b$.
(Similarly with strict inequalities if the interval is (half) open.)
\end{iteMize} 
These inequations describe a polyhedron ${\it PH}$ which contains all 
feasible tuples of values $(y_1,\dots,y_m,x_1,\dots,x_n)$. 
By the precondition of this lemma, there exists a
computation $\initconf \compto{\comp} \finalconfs$
and thus the polyhedron ${\it PH}$ is nonempty.
Therefore we obtain the closure of the polyhedron $\overline{{\it PH}}$
by replacing all strict 
inequalities $<, >$ with normal inequalities $\le, \ge$.
Thus $\overline{{\it PH}}$ contains ${\it PH}$,
but every point in $\overline{{\it PH}}$ is arbitrarily close to a point in ${\it PH}$.
Now we show that the vertices of the 
polyhedron $\overline{{\it PH}}$ have integer
coordinates.

Let $v = (y_1,\dots,y_m,x_1,\dots,x_n)$ be a column vector of the free
variables. Then the polyhedron $\overline{{\it PH}}$ can be described by the 
inequation $M \cdot v \le c$, where $c$ is a column vector of integers and
$M$ is an integer matrix. Now we analyze the shape of the matrix $M$.
Each inequation corresponds to a row in $M$. If the inequality
is $\le$ then the elements are in $\{0,1\}$, and if the 
inequality is $\ge$ then the elements are in $\{0,-1\}$.
Each of the inequations above refers to at most one variable $y_j$,
and possibly one continuous block of several variables $x_k,x_{k+1},\dots,x_{k+l}$.
Moreover, for each $y_j$, this block (if it is nonempty)
starts with the same variable $x_k$. This is because the $x_k,x_{k+1},\dots,x_{k+l}$
describe the delays of the timed transitions between the creation
of token $t_j$ and the moment where $t_j$ is used. $x_k$ is always the first
delay after the creation of $t_j$, and no delays can be left
out. Note that the token $t_j$ can be used more than once, because transitions with
read arcs do not consume the token. 
We present the inequalities in blocks, where the first block contains all
which refer to $y_1$, the second block contains all which 
refer to $y_2$, etc. The last block contains those inequations 
that do not refer to any $y_j$, but only to variables $x_i$. 
Inside each block we sort the inequalities 
w.r.t. increasing length of the $x_k,x_{k+1},\dots,x_{k+l}$ block, i.e., 
from smaller values of $l$ to larger ones. (For $y_j$ we have the same $k$.)
Thus the matrix $M$ has the following form:
\[
\left(
\begin{array}{rrrr|rrrrrr}
1 & 0 & 0 & 0 & 0 & 0 & 0 & 0 & 0 & 0 \\
1 & 0 & 0 & 0 & 0 & 1 & 1 & 1 & 0 & 0 \\
1 & 0 & 0 & 0 & 0 & 1 & 1 & 1 & 1 & 0 \\
\dots\\
0 & 1 & 0 & 0 & 0 & 0 & 0 & 0 & 0 & 0 \\
0 & -1 & 0 & 0 & -1 & -1 & 0 & 0 & 0 & 0 \\
0 & 1 & 0 & 0 & 1 & 1 & 1 & 1 & 0 & 0 \\
\dots
\end{array}
\right)
\]
Formally, the shape of these matrices is defined as follows.

\begin{defi}\label{def:ptpn-matrix}
We call a $(z \times m+n)$-matrix a {\em PTPN constraint matrix},
if every row has one of the following two forms.
Let $j \in \{1,\dots,m\}$ and $k(j) \in \{1,\dots,n\}$ be a number that
depends only on $j$, and let $\alpha \in \{-1,1\}$.
First form: $0^{j-1} \alpha 0^{m-j} 0^{k(j)-1} \alpha^* 0^*$.
Second form: $0^* \alpha^* 0^*$.
Matrices that contain only rows of the second form all called {\em 3-block matrices}
in \cite{BouyerBBR:PTA}.
\end{defi}

\begin{defi}\label{def:totally_unimodular} \cite{NW88} An integer matrix is called 
{\em totally unimodular} iff
the determinant of all its square submatrices is equal to $0$, $1$ or $-1$.
\end{defi}

\begin{lem}\label{lem:PTPN_totally_unimodular}
All PTPN constraint matrices are totally unimodular.
\end{lem}
\smallskip
\begin{proof}
First, every square submatrix of a PTPN constraint matrix has the same form
and is also a PTPN constraint matrix. Thus it suffices to show the 
property for square PTPN constraint matrices.
We show this by induction on the size.
The base case of size $1\times 1$
is trivial, because the single value must be in $\{-1,0,1\}$.
For the induction step consider a square $k\times k$ PTPN constraint matrix
$M$, with some $n,m$ s.t. $n+m = k$.
If $M$ does not contain any row of the first form
then $M$ is a 3-block matrix and thus totally unimodular by 
\cite{BouyerBBR:PTA} (Lemma 2). Otherwise, $M$ contains a row $i$ of the first
form where $M(i,j) \in \{-1,1\}$ for some $1 \le j \le m$. 
Without restriction let $i$ be such a row in $M$ where 
the number of nonzero entries is minimal.
Consider all rows $i'$ in $M$ where $M(i',j) \neq 0$.
Except for $M(i',j)$, they just contain (at most) one block
of elements $1$ (or $-1$) that starts at position $m+k(j)$. 
By adding/subtracting row $i$ to all these other rows $i'$ where
$M(i',j) \neq 0$ we obtain a new matrix $M'$ where $M'(i,j)$ 
is the only nonzero entry in column $j$ in $M'$ and
$\det(M') = \det(M)$. Moreover, $M'$ is also a PTPN constraint matrix,
because of the minimality of the nonzero block length in row $i$
and because all these blocks start at $m+k(j)$. I.e., in $M'$ these
modified rows $i'$ have the form $0^* 1^* 0^*$ or $0^* (-1)^* 0^*$.
We obtain $M''$ from $M'$ by deleting column $j$ and row $i$,
and $M''$ is a $(k-1)\times (k-1)$ PTPN constraint matrix (because $j \le m$).
By induction hypothesis, $M''$ is totally unimodular and $\det(M'') \in \{-1,0,1\}$.
By the cofactor method, 
$\det(M') = (-1)^{i+j} * M'(i,j) * \det(M'') \in \{-1,0,1\}$.
Thus $\det(M) = \det(M') \in \{-1,0,1\}$ and 
$M$ is totally unimodular.
\end{proof}

\begin{thm}\label{thm:NW88} \cite{NW88}.
Consider the polyhedron $\{v \in \R^k\ |\ M \cdot v \le c\}$ 
with $M$ a totally unimodular $(p \times k)$ matrix
and $c \in \Z^p$. Then the coordinates of its vertices are integers.
\end{thm}

Since our polyhedron $\overline{{\it PH}}$ is described by a PTPN constraint matrix,
which is totally unimodular by Lemma~\ref{lem:PTPN_totally_unimodular}, 
it follows from Theorem~\ref{thm:NW88} 
that the vertices of
$\overline{{\it PH}}$ have integer coordinates.

Since the ${\it Cost}$ function is linear in $x_1,\dots,x_n$
(and does not depend on $y_1,\dots,y_m$),
the infimum of the costs on ${\it PH}$ is obtained at a vertex
of $\overline{{\it PH}}$, which has integer coordinates by Theorem~\ref{thm:NW88}.
Therefore, one can get arbitrarily close to the infimum cost
with values $y_1,\dots,y_m,x_1,\dots,x_n$ which are arbitrarily 
close to some integers.
Thus, for every computation $\initconf \compto{\comp} \finalconfs$
there exists a modified computation $\comp'$
with values $y_1,\dots,y_m,x_1,\dots,x_n$ arbitrarily 
close to integers (i.e., $\comp'$ in $\delta$-form for arbitrarily small $\delta >0$)
such that $\initconf \compto{\comp'} \finalconfs$
and $\costof{\comp'} \le \costof{\comp}$.
(Note that the final configuration reached by $\comp'$ possibly differs from the
final configuration of $\comp$ in the ages of some tokens. However, this does
not matter, because the set of configurations $\finalconfs$ is closed
under such changes.)
\end{proof}

The following corollary follows directly from Lemma~\ref{lem:deltaform}
and shows that in order to find the infimum of the cost it suffices to 
consider only computations in $\delta$-form for arbitrarily small $\delta >0$.

\begin{cor}
For every $\delta>0$ we have
\[
\optcostof{\initconf,\finalconfs} =
\inf\{\costof{\comp}\, |\, \initconf \compto{\comp} \finalconfs,
\mbox{$\comp$ in $\delta$-form}\}
\]
\end{cor}

\section{Abstract PTPN}\label{sec:abstract_ptpn}
\noindent
We now reduce the Cost-Optimality problem to a simpler case without
explicit clocks by defining a new class of systems called 
{\em abstract PTPN} (for short A-PTPN), whose computations
represent PTPN computations in $\delta$-form, for infinitesimally small
values of $\delta >0$.
For each PTPN $\ptpn=\ptpntuple$, we define a corresponding
A-PTPN $\ptpn'$, sometimes denoted by $\aptpnof{\ptpn}$.
The A-PTPN $\ptpn'$ is syntactically of the same form $\ptpntuple$ as $\ptpn$.
However, $\ptpn'$ induces a different transition system, because
its configurations and operational semantics are different.
Below we define the set of markings of the A-PTPN, and then describe the transition
relation.
We will also explain the relation to the markings and the transition relation
induced by the original PTPN.
\paragraph{\bf Markings and Configurations}
Fix a $\delta:0<\delta\leq 2/5$.
A marking $\marking$ of $\ptpn$ in $\delta$-form is encoded by a marking
$\aptpnof\marking$ of $\ptpn'$ which is
described by a triple 
$\tuple{\highword,\zmset,\lowword}$ where 
$\highword,\lowword\in\wordsover{\left(\msetsover{(\places\times\zeroto{\maxval +1})}\right)}$ and
$\zmset\in \msetsover{(\places\times\zeroto{\maxval +1})}$.
The ages of the tokens in $\aptpnof\marking$ are integers and
therefore only carry the integral parts
of the tokens in the original PTPN.
However, the marking $\aptpnof\marking$ carries additional information about the
fractional parts of the tokens as follows.
The tokens in $\highword$ represent tokens
in $\marking$ that have \emph{high} fractional parts 
(their values are at most $\delta$ below
the next integer);
the tokens in $\lowword$ represent tokens
in $\marking$ that have \emph{low} fractional parts 
(their values at most $\delta$ above the previous integer); while
tokens in $\zmset$ represent tokens
in $\marking$ that have \emph{zero} fractional parts (their values are equal to
an integer).
Furthermore, the ordering among the fractional parts of tokens
in $\highword$ (resp.\ $\lowword$) is represented by the positions of the
multisets to which they belong in  $\highword$ (resp.\ $\lowword$).
Let $\marking =\marking_{-m}, \dots , \marking_{-1}, \marking_0,
\marking_1,\dots,\marking_n$ be the decomposition of $\marking$ into
fractional parts, as defined in Section~\ref{delta:section}.
Then we define
$\aptpnof\marking:= \tuple{\highword, \zmset, \lowword}$
with $\highword = b_{-m} \dots b_{-1}$, 
and 
$\lowword = b_1 \dots b_n$,
where
$b_i((p,\lfloor x\rfloor)) = \marking_i((p,x))$ if $x \le \maxval$.
(This is well defined, because $\marking_i$
contains only tokens with one particular fractional part.)
Furthermore,  $b_0((p,\maxval +1)) = \sum_{y > \maxval} \marking((p,y))$, i.e.,
all tokens whose age is $> \maxval$ are abstracted as tokens of age $\maxval +1$,
because the PTPN cannot distinguish between token ages $> \maxval$.
Note that $\highword$ and $\lowword$ represent tokens with fractional 
parts in increasing order.
An A-PTPN configuration is a control-state plus a marking.
If we apply $\aptpn$ to a set of configurations (i.e., $\aptpn(\finalconfs)$),
we implicitly restrict this set to the subset of configurations in $2/5$-form. 

\paragraph{\bf Transition Relation}

The transitions on the A-PTPN are defined as follows.
For every discrete transition 
$\transition = (q_1,q_2,\inputs,\reads,\outputs) \in \transitions$
we have  
$(q_1, b_{-m} \dots b_{-1}, b_0, b_1 \dots b_n) \ttrans
(q_2, c_{-m'} \dots c_{-1}, c_0, c_1 \dots c_{n'})$
if the following conditions are satisfied:
For every $i:-m \le i \le n$ 
there exist $b_{i}^I, b_{i}^R, b_{i}^{\it rest}, \hat{O}, b_0^O \in
\msetsover{(\places\times\zeroto{\maxval +1})}$
s.t. for every $0 < \epsilon < 1$ we have
\begin{iteMize}{$\bullet$}
\item
$b_{i} = b_{i}^I + b_{i}^R + b_{i}^{\it rest}$ for $-m \le i \le n$
\item
$\match((\sum_{i\neq 0} b_i^I)^{+\epsilon} + b_0^I, \inputs)$
\item
$\match((\sum_{i\neq 0} b_i^R)^{+\epsilon} + b_0^R, \reads)$
\item
$\match(\hat{O}^{+\epsilon} + b_0^O, \outputs)$ 
\item
There is a strictly monotone injection $f: \{-m,\dots,n\} \mapsto \{-m',\dots,n'\}$
where $f(0)=0$ 
s.t. $c_{f(i)} \ge b_i - b_i^I$
and $c_0 = b_0 - b_0^I + b_0^O$
and
$\sum_{i \neq 0} c_i = (\sum_{i\neq 0} b_{i} - b_{i}^I) + \hat{O}$.
\end{iteMize}
The intuition is that the A-PTPN tokens in $b_i$ for $i \neq 0$ represent
PTPN tokens with a little larger, and strictly positive,
fractional part. Thus their age is incremented by $\epsilon>0$
before it is matched to the input, read and output arcs.
The fractional parts of the tokens that are not involved
in the transition stay the same.
However, since all the time intervals in the PTPN have integer bounds,
the fractional parts of newly created tokens are totally arbitrary.
Thus they can be inserted at any position in the sequence, between any
positions in the sequence, or before/after the sequence
of existing fractional parts.
This is specified by the last condition on the
sequence $c_{-m'} \dots c_{-1}, c_0, c_1 \dots c_{n'}$.

The following lemma shows the connection between a PTPN
and its corresponding A-PTPN for discrete transition steps.

\begin{lem}\label{lem:PTPN_APTPN_disc}
Let $(q,M)$ be a PTPN configuration in $\delta$-form for some $\delta \le
1/5$.
There is an occurrence of a discrete transition in $\delta$-form
$(q,M) \ttrans (q',M')$ if and only if
$\aptpn((q,M)) \ttrans \aptpn((q',M'))$. 
\end{lem}
\begin{proof}
Let $\marking=\marking_{-m}+ \dots + \marking_{-1}+ \marking_0+
\marking_1+\dots+\marking_n$ be the unique decomposition of $\marking$ into
increasing fractional parts, and
$\aptpnof\marking:= \tuple{b_{-m} \dots b_{-1}, \zmset, b_1 \dots b_n}$,
as defined in Section~\ref{sec:abstract_ptpn}.
Let $\transition = (q,q',\inputs,\reads,\outputs)$.

Now we prove the first implication.
Provided that $(q,\marking) \ttrans (q',\marking')$
there exist $I,O,R,\marking^{\it rest} \in \msetsover{(\places\times\nnreals)}$ s.t. the 
following conditions are satisfied:
\begin{iteMize}{$\bullet$}
\item
$\marking = I + R + \marking^{\it rest}$
\item
$\match(I, \inputs)$, $\match(R, \reads)$ and 
$\match(O, \outputs)$.
\item
$\marking' = O + R + \marking^{\it rest}$.
\end{iteMize}
Thus each $\marking_i$ can be decomposed into parts 
$\marking_i = \marking_i^{I} + \marking_i^{R} + \marking_i^{\it rest}$,
where $I = \sum_i \marking_i^{I}$, 
$R =  \sum_i \marking_i^{R}$,
$\marking^{\it rest} = \sum_i \marking_i^{\it rest}$.
Let $b_i^I = \aptpnof{\marking_i^{I}}$,
$b_i^R = \aptpnof{\marking_i^{R}}$,
$b_i^{\it rest} = \aptpnof{\marking_i^{\it rest}}$.
Then $b_{i} = b_{i}^I + b_{i}^R + b_{i}^{\it rest}$.
Since the time intervals on transitions have integer bounds,
we obtain $\match((\sum_{i\neq 0} b_i^I)^{+\epsilon} + b_0^I, \inputs)$
and
$\match((\sum_{i\neq 0} b_i^R)^{+\epsilon} + b_0^R, \reads)$.

Similarly as $\marking$, the marking $O$ can be uniquely 
decomposed into parts with increasing fractional part of the ages of tokens,
i.e., $O = O_{-j}+ \dots + O_{-1}+ O_0+
O_1+\dots+O_k$.
Let $\hat{O} = \aptpnof{O-O_0}$ and $b_0^O = \aptpnof{O_0}$.
Thus we get $\match(\hat{O}^{+\epsilon} + b_0^O, \outputs)$.

Since $\marking' = O + R + \marking^{\it rest}$, the sequence of the
remaining parts of the $\marking_i$ is merged with the 
sequence $O_{-j}+ \dots + O_{-1}+ O_0+ O_1+\dots+O_k$.
Thus $\marking'$ can be uniquely decomposed into parts with 
increasing fractional part of the ages of tokens, i.e.,
$\marking' = \marking'_{-m'}+ \dots + \marking'_{-1}+ \marking'_0+
\marking'_1+\dots+\marking'_{n'}$.
Let $c_i = \aptpnof{\marking'_{i}}$.
Thus there is a strictly monotone injection 
$f: \{-m,\dots,n\} \mapsto \{-m',\dots,n'\}$
where $f(0)=0$ 
s.t. $c_{f(i)} \ge b_i - b_i^I$
and $c_0 = b_0 - b_0^I + b_0^O$
and
$\sum_{i \neq 0} c_i = (\sum_{i\neq 0} b_{i} - b_{i}^I) + \hat{O}$.

Thus
$\aptpnof{(q,\marking)} = 
(q, b_{-m} \dots b_{-1}, b_0, b_1 \dots b_n) \ttrans
(q', c_{-m'} \dots c_{-1}, c_0, c_1 \dots c_{n'})
= \aptpnof{(q',\marking')}
$.

Now we show the other direction.
If $\aptpnof{(q,\marking)} \ttrans \aptpnof{(q',\marking')}$ then
we have $\aptpnof{(q',\marking')} = (q', c_{-m'} \dots c_{-1}, c_0, c_1 \dots
c_{n'})$ s.t.
\begin{iteMize}{$\bullet$}
\item
$b_{i} = b_{i}^I + b_{i}^R + b_{i}^{\it rest}$ for $-m \le i \le n$
\item
$\match((\sum_{i\neq 0} b_i^I)^{+\epsilon} + b_0^I, \inputs)$
\item
$\match((\sum_{i\neq 0} b_i^R)^{+\epsilon} + b_0^R, \reads)$
\item
$\match(\hat{O}^{+\epsilon} + b_0^O, \outputs)$ 
\item
There is a strictly monotone injection $f: \{-m,\dots,n\} \mapsto \{-m',\dots,n'\}$
where $f(0)=0$ 
s.t. $c_{f(i)} \ge b_i - b_i^I$
and $c_0 = b_0 - b_0^I + b_0^O$
and
$\sum_{i \neq 0} c_i = (\sum_{i\neq 0} b_{i} - b_{i}^I) + \hat{O}$.
\end{iteMize}
As before, each $\marking_i$ can be decomposed into parts 
$\marking_i = \marking_i^{I} + \marking_i^{R} + \marking_i^{\it rest}$,
where $b_i^I = \aptpnof{\marking_i^{I}}$,
$b_i^R = \aptpnof{\marking_i^{R}}$, and
$b_i^{\it rest} = \aptpnof{\marking_i^{\it rest}}$.
Let $I = \sum_i \marking_i^{I}$, 
$R =  \sum_i \marking_i^{R}$, and
$\marking^{\it rest} = \sum_i \marking_i^{\it rest}$.
So we have $\marking = I + R + \marking^{\it rest}$.
Furthermore, since the interval bounds are integers,
we have $\match(I, \inputs)$, $\match(R, \reads)$ and 
$\match(O, \outputs)$.
Finally, due to the conditions on $\hat{O}$ and $b_0^O$,
there exists a marking $O$ s.t.
$\hat{O} + b_0^O = \aptpnof{O}$,
$\marking' = O + R + \marking^{\it rest}$
and
$\aptpnof{(q',\marking')} = (q', c_{-m'} \dots c_{-1}, c_0, c_1 \dots c_{n'})$.
Moreover, this $O$ can be chosen to be in $\delta$-form, for the following
reasons. The tokens in $O$ whose fractional part is the same as a fractional
part in $\marking$ are trivially in $\delta$-form, because $\marking$
is in $\delta$-form.
The tokens in $O$ whose fractional part is between two fractional
parts in $\marking$ is also trivially in $\delta$-form, because $\marking$
is in $\delta$-form.
Now consider the tokens in $O$ whose fractional part is larger
than any fractional part in $M_1 + \dots + M_n$.
Let $\delta_1$ be the maximal fractional part in $M_1 + \dots + M_n$.
We have $\delta_1 < \delta$, because $\marking$ is in $\delta$-form.
Since $\delta_1 < \delta$, the interval $(\delta_1:\delta)$
is nonempty and contains uncountably many different values.
Therefore there is still space for infinitely many different fractional parts
in $O$ in the interval $(\delta_1:\delta)$.
Finally consider the tokens in $O$ whose fractional part is smaller
than any fractional part in $M_{-m} + \dots + M_{-1}$.
Let $\delta_2$ be the minimal fractional part in $M_{-m} + \dots + M_{-1}$.
We have $\delta_2 > 1-\delta$, because $\marking$ is in $\delta$-form.
Therefore there is still space for infinitely many different fractional parts
in $O$ in the nonempty interval $(1-\delta:\delta_2)$.

Thus, since $O$ is in $\delta$-form, 
the transition $(q,\marking) \ttrans (q',\marking')$ is in
$\delta$-form, as required.
\end{proof}

Now we show how to encode timed transitions into A-PTPN.
We define A-PTPN transitions that encode the
effect of PTPN detailed timed transitions $\movesto{x}$ 
for $x \in (0:\delta)$ or $x \in (1-\delta:1)$ for sufficiently small
$\delta>0$. We call these {\em abstract timed transitions}.
For any multiset $b \in \msetsover{(\places\times\zeroto{\maxval +1})}$
let $b^+ \in \msetsover{(\places\times\zeroto{\maxval +1})}$ 
be defined by $b^+((p,x+1)) = b((p,x))$ for $x \le \maxval$
and $b^+((p,\maxval +1)) = b((p,\maxval +1)) + b((p,\maxval))$, i.e., the age
$\maxval+1$ represents all ages $> \maxval$.
There are 4 different types of abstract timed transitions.
(In the following all $b_i$ are nonempty.)

\begin{enumerate}[\hbox to8 pt{\hfill}]      
\item\noindent{\hskip-12 pt\bf Type 1:}\
$(q_1, b_{-m} \dots b_{-1}, b_0, b_1 \dots b_n) \movesto{}
(q_1, b_{-m} \dots b_{-1}, \emptyset, b_0 b_1 \dots b_n)$.
This simulates a very small delay $\delta >0$ where 
the tokens of integer age in $b_0$ now have a positive fractional part,
but no tokens reach an integer age.
\item\noindent{\hskip-12 pt\bf Type 2:}\
$(q_1, b_{-m} \dots b_{-1}, \emptyset, b_1 \dots b_n) \movesto{} 
(q_1, b_{-m} \dots b_{-2}, b_{-1}^{+}, b_1 \dots b_n)$.
This simulates a very small delay $\delta >0$ in the case where
there were no tokens of integer age
and the tokens in $b_{-1}$ just reach the next higher integer age.
\item\noindent{\hskip-12 pt\bf Type 3:}\
$(q_1, b_{-m} \dots b_{-1}, b_0, b_1 \dots b_n) \movesto{} 
(q_1, b_{-m}^{+} \dots b_{-2}^{+} b_{-1}^{+} b_0 \dots b_k, \emptyset,
b_{k+1}^{+} \dots b_n^{+})$ for a $k \in \{0,\dots,n\}$.
This simulates a delay in $(1-\delta:1)$ where the tokens in $b_0 \dots b_k$
do not quite reach the next higher integer and no token gets an integer age.
\item\noindent{\hskip-12 pt\bf Type 4:}\
$(q_1, b_{-m} \dots b_{-1}, b_0, b_1 \dots b_n) \movesto{} 
(q_1, b_{-m}^{+} \dots b_{-2}^{+} b_{-1}^{+} b_0 \dots b_k, b_{k+1}^{+},
b_{k+2}^{+} \dots b_n^{+})$ for some $k \in \{0,\dots,n-1\}$.
This simulates a delay in $(1-\delta:1)$ where the tokens in $b_0, \dots b_k$
do not quite reach the next higher integer and the tokens
on $b_{k+1}$ just reach the next higher integer age.
\end{enumerate}

\noindent The cost model for A-PTPN is defined as follows.
For every transition $\transition \in \transitions$ we have
$\costof{(q_1,M_1) \ttrans (q_2,M_2)} := \costof{t}$, just like in PTPN.
For abstract timed transitions of types 1 and 2 we define the cost as zero.
For abstract timed transitions $(q,M_1) \movesto{} (q,M_2)$ of
types 3 and 4, we define 
$\costof{(q,M_1) \movesto{} (q,M_2)} := \sum_{\place \in \places}
|M_1(\place)|*\costof{\place}$ (i.e., as if the elapsed time had length 1).
The intuition is that, as $\delta$ converges to zero, the cost 
of the PTPN timed transitions of length in $(0:\delta)$ (types 1 and 2)
or in $(1-\delta:1)$ (types 3 and 4) converges to the cost of the
corresponding abstract timed transitions in the A-PTPN.
This will be shown formally in Lemma~\ref{lem:cost-PTPN-APTPN}.

The following two lemmas show the connection between (detailed) timed
transitions in PTPN and abstract timed transitions in the corresponding
A-PTPN.

\begin{lem}\label{lem:PTPN_APTPN_fast}
Let $(q,M)$ be a PTPN configuration in $\delta$-form for some $\delta \le
1/5$ and $x \in (0:\delta)$.
There is a PTPN detailed timed transition 
$(q,M) \movesto{x} (q,M^{+x})$ if and only if
there is an A-PTPN abstract timed transition of type 1 or 2 s.t.
$\aptpn((q,M)) \movesto{} \aptpn((q,M^{+x}))$. 
\end{lem}
\begin{proof}
Let $\marking=\marking_{-m}+ \dots + \marking_{-1}+ \marking_0+
\marking_1+\dots+\marking_n$ be the unique decomposition of $\marking$ into
increasing fractional parts (as defined in Section~\ref{delta:section}), and
$\aptpnof\marking:= \tuple{b_{-m} \dots b_{-1}, \zmset, b_1 \dots b_n}$,
as defined as above.
Let $\epsilon$ be the fractional part of the ages of the tokens
in $\marking_{-1}$. Since $(q,M)$ is in $\delta$-form, we have
$0 < 1-\epsilon < \delta$. 
Now there are two cases.

In the first case we have $x < 1-\epsilon$. 
Then the tokens in $M_{-1}^{+x}$ will have fractional
part $\epsilon +x \in (1-\delta:1)$, and the tokens in $M_{0}^{+x}$ 
will have fractional part $x \in (0:\delta)$.
Therefore $\aptpn((q,M)) =  
(q, \tuple{b_{-m} \dots b_{-1}, \zmset, b_1 \dots b_n})
\movesto{}
(q, \tuple{b_{-m} \dots b_{-1}, \emptyset, \zmset b_1 \dots b_n})
=
\aptpn((q,M^{+x}))$, by an A-PTPN abstract timed transition of type 1,
if and only if $(q,M)\!\movesto{x} (q,M^{+x})$. 

In the second case we must have $x = 1-\epsilon$ and $M_0 = \emptyset$,
because $(q,M) \movesto{x} (q,M^{+x})$ is a detailed timed transition.
In this case exactly the tokens in $M_{-1}$ reach the next higher integer age,
i.e., the tokens in $M_{-1}^{+x}$ have integer age and the integer is one higher than
the integer part of the age of the tokens in $M_0$.
Thus $\aptpn((q,M)) =  
(q, \tuple{b_{-m} \dots b_{-1}, \emptyset, b_1 \dots b_n})
\movesto{}
(q, \tuple{b_{-m} \dots b_{-2}, b_{-1}^+, b_1 \dots b_n})
=
\aptpn((q,M^{+x}))$, by an A-PTPN abstract timed transition of type 2,
if and only if $(q,M) \movesto{x} (q,M^{+x})$.
\end{proof}

\begin{lem}\label{lem:PTPN_APTPN_slow}
Let $(q,M)$ be a PTPN configuration in $\delta$-form for some $\delta \le
1/5$ and $x \in (1-\delta:1)$.
There is a PTPN timed transition 
$(q,M) \movesto{x} (q,M^{+x})$ if and only if
there is an A-PTPN transition of either type 3 or 4 s.t.
$\aptpn((q,M)) \movesto{} \aptpn((q,M^{+x}))$. 
\end{lem}
\begin{proof}
Let $\marking=\marking_{-m}+ \dots + \marking_{-1}+ \marking_0+
\marking_1+\dots+\marking_n$ be the unique decomposition of $\marking$ into
increasing fractional parts (as defined in Section~\ref{delta:section}), and
$\aptpnof\marking:= \tuple{b_{-m} \dots b_{-1}, \zmset, b_1 \dots b_n}$,
as defined above.
Let $\epsilon_k$ be the fractional part of the ages of the tokens
in $\marking_{k}$ for $0 \le k \le n$. 
Since $(q,M)$ is in $\delta$-form, we have
$0 < \epsilon_k < \delta$. 
Now there are two cases.

In the first case we have $x \in (1-\epsilon_{k+1}:1-\epsilon_k) \subseteq
(1-\delta:1)$ for some $0 \le k \le n$.
(If $k=n$ we have $x \in (1-\delta:1-\epsilon_n)$, and
if $k=0$ we have $x \in (1-\epsilon_1:1)$.)
Then, in the step from $M_{k+1}$ to $M_{k+1}^{+x}$, the token ages 
in $M_{k+1}$ reach and
slightly exceed the next higher integer age, while
the token ages in $M_{k}^{+x}$ still stay slightly below the next higher
integer.
Therefore $\aptpn((q,M)) =  
(q, \tuple{b_{-m} \dots b_{-1}, \zmset, b_1 \dots b_n})
\movesto{}
(q, \tuple{b_{-m}^+ \dots b_{-1}^+ b_0 \dots b_k, \emptyset, b_{k+1}^+ \dots b_n^+})
=
\aptpn((q,M^{+x}))$, by an A-PTPN abstract timed transition of type 3,
if and only if $(q,M)\!\movesto{x} (q,M^{+x})$.

The only other case is where $x = 1-\epsilon_{k+1}$ for
some $k \in \{0,\dots, n-1\}$.
Here exactly the tokens in $M_{k+1}$ reach the next higher integer age.
Therefore $\aptpn((q,M)) =  
(q, \tuple{b_{-m} \dots b_{-1}, \zmset, b_1 \dots b_n})
\movesto{}
(q, \tuple{b_{-m}^+ \dots b_{-1}^+ b_0 \dots b_k, b_{k+1}^+, b_{k+1}^+ \dots b_n^+})
=
\aptpn((q,M^{+x}))$, by an A-PTPN abstract timed transition of type 4,
if and only if $(q,M) \movesto{x} (q,M^{+x})$.
\end{proof}

The following Lemma~\ref{lem:cost-PTPN-APTPN}, which follows from
Lemmas~\ref{lem:PTPN_APTPN_disc},\ref{lem:PTPN_APTPN_fast},
and \ref{lem:PTPN_APTPN_slow},
shows the connection between the computation costs in a PTPN and
the corresponding A-PTPN.

\begin{lem}\label{lem:cost-PTPN-APTPN}\ 
\begin{enumerate}[\em(1)]
\item
Let $\conf_0$ be a PTPN configuration where all tokens have integer ages.
For every PTPN computation $\comp = \conf_0 \movesto{} \dots \movesto{} \conf_n$
in detailed form and $\delta$-form s.t. $n*\delta \le 1/5$ there exists a corresponding
A-PTPN computation $\comp' = \aptpn(\conf_0) \movesto{} \dots \movesto{} \aptpn(\conf_n)$
s.t. 
\[
|\costof{\comp}-\costof{\comp'}| \le n*\delta*(\max_{0\le i\le n}|\conf_i|)*
(\max_{\place \in \places}\costof{\place})
\]
\item
Let $\conf_0'$ be an A-PTPN configuration $\tuple{\epsilon,\zmset,\epsilon}$.
For every A-PTPN computation $\comp' = \conf_0' \movesto{} \dots \movesto{} \conf_n'$
and every $0 < \delta \le 1/5$ there exists a PTPN computation
$\comp = \conf_0 \movesto{} \dots \movesto{} \conf_n$ in
detailed form and $\delta$-form s.t. $\conf_i' = \aptpn(\conf_i)$ for $0 \le i \le n$
and 
\[
|\costof{\comp}-\costof{\comp'}| \le n*\delta*(\max_{0\le i\le n}|\conf_i'|)*
(\max_{\place \in \places}\costof{\place})
\]
\end{enumerate}
\end{lem}
\begin{proof}
For the first part let $\comp = \conf_0 \movesto{} \dots \movesto{}
\conf_n$ be a PTPN computation in detailed form and $\delta$-form
s.t. $n*\delta \le 1/5$.
So every timed transition
$\movesto{x}$ has either $x \in (0:\delta)$ or $x \in (1-\delta:1)$.
Furthermore, the fractional part of the 
age of every token in any configuration $\conf_i$ is $< i*\delta$ away
from the nearest integer, because $\conf_0$ only contains tokens with integer
ages. Since $i \le n$ these ages are $< n*\delta \le 1/5$ away from the nearest
integer. Moreover, $\comp$ is detailed and thus
Lemmas~\ref{lem:PTPN_APTPN_disc},
\ref{lem:PTPN_APTPN_fast} and \ref{lem:PTPN_APTPN_slow} apply.
Thus there exists a corresponding
A-PTPN computation $\comp' = \aptpn(\conf_0) \movesto{} \dots \movesto{}
\aptpn(\conf_n)$.
By definition of the cost of A-PTPN transitions, for every discrete transition
$\conf_i \movesto{} \conf_{i+1}$ we have
$\costof{\conf_i \movesto{} \conf_{i+1}} = 
\costof{\aptpn(\conf_i) \movesto{} \aptpn(\conf_{i+1})}$.
Moreover, for every timed transition
$\conf_i \movesto{x} \conf_{i+1}$ we have
$|\costof{\conf_i \movesto{x} \conf_{i+1}} - 
\costof{\aptpn(\conf_i) \movesto{} \aptpn(\conf_{i+1})}| \le 
\delta*|\conf_i|*(\max_{\place \in \places}\costof{\place})$, 
because either $x \in (0:\delta)$ or $x \in (1-\delta:1)$.
Therefore $|\costof{\comp}-\costof{\comp'}| \le n*\delta*(\max_{0\le i\le n}|\conf_i|)*
(\max_{\place \in \places}\costof{\place})$ as required.

For the second part let $\conf_0$ be a PTPN configuration
s.t. $\tuple{\epsilon,\zmset,\epsilon} = \conf_0' = \aptpn(\conf_0)$,
i.e., all tokens in $\conf_0$ have integer ages.
We now use Lemmas~\ref{lem:PTPN_APTPN_disc},
\ref{lem:PTPN_APTPN_fast} and \ref{lem:PTPN_APTPN_slow}
to construct the PTPN computation $\comp$.
Let $\delta_i := \delta * 2^{i-n}$ for $0 \le i \le n$.
The construction ensures the following invariants.
(1) $\conf_i' = \aptpn(\conf_i)$, and
(2) $\conf_{i}$ is in $\delta_i$-form.
Condition (1) follows directly from Lemmas~\ref{lem:PTPN_APTPN_disc},
\ref{lem:PTPN_APTPN_fast} and \ref{lem:PTPN_APTPN_slow}.
For the base case $i=0$, condition (2) holds trivially, because
all tokens in $\conf_0$ have integer ages.
Now we consider the step from $i$ to $i+1$.
Since $\conf_i$ is in $\delta_i$-form, we obtain from
Lemmas~\ref{lem:PTPN_APTPN_disc},
\ref{lem:PTPN_APTPN_fast} and \ref{lem:PTPN_APTPN_slow}
that if the $i-th$ transition in this sequence is a timed transition
$\movesto{x}$ then 
either $x \in (0:\delta_i)$ or $x \in (1-\delta_i:1)$.
Therefore, since 
$\conf_i$ is in $\delta_i$-form,
$\conf_{i+1}$ is in $(2*\delta_i)$-form and thus in
$\delta_{i+1}$-form.

Now we consider the cost of the PTPN computation $\comp$.
By definition of the cost of A-PTPN transitions, for every discrete transition
$\conf_i \movesto{} \conf_{i+1}$ we have
$\costof{\conf_i \movesto{} \conf_{i+1}} = 
\costof{\aptpn(\conf_i) \movesto{} \aptpn(\conf_{i+1})}$.
Moreover, for every timed transition
$\conf_i \movesto{x} \conf_{i+1}$ we have
$|\costof{\conf_i \movesto{x} \conf_{i+1}} - 
\costof{\aptpn(\conf_i) \movesto{} \aptpn(\conf_{i+1})}| \le 
\delta_i*|\conf_i'|*(\max_{\place \in \places}\costof{\place})$, 
because either $x \in (0:\delta_i)$ or $x \in (1-\delta_i:1)$.
Therefore $|\costof{\comp}-\costof{\comp'}| \le n*\delta*(\max_{0\le i\le n}|\conf_i'|)*
(\max_{\place \in \places}\costof{\place})$ as required.
\end{proof}

The following theorem summarizes the results of this section.
To compute the optimal cost of a PTPN, it suffices to consider the
corresponding A-PTPN.

\begin{thm}\label{thm:samecost-PTPN-APTPN}
The infimum of the costs in a PTPN coincides with the infimum of the costs in the
corresponding A-PTPN.
\[
\inf\{\costof{\comp} \,|\, \initconf \compto{\comp} \finalconfs\}
= 
\inf\{\costof{\comp'} \,|\, \aptpn(\initconf) \compto{\comp'} \aptpn(\finalconfs)\}
\]
\end{thm}
\begin{proof}
We show that neither of the two costs for PTPN and A-PTPN can be larger than
the other.\\
Let $I := \inf\{\costof{\comp} \,|\, \initconf 
\compto{\comp} \finalconfs\}$
and
$I' := \inf\{\costof{\comp'} \,|\, \aptpn(\initconf) \compto{\comp'}
\aptpn(\finalconfs)\}$.

First we show that $I' \not > I$.
By definition of $I$, for every $\lambda > 0$ there is a 
computation $\initconf \compto{\comp_\lambda} \finalconfs$,
s.t. $\costof{\comp_\lambda} - I \le \lambda$.
Without restriction we can assume that $\comp_\lambda$
is also in detailed form.
Let $n_\lambda := |\comp_\lambda|$ be the length of $\comp_\lambda$
and
$\comp_\lambda = \conf_0 \movesto{} \dots \movesto{}
\conf_{n_\lambda}$.
Let $\delta_\lambda := \min\{1/(5n_\lambda), \lambda/(n_\lambda * (\max_{0\le i\le n_\lambda}|\conf_i|)*
(\max_{\place \in \places}\costof{\place}))\}$.

By Lemma~\ref{lem:deltaform} there exists a computation 
$\initconf \compto{\comp''_\lambda} \finalconfs$
in detailed form and
$\delta_\lambda$-form where $|\comp''_\lambda| = |\comp_\lambda|$ and
$\comp''_\lambda = \conf''_0 \movesto{} \dots \movesto{}
\conf''_{n_\lambda}$ s.t. $|\conf''_i| = |\conf_i|$
and $\costof{\comp''_\lambda} \le \costof{\comp_\lambda}$.
It follows that $\costof{\comp''_\lambda} - I \le \lambda$.

By part (1) of Lemma~\ref{lem:cost-PTPN-APTPN}, there exists a  
corresponding A-PTPN computation 
$\comp'_\lambda = \aptpn(\conf_0'') \movesto{} \dots \movesto{} \aptpn(\conf_{n_\lambda}'')$
s.t. $|\costof{\comp''_\lambda}-\costof{\comp'_\lambda}| \le
n_\lambda*\delta_\lambda*(\max_{0\le i\le
  n_\lambda}|\conf''_i|)*(\max_{\place \in \places}\costof{\place}) \le \lambda$.
Thus we obtain
$\costof{\comp'_\lambda} - I \le 2\lambda$.
Since this holds for every $\lambda >0$ we get
$I' \not > I$.

Now we show that $I \not > I'$.
By definition of $I'$, for every $\lambda > 0$ there is a 
A-PTPN computation $\initconf \compto{\comp'_\lambda} \finalconfs$,
s.t. $\costof{\comp'_\lambda} - I' \le \lambda$.
Let $n_\lambda := |\comp'_\lambda|$ be the length of $\comp'_\lambda$
and
$\comp'_\lambda = \conf'_0 \movesto{} \dots \movesto{}
\conf'_{n_\lambda}$.
Let $\delta_\lambda := \min\{1/(5n_\lambda), \lambda/(n_\lambda * (\max_{0\le i\le n_\lambda}|\conf'_i|)*
(\max_{\place \in \places}\costof{\place}))\}$.

By Lemma~\ref{lem:cost-PTPN-APTPN} (2), there exists a  
corresponding PTPN computation 
$\comp_\lambda = \conf_0 \movesto{} \dots \movesto{} \conf_{n_\lambda}$
in detailed form and $\delta_\lambda$-form 
s.t. $\conf_i' = \aptpn(\conf_i)$ and
$|\costof{\comp_\lambda}-\costof{\comp'_\lambda}| \le
n_\lambda*\delta_\lambda*(\max_{0\le i\le
  n_\lambda}|\conf'_i|)*(\max_{\place \in \places}\costof{\place}) \le \lambda$.
Thus we obtain
$\costof{\comp_\lambda} - I' \le 2\lambda$.
Since this holds for every $\lambda >0$ 
we get $I \not > I'$.

By combining $I' \not > I$ with $I \not > I'$
we obtain $I=I'$ as required.
\end{proof}

\subsection{The Running Example (cont.)}\label{subsec:example2}

We consider the running example from Subsection~\ref{subsec:example1} again.

\begin{figure}[htbp]
\begin{center}
\begin{tikzpicture}[show background rectangle,label distance=-1.2mm]

\node[transition,label=above:{\tiny $\transition_1$},label=right:{\tiny $1$}](t1){}; %
\node[place,,minimum size=10mm,below  = 9mm of t1,label=right:{\tiny $\place_2$},label=left:{\tiny $2$}] (p2){}
[children are tokens, token distance=5mm]
child {node [token] {$6.5$}};

\node[place,draw=blue,fill=blue,left  = 7mm of p2,label=left:{\tiny $\state_1$}] (q1){};
\node[place,minimum size=10mm,left  = 7mm of q1,label=left:{\tiny $\place_1$},label=above:{\tiny $3$}] (p1){}
[children are tokens, token distance=5mm]
child {node [token] {$3.1$}}
child {node [token] {$3.1$}}
child {node [token] {$2.5$}};

\node[place,draw=blue,fill=blue,right  = 7mm of p2,label=right:{\tiny $\state_2$}] (q2){};
\node[place,minimum size=10mm,right  = 7mm of q2,label=right:{\tiny $\place_3$},label=below:{\tiny $0$}] (p3){}
[children are tokens, token distance=5mm]
child {node [token] {$0.1$}}
child {node [token] {$0.1$}};
;

\node[transition,below  =  9mm of p2,label=below:{\tiny $\transition_2$},label=left:{\tiny $3$}](t2){};

\draw[->,thick] (t1) -- (p2);
\draw[->,thick] (q1) -- (t1);
\draw[->,thick] (t1) -- (q2);
\draw[->,thick] (p1) -- node[sloped,above=-1mm] {\tiny $(0,3]$} (t1);
\draw[->,thick] (t1) --  node[sloped,above=-1mm] {\tiny $[1,5)$} (p2);
\draw[->,thick] (t1) -- node[sloped,above=-1mm] {\tiny $(2,\infty)$} (p3);
\draw[->,thick] (t2) -- node[sloped,below=-1mm] {\tiny $[0,\infty)$}(p1);
\draw[<->,thick] (t2) -- node[sloped,below=-1mm] {\tiny $[2,2]$}(p2);
\draw[->,thick] (p3) -- node[sloped,below=-1mm] {\tiny $[1,4)$} (t2);
\draw[->,thick] (q2) -- (t2);
\draw[->,thick] (t2) -- (q1);

\end{tikzpicture}
\end{center}
\caption{A simple example of a PTPN.
}
\end{figure}

\paragraph{\bf Abstract Markings}
Fix $\delta=0.2$. 
Then the configuration
\[
\conf=
\begin{array}[t]{l}
[\tuple{\place_1,2.1},\tuple{\place_1,1.0},\tuple{\place_1,2.85},\tuple{\place_1,3.9},
\\
\tuple{\place_2,1.1},\tuple{\place_2,9.1},\tuple{\place_2,1.0},\tuple{\place_2,9.85},
\\
\tuple{\place_3,8.1},\tuple{\place_3,0.85},\tuple{\place_3,2.9},\tuple{\place_3,4.9},\tuple{\place_3,9.0}]
\end{array}
\]
is in $\delta$-form.
We have 
\[
\begin{array}{l}
\conf_1=\aptpnof\conf=
\\
\left(
\state_1,
\left(
\left[
\begin{array}{c}
\tuple{\place_1,2} \\,\\ \tuple{\place_2,6}  \\,\\   \tuple{\place_3,0}
\end{array}
\right]
\left[
\begin{array}{c}
\tuple{\place_1,3} \\,\\ \tuple{\place_3,2} \\,\\ \tuple{\place_3,4} 
\end{array}
\right]

,

\left[
\begin{array}{c}
\tuple{\place_1,1} \\,\\ \tuple{\place_2,1} \\,\\ \tuple{\place_3,6} 
\end{array}
\right]

,

\left[
\begin{array}{c}
\tuple{\place_1,2} \\,\\ \tuple{\place_2,1}\\,\\ \tuple{\place_2,6} \\,\\ \tuple{\place_3,6} 
\end{array}
\right]
\right)
\right)
\end{array}
\]

Note that token ages $> \maxval$ are abstracted as $\maxval+1$.
Since here $\maxval = 5$, all token ages $> 5$ are abstracted as $6$.

Below we describe four examples of abstract computation steps
(not the same as in Subsection~\ref{subsec:example1}).

(i) A type $1$ transition from $\conf_1$ leads to
\[
\begin{array}{l}
\conf_2=
\\
\left(
\state_1
\left(
\left[
\begin{array}{c}
\tuple{\place_1,2} \\,\\ \tuple{\place_2,6}  \\,\\   \tuple{\place_3,0}
\end{array}
\right]
\left[
\begin{array}{c}
\tuple{\place_1,3} \\,\\ \tuple{\place_3,2} \\,\\ \tuple{\place_3,4} 
\end{array}
\right]

,

\emptyset

,
\left[
\begin{array}{c}
\tuple{\place_1,1} \\,\\ \tuple{\place_2,1} \\,\\ \tuple{\place_3,6} 
\end{array}
\right]

\left[
\begin{array}{c}
\tuple{\place_1,2} \\,\\ \tuple{\place_2,1}\\,\\ \tuple{\place_2,6} \\,\\ \tuple{\place_3,6} 
\end{array}
\right]
\right)
\right)
\end{array}
\]

(ii) A type $2$ transition from $\conf_2$ leads to
\[
\begin{array}{l}
\conf_3=
\\
\left(
\state_1,
\left(
\left[
\begin{array}{c}
\tuple{\place_1,2} \\,\\ \tuple{\place_2,6}  \\,\\   \tuple{\place_3,0}
\end{array}
\right]

,

\left[
\begin{array}{c}
\tuple{\place_1,4} \\,\\ \tuple{\place_3,3} \\,\\ \tuple{\place_3,5} 
\end{array}
\right]

,

\left[
\begin{array}{c}
\tuple{\place_1,1} \\,\\ \tuple{\place_2,1} \\,\\ \tuple{\place_3,6} 
\end{array}
\right]

\left[
\begin{array}{c}
\tuple{\place_1,2} \\,\\ \tuple{\place_2,1}\\,\\ \tuple{\place_2,6} \\,\\ \tuple{\place_3,6} 
\end{array}
\right]
\right)
\right)
\end{array}
\]

(iii) A type $3$ transition from $\conf_3$ leads to
\[
\begin{array}{l}
\conf_4=
\\
\left(
\state_1,
\left(
\left[
\begin{array}{c}
\tuple{\place_1,3} \\,\\ \tuple{\place_2,6}  \\,\\   \tuple{\place_3,1}
\end{array}
\right]

\left[
\begin{array}{c}
\tuple{\place_1,4} \\,\\ \tuple{\place_3,3} \\,\\ \tuple{\place_3,5} 
\end{array}
\right]

\left[
\begin{array}{c}
\tuple{\place_1,1} \\,\\ \tuple{\place_2,1} \\,\\ \tuple{\place_3,6} 
\end{array}
\right]

,

\emptyset

,

\left[
\begin{array}{c}
\tuple{\place_1,3} \\,\\ \tuple{\place_2,2}\\,\\ \tuple{\place_2,6} \\,\\ \tuple{\place_3,6} 
\end{array}
\right]
\right)
\right)
\end{array}
\]

(iv) A type $4$ transition from $\conf_3$ leads to
\[
\begin{array}{l}
\conf_5=
\\
\left(
\state_1,
\left(
\left[
\begin{array}{c}
\tuple{\place_1,3} \\,\\ \tuple{\place_2,6}  \\,\\   \tuple{\place_3,1}
\end{array}
\right]

\left[
\begin{array}{c}
\tuple{\place_1,4} \\,\\ \tuple{\place_3,3} \\,\\ \tuple{\place_3,5} 
\end{array}
\right]

,

\left[
\begin{array}{c}
\tuple{\place_1,2} \\,\\ \tuple{\place_2,2} \\,\\ \tuple{\place_3,6} 
\end{array}
\right]

,

\left[
\begin{array}{c}
\tuple{\place_1,3} \\,\\ \tuple{\place_2,2}\\,\\ \tuple{\place_2,6} \\,\\ \tuple{\place_3,6} 
\end{array}
\right]
\right)
\right)
\end{array}
\]

Below, we give three concrete timed transitions that correspond to the 
abstract steps (i)-(iii) described above.

\[
\begin{array}[t]{l}
[\tuple{\place_1,2.1},\tuple{\place_1,1.0},\tuple{\place_1,2.85},\tuple{\place_1,3.9},
\\
\tuple{\place_2,1.1},\tuple{\place_2,9.1},\tuple{\place_2,1.0},\tuple{\place_2,9.85},
\\
\tuple{\place_3,8.1},\tuple{\place_3,0.85},\tuple{\place_3,2.9},\tuple{\place_3,4.9},\tuple{\place_3,9.0}]

\\ \\\xtimedmovesto{0.01}\\ \\

[\tuple{\place_1,2.11},\tuple{\place_1,1.01},\tuple{\place_1,2.86},\tuple{\place_1,3.91},
\\
\tuple{\place_2,1.11},\tuple{\place_2,9.11},\tuple{\place_2,1.01},\tuple{\place_2,9.86},
\\
\tuple{\place_3,8.11},\tuple{\place_3,0.86},\tuple{\place_3,2.91},\tuple{\place_3,4.91},\tuple{\place_3,9.01}]

\\ \\\xtimedmovesto{0.09}\\ \\

[\tuple{\place_1,2.2},\tuple{\place_1,1.1},\tuple{\place_1,2.95},\tuple{\place_1,4.0},
\\
\tuple{\place_2,1.2},\tuple{\place_2,9.2},\tuple{\place_2,1.1},\tuple{\place_2,9.95},
\\
\tuple{\place_3,8.2},\tuple{\place_3,0.95},\tuple{\place_3,3.0},\tuple{\place_3,5.0},\tuple{\place_3,9.1}]

\\ \\\xtimedmovesto{0.85}\\ \\

[\tuple{\place_1,3.05},\tuple{\place_1,1.95},\tuple{\place_1,3.8},\tuple{\place_1,4.85},
\\
\tuple{\place_2,2.05},\tuple{\place_2,10.05},\tuple{\place_2,1.95},\tuple{\place_2,10.8},
\\
\tuple{\place_3,9.05},\tuple{\place_3,1.8},\tuple{\place_3,3.85},\tuple{\place_3,5.85},\tuple{\place_3,9.95}]

\end{array}
\]

A concrete timed transitions that correspond to the 
abstract step  (iv) is the following

\[
\begin{array}[t]{l}

[\tuple{\place_1,2.2},\tuple{\place_1,1.1},\tuple{\place_1,2.95},\tuple{\place_1,4.0},
\\
\tuple{\place_2,1.2},\tuple{\place_2,9.2},\tuple{\place_2,1.1},\tuple{\place_2,9.95},
\\
\tuple{\place_3,8.2},\tuple{\place_3,0.95},\tuple{\place_3,3.0},\tuple{\place_3,5.0},\tuple{\place_3,9.1}]

\\ \\\xtimedmovesto{0.9}\\ \\

[\tuple{\place_1,3.1},\tuple{\place_1,2.0},\tuple{\place_1,3.85},\tuple{\place_1,4.9},
\\
\tuple{\place_2,2.1},\tuple{\place_2,10.1},\tuple{\place_2,2.0},\tuple{\place_2,10.85},
\\
\tuple{\place_3,9.1},\tuple{\place_3,1.85},\tuple{\place_3,3.9},\tuple{\place_3,5.9},\tuple{\place_3,10.0}]

\end{array}
\]

\section{Abstracting Costs in A-PTPN}\label{sec:cost-abstraction}
\noindent
Given an A-PTPN, the cost-threshold problem is whether
there exists a computation 
$\aptpn(\initconf) \compto{\comp} \aptpn(\finalconfs)$
s.t. $\costof{\comp} \le \threshold$ for a given threshold $\threshold$.

We now reduce this question to a question about simple coverability
in a new model called AC-PTPN. The idea is to encode the cost of the
computation into a part of the control-state.
For every A-PTPN and cost threshold $\threshold \in \nat$ 
there is a corresponding AC-PTPN that is defined as follows.

For every A-PTPN configuration $(q, b_{-m} \dots b_{-1}, b_0, b_1 \dots b_n)$
there are AC-PTPN configurations 
$((q,y), b_{-m} \dots b_{-1}, b_0, b_1 \dots b_n)$
for all integers $y: 0 \le y \le \threshold$, where $y$ represents the remaining
allowed cost of the computation.
We define a finite set of functions $\ac_y$ for $0 \le y \le \threshold$
that map A-PTPN configurations to AC-PTPN configurations
s.t. $\ac_y((q, b_{-m} \dots b_{-1}, b_0, b_1 \dots b_n)) =
((q,y), b_{-m} \dots b_{-1}, b_0, b_1 \dots b_n)$.

For every discrete transition 
$\transition = (q_1,q_2,\inputs,\reads,\outputs) \in \transitions$
with 
\[
(q_1, b_{-m} \dots b_{-1}, b_0, b_1 \dots b_n) \ttrans
(q_2, c_{-m'} \dots c_{-1}, c_0, c_1 \dots c_{n'})
\]
in the A-PTPN,
we have instead
\[
((q_1,y), b_{-m} \dots b_{-1}, b_0, b_1 \dots b_n) \ttrans
((q_2,y-\costof{\transition}, c_{-m'} \dots c_{-1}, c_0, c_1 \dots c_{n'})
\]
in the AC-PTPN for all nonnegative integers $y$ s.t. $\threshold \ge y \ge \costof{\transition}$.
I.e., we deduct the cost of the transition from the remaining allowed cost 
of the computation.

For every A-PTPN abstract timed transition of the types 1 and 2
$(q_1, \dots) \movesto{} (q_1, \dots)$
we have corresponding AC-PTPN abstract timed transitions of the same types 1 and 2
where $((q_1,y), \dots) \movesto{} ((q_1,y), \dots)$ for all 
$0 \le y \le \threshold$. I.e., infinitesimally small delays do not cost
anything, and thus no cost is deducted from the remaining allowed cost.

For every A-PTPN abstract timed transition of type 3
(for some $k \in \{0,\dots,n\}$)
\[
(q_1, b_{-m} \dots b_{-1}, b_0, b_1 \dots b_n) \movesto{} 
(q_1, b_{-m}^{+} \dots b_{-2}^{+} b_{-1}^{+} b_0 \dots b_k, \emptyset,
b_{k+1}^{+} \dots b_n^{+})
\] 
we have corresponding AC-PTPN abstract timed
transitions of type 3 with
\[
((q_1,y), b_{-m} \dots b_{-1}, b_0, b_1 \dots b_n) \movesto{} 
((q_1,y-z), b_{-m}^{+} \dots b_{-2}^{+} b_{-1}^{+} b_0 \dots b_k, \emptyset,
b_{k+1}^{+} \dots b_n^{+})
\]
for all nonnegative integers $y$ s.t. $\threshold \ge y \ge z$
where $z = \sum_{i=-m}^n \sum_{\place \in \places}
|b_i(\place)|*\costof{\place}$. I.e., we deduct the cost of the timed step
(whose length is infinitesimally close to 1) from the remaining allowed cost 
of the computation.

\noindent
Similarly, for every A-PTPN abstract timed transition of type 4
(for some $k \in \{0,\dots,n-1\}$) with
\[
(q_1, b_{-m} \dots b_{-1}, b_0, b_1 \dots b_n) \movesto{} 
(q_1, b_{-m}^{+} \dots b_{-2}^{+} b_{-1}^{+} b_0 \dots b_k, b_{k+1}^{+},
b_{k+2}^{+} \dots b_n^{+})
\] 
we have corresponding AC-PTPN abstract timed
transitions of type 4 with
\[
((q_1,y), b_{-m} \dots b_{-1}, b_0, b_1 \dots b_n) \movesto{} 
((q_1,y-z), b_{-m}^{+} \dots b_{-2}^{+} b_{-1}^{+} b_0 \dots b_k, b_{k+1}^{+},
b_{k+2}^{+} \dots b_n^{+})
\]
for all nonnegative integers $y$ s.t. $\threshold \ge y \ge z$
where $z = \sum_{i=-m}^n \sum_{\place \in \places}
|b_i(\place)|*\costof{\place}$. I.e., we deduct the cost of the timed step
(whose length is infinitesimally close to 1) from the remaining allowed cost 
of the computation.

The following lemma summarizes the connection between A-PTPN and
AC-PTPN.

\begin{lem}\label{lem:ac-ptpn-comp}
There exists an A-PTPN computation 
$\aptpn(\initconf) \compto{\comp} \aptpn(\finalconfs)$ with
$\costof{\comp} \le \threshold$ iff
there exists a corresponding AC-PTPN computation 
$\ac_\threshold(\aptpn(\initconf)) \compto{\comp'} 
\bigcup_{0 \le y \le \threshold} 
\ac_y(\aptpn(\finalconfs))$
\end{lem}
\begin{proof}
Directly from the definition of AC-PTPN.
\end{proof}

Moreover, the connection between A-PTPN and AC-PTPN is constructive, i.e., the
cost-threshold problem for A-PTPN can be reduced to the coverability problem
for AC-PTPN described in Lemma~\ref{lem:ac-ptpn-comp}.

Note that, unlike A-PTPN, AC-PTPN are not monotone.
This is because steps of types 3 and 4 with
more tokens on cost-places cost more, and thus
cost-constraints might block such transitions from larger configurations.
Thus one cannot directly reduce the AC-PTPN coverability problem to
an A-PTPN coverability problem.

\section{The Abstract Coverability Problem}\label{sec:phase}
\noindent
We describe a general construction for solving reachability/coverability problems
for infinite-state systems under some abstract conditions.
In particular, we do {\em not} assume that the behavior of these systems is
fully monotone w.r.t. some well-quasi-order 
(unlike in many related works \cite{Parosh:Bengt:Karlis:Tsay:general:IC,Finkel:Schnoebelen:everywhere:TCS}).

In subsequent sections we will show how this construction can be applied to
solve AC-PTPN coverability, and thus the A-PTPN and PTPN cost-threshold problems.

\subsection{The Generalized Valk-Jantzen Construction}

\begin{thm}\label{thm:ValkJantzen}(Valk \& Jantzen
  \cite{ValkJantzen})\quad
Given an upward-closed set $V \subseteq \nat^k$, the finite set 
$V_{\it min}$ of minimal elements of $V$ is
effectively computable iff for any vector $\vecdef u \in \nat_\omega^k$ the
predicate $\vecdef u\!\downarrow \cap \;V \neq \emptyset$ is decidable.
\end{thm}

We now show a generalization of this result.

\begin{thm}\label{thm:GVJ}
Let $(\Omega,\le)$ be a set 
with a decidable well-quasi-order (wqo) $\le$, and let $V \subseteq \Omega$ be
upward-closed and recursively enumerable.  Then the finite set $V_{\it min}$ of minimal elements 
of $V$ is effectively constructible if and only if for every finite subset 
$X \subseteq \Omega$ it is decidable if $V \cap \overline{\uc{X}} \neq
\emptyset$ (i.e., if $\exists v \in V.\, v \notin \uc{X}$).
\end{thm}
\begin{proof}
$V_{\it min}$ is finite, since $\le$ is a wqo.
For the only-if part, since $\uc{X}$ is upward-closed,
it suffices to check for each of the finitely many elements of
$V_{\it min}$ if it is not in $\uc{X}$. This is possible, because $X$ is
finite and $\le$ is decidable.

For the if-part, we start with $X=\emptyset$ and keep adding elements to
$X$ until $\uc{X} = V$. In every step we do the 
(by assumption decidable) check if 
$\exists v \in V.\, v \notin \uc{X}$. If no, we stop. 
If yes, we enumerate $V$ and check for every element $v$ if
$v \notin \uc{X}$ (this is possible since $X$ is finite and $\le$
is decidable). Eventually, we will find such a $v$, add it to the set $X$,
and do the next step.
Consider the sequence of elements $v_1,v_2,\dots$ which are added to $X$ in
this way. By our construction $v_j \not\ge v_i$ for $j > i$. Thus the sequence
is finite, because $\le$ is a wqo. Therefore the algorithm terminates
and the final set $X$ satisfies $\not\exists v \in V.\, v \notin \uc{X}$, 
i.e., $V \subseteq \uc{X}$. Furthermore, by our construction 
$X \subseteq V$ and thus $\uc{X} \subseteq V\!\uparrow = V$.
Thus $\uc{X} = V$. Finally, we remove all non-minimal elements
from $X$ (this is possible since $X$ is finite and $\le$ is decidable)
and obtain $V_{\it min}$. 
\end{proof}


\begin{cor}\label{cor:GVJ_string}
Let $\Sigma$ be a finite alphabet and $V \subseteq \Sigma^*$
a recursively enumerable set that is 
upward-closed w.r.t. the substring ordering $\le$.
The following three properties are equivalent.
\begin{enumerate}
\item
The finite set $V_{\it min}$ of minimal elements 
of $V$ is effectively constructible.
\item
For every finite subset $X \subseteq \Sigma^*$ it is decidable if
$\exists v \in V.\, v \notin \uc{X}$.
\item
For every regular language $R \subseteq \Sigma^*$ it is decidable
if $R \cap V = \emptyset$.
\end{enumerate}
\end{cor}
\begin{proof}
By Higman's Lemma \cite{Higman:divisibility}, 
the substring order $\le$ is a wqo on $\Sigma^*$ and thus
$V_{\it min}$ is finite.
Therefore the equivalence of (1) and (2) follows from Theorem~\ref{thm:GVJ}.
Property (1) implies that $V$ is an effectively constructible
regular language, which implies property (3). 
Property (2) is equivalent to checking whether 
$V \cap \overline{\uc{X}} \neq \emptyset$ and 
$\overline{\uc{X}}$ is effectively regular because $X$ is finite.
Therefore, (3) implies (2) and thus (1). 
\end{proof}

Another interpretation of Corollary~\ref{cor:GVJ_string} is the following.
An upward-closed set of strings $V$ is always regular, but this regular language 
is effectively constructible if and only if regular queries to $V$ are
decidable.

Note that Theorem~\ref{thm:GVJ} 
(and even Corollary~\ref{cor:GVJ_string}, via an encoding of vectors
into strings) imply Theorem~\ref{thm:ValkJantzen}.

\subsection{The Abstract Phase Construction}

We define some sufficient abstract conditions on infinite-state
transition systems under which a general
reachability/coverability problem is decidable. 
Intuitively, we have two different types of transition relations.
The first relation is monotone (w.r.t. a given quasi-order) on the whole
state space, while the second relation is only defined/enabled
on an upward-closed subspace. The quasi-order is not a well quasi-order 
on the entire space, but only on the subspace. In particular, this is not
a well-quasi-ordered transition system in the sense of
\cite{Parosh:Bengt:Karlis:Tsay:general:IC,Finkel:Schnoebelen:everywhere:TCS}, but more general.

We call the following algorithm the {\em abstract phase construction},
because we divide sequences of transitions into phases, separated
by occurrences of transitions of the second kind.

\begin{defi}\label{def:phase}
We say that a structure $(S,C,\le,\rightarrow, \rightarrow_A, \rightarrow_B,{\it init}, F)$
satisfies the {\em abstract phase construction requirements} 
iff the following conditions hold.
\begin{enumerate}[\hbox to8 pt{\hfill}]                                        

\item\noindent{\hskip-12 pt\bf 1.:}\
$S$ is a (possibly infinite) set of states, $C \subseteq S$ is a finite
subset, ${\it init} \in S$ is the initial state and $F \subseteq S$ is a 
(possibly infinite) set of final states.
\item\noindent{\hskip-12 pt\bf 2.:}\
$\le$ is a decidable quasi-order on $S$.
Moreover, $\le$ is a well-quasi-order on the subset $\uc{C}$
(where $\uc{C} = \{s \in S\,|\, \exists c \in C.\, s \ge c\}$).
\item\noindent{\hskip-12 pt\bf 3.:}\
$\rightarrow \, =\, \rightarrow_A \cup \rightarrow_B$
\item\noindent{\hskip-12 pt\bf 4.:}\
$\rightarrow_A \,\subseteq\, S \times S$ is a monotone (w.r.t. $\le$) 
transition relation on $S$.
\item\noindent{\hskip-12 pt\bf 5.a.:}\
$\rightarrow_B \,\subseteq\, \uc{C} \times \uc{C}$ 
is a monotone (w.r.t. $\le$) transition relation on $\uc{C}$.
\item\noindent{\hskip-12 pt\bf 5.b.:}\ 
For every finite set $X \subseteq \uc{C}$ we have that the finitely many
minimal elements of the upward-closed set ${\it Pre}_{\rightarrow_B}(\uc{X})$ are effectively 
constructible.
\item\noindent{\hskip-12 pt\bf 6.a.:}\
${\it Pre}_{\rightarrow_A}^*(F)$ is upward-closed and decidable.
\item\noindent{\hskip-12 pt\bf 6.b.:}\
The finitely many minimal elements of
${\it Pre}_{\rightarrow_A}^*(F) \cap \uc{C}$ are effectively
constructible.
\item\noindent{\hskip-12 pt\bf 7.a.:}\
For any finite set $U \subseteq \uc{C}$, 
the set ${\it Pre}_{\rightarrow_A}^*(\uc{U})$ is decidable.
\item\noindent{\hskip-12 pt\bf 7.b.:}\
For any finite sets $U,X \subseteq \uc{C}$,
it is decidable if
$\overline{\uc{X}} \cap {\it Pre}_{\rightarrow_A}^*(\uc{U}) 
\cap \uc{C} \neq \emptyset$.
(In other words, it is decidable if
$\exists z \in (\overline{\uc{X}} \cap \uc{C}).\, z \rightarrow_A^* \uc{U}$.)
\end{enumerate}
\end{defi}

\noindent Note that in 7.a and 7.b the set 
${\it Pre}_{\rightarrow_A}^*(\uc{U})$ is not necessarily
constructible, because $\le$ is not a well-quasi-order on $S$.
Note also that $F$ is not necessarily upward-closed.

\begin{thm}\label{thm:phase}
If $(S,C,\le,\rightarrow, \rightarrow_A, \rightarrow_B,{\it init}, F)$
satisfies the abstract phase construction requirements of
Def.~\ref{def:phase},
then the problem ${\it init} \rightarrow^* F$ is decidable.
\end{thm}
\begin{proof}
By Def.~\ref{def:phase} (cond. 3), 
we have ${\it init} \rightarrow^* F$ iff
(1) ${\it init} \rightarrow_A^* F$, or
(2) ${\it init} \rightarrow_A^* (\rightarrow_B \rightarrow_A^*)^+ F$.

Condition (1) can be checked directly, by Def.~\ref{def:phase} (cond. 6.a).

In order to check condition (2), we first construct a sequence of 
minimal finite sets $U_k \subseteq \uc{C}$ for $k=1,2,\dots$ such that
$\uc{U_k} = \{s \in S \,|\, \exists j:1 \le j \le k.\ s (\rightarrow_B \rightarrow_A^*)^{j} F\}$
and show that this sequence converges.

First we construct the minimal finite set $U_1' \subseteq \uc{C}$ s.t.
$\uc{U_1'} = {\it Pre}_{\rightarrow_A}^*(F) \cap \uc{C}$.
This is possible by conditions 6.a and 6.b of Def.~\ref{def:phase}.
Then we construct the minimal finite set $U_1 \subseteq \uc{C}$ s.t.
$\uc{U_1} = {\it Pre}_{\rightarrow_B}(\uc{U_1'})$.
This is possible by conditions 5.a and 5.b of Def.~\ref{def:phase}.

For $k=1,2,\dots$ we repeat the following steps.
\begin{iteMize}{$\bullet$}
\item
Given the finite set $U_k \subseteq \uc{C}$, we construct 
the minimal finite set $U_{k+1}' \subseteq \uc{C}$ s.t.
$\uc{U_{k+1}'} = {\it Pre}_{\rightarrow_A}^*(\uc{U_k}) \cap \uc{C}$.
This is possible because of Theorem~\ref{thm:GVJ}, which
we instantiate as follows.
Let $\Omega = \uc{C}$ and $V = {\it Pre}_{\rightarrow_A}^*(\uc{U_k}) \cap \uc{C}$.
Using the conditions from Def.~\ref{def:phase} we have the following:
By condition 2, $\le$ is a decidable well-quasi-order
on $\uc{C}$. By condition 4, $V = {\it Pre}_{\rightarrow_A}^*(\uc{U_k}) \cap \uc{C}$
is upward-closed, since $\rightarrow_A$ is monotone.
By conditions 7.a and 2, $V$ is decidable,
and by condition 7.b the question $\overline{\uc{X}} \cap V \neq \emptyset$
is decidable. Thus, by Theorem~\ref{thm:GVJ}, the finitely many
minimal elements of $V$, i.e., the set $U_{k+1}'$, are effectively 
constructible. 
\item
Given $U_{k+1}'$, we construct the minimal finite set 
$U_{k+1}'' \subseteq \uc{C}$ s.t.
$\uc{U_{k+1}''} = {\it Pre}_{\rightarrow_B}(\uc{U_{k+1}'})$.
This is possible by conditions 5.a and 5.b of Def.~\ref{def:phase}.

Then let $U_{k+1}$ be the finite set of minimal elements
of $U_{k+1}'' \cup U_k$.
\end{iteMize}
The sequence $\uc{U_1}, \uc{U_2}, \dots$ is a monotone-increasing
sequence of upward-closed subsets of $\uc{C}$, where $U_k$ is the finite set of minimal elements
of $\uc{U_k}$. This sequence converges, because $\le$ is a well-quasi-order
on $\uc{C}$ by condition 2 of Def.~\ref{def:phase}.
Therefore, we get $\uc{U_n} = \uc{U_{n+1}}$ for some finite index $n$.
Since $\le$ is only assumed to be a quasi-order (instead of an order) the
finite sets of minimal elements $U_n$ and $U_{n+1}$ representing $\uc{U_n}$ 
and $\uc{U_{n+1}}$ are not necessarily the same. However, we can still check
whether $\uc{U_n} = \uc{U_{n+1}}$, and thus detect
convergence, because $U_n$ and $U_{n+1}$ are finite and
$\le$ is decidable by condition 2 of Def.~\ref{def:phase}.

We obtain $\uc{U_n} = \{s \in S \,|\, s (\rightarrow_B \rightarrow_A^*)^* F\}$,
because transition $\rightarrow_B$ is only enabled in $\uc{C}$
by Def.~\ref{def:phase} (cond. 5.a).

Finally, by Def.~\ref{def:phase} (cond. 7.a) we can do the final
check whether ${\it init} \in {\it Pre}_{\rightarrow_A}^*(\uc{U_n})$
and thus decide condition (2). 
\end{proof}

In the following section we use
Theorem~\ref{thm:phase} to solve the optimal cost problem for PTPN.
However, it also has many other applications, when used with different 
instantiations.

\begin{remark}
Theorem~\ref{thm:phase} can be used
to obtain a simple proof of decidability of the coverability problem for Petri nets with one
inhibitor arc. Normal Petri net transitions are described by $\trans{}_A$,
while the inhibited transition is described by $\trans{}_B$.
(This uses the decidability of the normal Petri net reachability problem 
\cite{Mayr:SIAM84} to prove conditions 7.a and 7.b).

A different instantiation could be used to show the decidability of the
reachability problem for generalized classes of lossy FIFO-channel systems \cite{AbJo:lossy:IC,CeFiPu:unreliable:IC}
where, e.g., an extra type of transition $\trans{}_B$ is only enabled when 
some particular channel is empty.
\end{remark}

\section{The Main Result}\label{sec:main}

In this section we state the main results on the decidability 
and complexity of the cost-threshold problem.

\subsection{The Lower Bound}

We show that the cost-threshold problem for PTPN is computationally
at least as hard as two other known problems.
\begin{enumerate}[(1)]
\item
The reachability problem for Petri nets with one inhibitor arc.
\item
The control-state reachability problem for timed Petri nets (TPN)
without any cost model.
\end{enumerate}

The first item establishes a connection between
PTPN and Petri nets with one inhibitor arc. 
This is interesting in itself, but it only yields a relatively weak
EXPSPACE lower bound (via the EXPSPACE lower bound for reachability 
in ordinary Petri nets \cite{Lipton:YALE76}).
The second item shows that the problem is 
$F_{\omega^{\omega^\omega}}$-hard in the fast growing hierarchy,
by using a recent result from \cite{HSS:LICS2012}.
In particular this means that the cost-threshold problem for PTPN
is non primitive recursive.

\begin{defi}\label{def:ian}
Petri nets with one inhibitor arc \cite{Reinhardt:inhibitor,Bonnet-mfcs11}
are an extension of Petri nets. 
They contain a special {\em inhibitor arc} that prevents a 
certain transition from
firing if a certain place is nonempty.

Formally, a Petri net with an inhibitor arc is described by a tuple
$N=\iantuple$ where $\inhibitor$ describes a modified firing rule
for transition $\inhibitortransition$: it can fire only if 
$\inhibitorplace$ is empty.
\begin{enumerate}[$\bullet$]
\item $\states$ is a finite set of control-states
\item $\places$ is a finite set of places
\item $\transitions$ is a finite set of transitions.
Every transition $\transition \in \transitions$ has the form
$\transition = (q_1,q_2,I,O)$ where $q_1,q_2 \in \states$ and
$I,O \in \msetsover{P}$.
\end{enumerate}
Let $(q,\marking) \in \states\times\msetsover{\places}$ be a configuration of $N$. 
\begin{iteMize}{$\bullet$}
\item
If $\transition \in \transitions - \{\inhibitortransition\}$
then $\transition = (q_1,q_2,I,O) \in \transitions$ is enabled at
configuration $(q,\marking)$ iff
$q=q_1$ and $\marking \ge I$.
\item
If $\transition = \inhibitortransition$
then $\transition = (q_1,q_2,I,O) \in \transitions$ is enabled at
configuration $(q,\marking)$ iff $q=q_1$ and $\marking \ge I$
and $\marking(\inhibitorplace)=0$.
\end{iteMize}
Firing $\transition$ yields the new
configuration $(q_2,\marking')$ where
$\marking' = \marking - I + O$.

The reachability problem for Petri nets with 
one inhibitor arc is decidable \cite{Reinhardt:inhibitor,Bonnet-mfcs11}.
\end{defi}

Now we describe a polynomial time reduction from the 
reachability problem for Petri nets with 
one inhibitor arc
to the cost-threshold problem for PTPN.

\begin{lem}\label{thm:lowerbound}
Let $\iantuple$ be a Petri net with one inhibitor arc
with initial configuration $(\initstate, \lst{})$
and final configuration $(\finalstate, \lst{})$.

One can construct in polynomial time 
a PTPN $\tuple{\states',\places',\transitions',\costfun}$
with initial configuration $\initconf = \tuple{\initstate,\lst{}}$
and set of final configurations $\finalconfs = \{(q'_{\it fin}, M)\ |\ M \in
\msetsover{(\places\times\nnreals)}\}$ s.t. 
$(\initstate, \lst{}) \trans{*} (\finalstate, \lst{})$ iff
$\optcostof{\initconf, \finalconfs}=0$.
\end{lem}
\begin{proof}
Let $\iantuple$ be a Petri net with one inhibitor arc
with initial configuration $(\initstate, \lst{})$
and final configuration $(\finalstate, \lst{})$.
We construct in polynomial time a PTPN $\tuple{\states',\places',\transitions',\costfun}$
with initial configuration $\initconf = \tuple{\initstate,\lst{}}$
and set of final configurations $\finalconfs = \{(q'_{\it fin}, M)\ |\ M \in
\msetsover{(\places\times\nnreals)}\}$ s.t. 
$(\initstate, \lst{}) \trans{*} (\finalstate, \lst{})$ iff
$\inf\{\costof{\comp} \,|\, \initconf \compto{\comp} \finalconfs\}=0$.

Let $\states' = \states \cup \{q'_{\it fin}, q^1_{\it wait}, 
q^2_{\it wait}\}$.
Let $\places' = \places \cup \{p^1_{\it wait}, p^2_{\it wait}\}$.
We define $\costfun(\place)=1$ for every $\place \in \places$,
$\costfun(\place)=0$ for $\place \in \places' - \places$,
and
$\costfun(\transition')=0$ for $\transition' \in \transitions'$.
In order to define the transitions, we need a function that 
transforms multisets of places into multisets over $\places \times\Intervals$
by annotating them with time intervals.
Let $\lst{\place_1, \dots, \place_n} \in \msetsover\places$
and $\interval \in \Intervals$. Then
$\annotate(\lst{\place_1, \dots, \place_n}, \interval) =
\lst{(\place_1,\interval), \dots, (\place_n,\interval)} \in \msetsover{(\places\times\Intervals)}$.

For every transition $\transition \in \transitions - \{\inhibitortransition\}$
with $\transition = (q_1,q_2,I,O)$ we have a transition
$\transition' = (q_1,q_2,I',O') \in \transitions'$ where
$I' = \annotate(I \cap \msetsover{(\places -\{\inhibitorplace\})},[0:\infty))
+ \annotate(I \cap \msetsover{\{\inhibitorplace\}}, [0:0])$
and 
$O' = \annotate(O, [0:0])$. I.e., the age
of the input tokens from $\inhibitorplace$ must be zero and for the other
input places the age does not matter. The transitions always output tokens of age
zero.
Instead of $\inhibitortransition=(q_1^i, q_2^i, I^i, O^i) \in \transitions$
with the inhibitor arc $\inhibitor$, 
we have the following transitions in $\transitions'$:
$(q_1^i, q^1_{\it wait}, \annotate(I^i, [0:\infty)), \lst{(p^1_{\it wait},[0:0])})$
and 
$(q^1_{\it wait}, q_2^i, \lst{(p^1_{\it wait},(0:1])}, \annotate(O, [0:0]))$.
This simulates $\inhibitortransition$ in two steps while enforcing an
arbitrarily small, but nonzero, delay. This is because the token on place
$p^1_{\it wait}$ needs to age from age zero to an age $>0$.
If $\inhibitorplace$ is empty then this yields a faithful simulation of a step
of the Petri net with one inhibitor arc. Otherwise, the tokens on
$\inhibitorplace$ will age to a nonzero age and can never be consumed in the
future. I.e., a token with nonzero age on $\inhibitorplace$ will always stay
there and indicate an unfaithful simulation.

To reach the set of final configurations $\finalconfs$, we add the
following two transitions:
$(q_{\it fin}, q^2_{\it wait}, \lst{}, \lst{(p^2_{\it wait},[0:0])})$
and
$(q^2_{\it wait}, q_{\it fin}', \lst{(p^2_{\it wait},[1:1])}, \lst{})$.
This enforces a delay of exactly one time unit at the end of the computation,
i.e., just before reaching $\finalconfs$.

If $(\initstate, \lst{}) \trans{*} (\finalstate, \lst{})$ in the Petri net
with one inhibitor arc, then for every $\epsilon >0$ there is a computation
$\initconf \compto{\comp} (q_{\it fin}, \lst{})$ in the PTPN which
faithfully simulates it and has $\costof{\comp} < \epsilon$, 
because the enforced delays can be made
arbitrarily small. The final step to $\finalconfs = \{(q'_{\it fin}, M)\ |\ M \in
\msetsover{(\places\times\nnreals)}\}$ takes one time unit, but
costs nothing, because there are no tokens on cost-places.
Thus $\optcostof{\initconf, \finalconfs} =
\inf\{\costof{\comp} \,|\, \initconf \compto{\comp} \finalconfs\}=0$.

On the other hand, if $\optcostof{\initconf, \finalconfs} 
=\inf\{\costof{\comp} \,|\, \initconf \compto{\comp} \finalconfs\}=0$ 
then the last step from $q_{\it fin}$ to $q_{\it fin}'$ must have
taken place with no tokens on places in $\places$. In particular,
$\inhibitorplace$ must have been empty. Therefore, the PTPN did a faithful
simulation of a computation 
$(\initstate, \lst{}) \trans{*} (\finalstate, \lst{})$ in the Petri net
with one inhibitor arc, i.e., the transition $\inhibitortransition$ was
only taken when $\inhibitorplace$ was empty.
Thus $(\initstate, \lst{}) \trans{*} (\finalstate, \lst{})$.
\end{proof}

Lemma~\ref{thm:lowerbound} implies that
even a special case of the cost-threshold problem for PTPN,
namely the question $\optcostof{\initconf, \finalconfs} = 0$,
is at least as hard as the 
reachability problem for Petri nets with one inhibitor arc.
The exact complexity of this problem is not known, but it is
at least as hard as reachability in standard Petri nets 
(without any inhibitor arc).
The exact complexity of the reachability problem for 
standard Petri nets is not known either, but an
EXPSPACE lower bound has been shown in \cite{Lipton:YALE76}.

The connection between the cost-threshold problem for PTPN
and the control-state reachability problem for TPN is very easy to show.

\begin{lem}\label{lem:lower2}
The control-state reachability problem for a timed Petri net (TPN)
can be reduced in polynomial time to the
question $\optcostof{\initconf, \finalconfs} = 0$ for 
a PTPN.
\end{lem}
\begin{proof}
We construct the PTPN by extending the TPN with a cost 
function that assigns cost zero to all transitions and places.
We define $\initconf$ to be the initial configuration of the TPN
and the set $\finalconfs$ by the target control-state of the TPN.
If the target control-state is reachable in the TPN
then $\optcostof{\initconf, \finalconfs} = 0$ in the PTPN,
otherwise $\optcostof{\initconf, \finalconfs}$ is undefined.
\end{proof}

It has been shown in \cite{HSS:LICS2012} (Corollary 6.) that the
control-state reachability problem for timed Petri nets (TPN)
is $F_{\omega^{\omega^\omega}}$-hard in the fast growing hierarchy.
By Lemma~\ref{lem:lower2}, we obtain the following theorem.

\begin{thm}
The cost-threshold problem for PTPN is 
$F_{\omega^{\omega^\omega}}$-hard in the fast growing hierarchy,
and thus non primitive recursive.
\end{thm}

\subsection{The Upper Bound}

Here we state the main computability result of the paper. Its proof
refers to several auxiliary lemmas that will be shown in the following sections.

\begin{thm}
Consider a PTPN $\ptpn=\ptpntuple$
with initial configuration $\initconf = \tuple{\initstate,\lst{}}$
and set of final configurations 
$\finalconfs = \{(q_{\it fin}, M)\ |\ M \in
\msetsover{(\places\times\nnreals)}\}$.\\
Then 
$\optcostof{\initconf, \finalconfs}$
is computable.
\end{thm}
\begin{proof}
$\optcostof{\initconf, \finalconfs} 
= \inf\{\costof{\comp} \,|\, \initconf \compto{\comp} \finalconfs\}
= \inf\{\costof{\comp'} \,|\, \aptpn(\initconf) \compto{\comp'}
\aptpn(\finalconfs)\}$,
by Theorem~\ref{thm:samecost-PTPN-APTPN}.
Thus it suffices to
consider the computations 
$\aptpn(\initconf) \compto{\comp'} \aptpn(\finalconfs)$
of the corresponding A-PTPN. In particular,
$\optcostof{\initconf, \finalconfs}
\in \nat$, provided that it exists.

To compute this value, it suffices to solve the cost-threshold problem 
for any given threshold $\threshold\in\nat$, i.e., to 
decide if $\aptpn(\initconf) \compto{\comp} \aptpn(\finalconfs)$ for
some $\comp$ with $\costof{\comp} \le \threshold$.

To show this, we first decide if 
$\aptpn(\initconf) \compto{\comp} \aptpn(\finalconfs)$ for
any $\comp$ (i.e., reachability). 
This can be reduced to the cost-threshold problem by setting all place and
transition costs to zero and solving the cost-threshold problem for
$\threshold=0$. If no, then no final state is reachable and we represent this
by $\inf\{\costof{\comp} \,|\, \initconf \compto{\comp} \finalconfs\}=\infty$.
If yes, then we can find the optimal cost $\threshold$ by solving the
cost-threshold problem for threshold $\threshold=0,1,2,3,\dots$ until the answer is yes.

Now we show how to solve the cost-threshold problem.
By Lemma~\ref{lem:ac-ptpn-comp}, this question is equivalent to
a reachability problem 
$\ac_\threshold(\aptpn(\initconf)) \compto{*} 
\bigcup_{0 \le y \le \threshold} 
\ac_y(\aptpn(\finalconfs))$
in the corresponding AC-PTPN.
This reachability problem is decidable by Lemma~\ref{lem:AC-PTPN-reach}.
\end{proof}

Now we prove the remaining Lemma~\ref{lem:AC-PTPN-reach}.
For this we need some auxiliary definitions.

\begin{defi}\label{def:lessfree}
We define a partial order $\lessfree$ on AC-PTPN configurations.
Given two such configurations 
$\beta = (q_\beta, (b_{-m} \dots b_{-1}, b_0, b_1 \dots b_n))$
and
$\gamma = (q_\gamma, (c_{-m'} \dots c_{-1}, c_0, c_1 \dots c_{n'}))$
we have
$\beta \lessfree \gamma$ iff
$q_\beta = q_\gamma$
and there exists a strictly monotone function
$f: \{-m,\dots,n\} \mapsto \{-m', \dots, n'\}$ where $f(0) = 0$ s.t.
\begin{enumerate}[(1)]
\item
$c_{f(i)} - b_i \in \msetsover{(\freeplaces\times \zeroto{\maxval +1})}$,
for $-m \le i \le n$.
\item
$c_j \in \msetsover{(\freeplaces\times \zeroto{\maxval +1})}$, if
$\not\exists i \in \{-m,\dots,n\}.\, f(i)=j$.
\end{enumerate}
(Intuitively, $\gamma$ is obtained from $\beta$ by adding tokens on
free-places, while the tokens on cost-places are unchanged.)
In this case, if 
$\alpha = (q_\beta, (c_{-m'}-b_{f^{-1}(-m')}, \dots, c_{-1} - b_{f^{-1}(-1)},
c_0 - b_0, c_1-b_{f^{-1}(1)}, \dots, c_{n'}-b_{f^{-1}(n')}))$ 
then we write $\alpha \lessplus \beta = \gamma$.
(Note that $\alpha$ is not uniquely defined, because it depends on the choice
of the function $f$. However one such $\alpha$ always exists and only contains 
tokens on $\freeplaces$.)

The partial order $\lesscost$ on configurations of AC-PTPN is defined
analogously with $\costplaces$ instead of $\freeplaces$, i.e., 
$\gamma$ is obtained from $\beta$ by adding tokens on cost-places.

The partial order $\lessall$ on configurations of AC-PTPN is defined
analogously with $\places$ instead of $\freeplaces$, i.e., 
$\gamma$ is obtained from $\beta$ by adding tokens on any places,
and $\lessall = \lesscost \cup \lessfree$.
\end{defi}

\begin{lem}\label{lem:lessfree}
The relations of Def.~\ref{def:lessfree} have the following properties.
\begin{enumerate}[\em(1)]
\item
$\lessfree$, $\lesscost$ and $\lessall$
are decidable quasi-orders on the set of all AC-PTPN configurations.
\item
For every AC-PTPN configuration $c$, $\lessfree$ 
is a well-quasi-order on 
the set $\uc{\{c\}} = \{s\,|\, c \lessfree s\}$ 
(i.e., here $\uparrow$ denotes the upward-closure w.r.t.
$\lessfree$).
\item
$\lessall$ is a well-quasi-order on
the set of all AC-PTPN configurations.
\end{enumerate}
\end{lem}
\newpage

\proof\hfill
\begin{enumerate}[(1)]
\item
For the decidability we note that if 
\[\beta = (q_\beta, (b_{-m} \dots b_{-1}, b_0, b_1 \dots b_n))
\qquad\hbox{and}\qquad
\gamma = (q_\gamma, (c_{-m'} \dots c_{-1}, c_0, c_1 \dots c_{n'}))\]
then  only finitely many strictly monotone functions
$f: \{-m,\dots,n\} \mapsto \{-m', \dots, n'\}$ exist with $f(0)=0$,
which need to be explored. Since addition/subtraction/inclusion on finite 
multisets are computable, the result follows.

Moreover, $\lessfree$, $\lesscost$ and $\lessall$ are quasi-orders in
the set of all AC-PTPN configurations.
Reflexivity holds trivially, and transitivity can easily be shown by composing
the respective functions $f$.
\item
Now we show that $\lessfree$ is a well-quasi-order on 
the set $\uc{\{c\}} = \{s\,|\, c \lessfree s\}$ for every AC-PTPN
configuration $c$.
Consider an infinite sequence $\beta_0, \beta_1, \dots$ of AC-PTPN
configurations where $\beta_i \in \uc{\{c\}}$ for every $i$.
It follows that there exists an infinite sequence of AC-PTPN
configurations $\alpha_0, \alpha_1, \dots$ s.t. $\alpha_i$ only contains
tokens on $\freeplaces$ and $\beta_i = c \lessplus \alpha_i$ for all $i$.
Since $\freeplaces\times \zeroto{\maxval +1}$ is finite, multiset-inclusion
is a wqo on $\msetsover{(\freeplaces\times \zeroto{\maxval +1})}$, by 
Dickson's Lemma \cite{Dickson:lemma}.
Any AC-PTPN configuration $\alpha_i$ consists of 4 parts: A control-state (out of a 
finite domain), a finite sequence over
$\msetsover{(\freeplaces\times \zeroto{\maxval +1})}$, 
an element of $\msetsover{(\freeplaces\times \zeroto{\maxval +1})}$,
and another finite sequence over $\msetsover{(\freeplaces\times \zeroto{\maxval +1})}$.
Thus, by applying Higman's Lemma \cite{Higman:divisibility} to each part,
we obtain that there must exist indices $i < j$ s.t.
$\alpha_i \lessfree \alpha_j$. 
Therefore $\beta_i = c \lessplus \alpha_i \lessfree c \lessplus \alpha_j
= \beta_j$, and thus $\lessfree$ is a wqo on $\uc{\{c\}}$.
\item
Now we show that $\lessall$ is a well-quasi-order on the set of all AC-PTPN configurations.
Consider an infinite sequence $\beta_0, \beta_1, \dots$ of AC-PTPN configurations.
Since $\places\times \zeroto{\maxval +1}$ is finite, multiset-inclusion
is a wqo on $\msetsover{(\places\times \zeroto{\maxval +1})}$, by 
Dickson's Lemma \cite{Dickson:lemma}.
Any AC-PTPN configuration consists of 4 parts: A control-state (out of a 
finite domain), a finite sequence over
$\msetsover{(\places\times \zeroto{\maxval +1})}$, 
an element of $\msetsover{(\places\times \zeroto{\maxval +1})}$,
and another finite sequence over $\msetsover{(\places\times \zeroto{\maxval +1})}$.
Thus, by applying Higman's Lemma \cite{Higman:divisibility} to each part,
we obtain that there must exist indices $i < j$ s.t.
$\beta_i \lessall \beta_j$. Thus $\lessall$ is a wqo. 
\qed
\end{enumerate}

\noindent Now we prove the required decidability of 
the AC-PTPN reachability/coverability problem. 

\begin{lem}\label{lem:AC-PTPN-reach}
Given an instance of the PTPN cost problem and a threshold 
$\threshold\in\nat$, the reachability question
$\ac_\threshold(\aptpn(\initconf)) \compto{*} 
\bigcup_{0 \le y \le \threshold} 
\ac_y(\aptpn(\finalconfs))$ in the corresponding AC-PTPN is decidable.
\end{lem}
\begin{proof}
We instantiate a structure 
$(S,C,\le,\rightarrow,\rightarrow_A,\rightarrow_B,{\it init},F)$,
show that it satisfies the requirements of Def.~\ref{def:phase},
and then apply Theorem~\ref{thm:phase}.

\noindent
Let $S$ be the set of all AC-PTPN configurations of the form
$((q,y), b_{-m} \dots b_{-1}, b_0, b_1 \dots b_n)$ where $y \le \threshold$.

\noindent
Let $C$ be the set of all AC-PTPN configurations of the form
$((q,y), b_{-m'} \dots b_{-1}, b_0, b_1 \dots b_{n'})$ where $y \le \threshold$,
and $b_i \in \msetsover{(\costplaces\times \zeroto{\maxval +1})}$
and $\sum_{j=-m'}^{n'} |b_j| \le \threshold$.
In other words, the configurations in $C$ only contain tokens on cost-places
and the size of these configurations is limited by $\threshold$.
$C$ is finite, because $\costplaces$, $\maxval$ and $\threshold$ are finite.

Let $\le := \lessfree$ of Def.~\ref{def:lessfree}, i.e., in this proof
$\uparrow$ denotes the upward-closure w.r.t.
$\lessfree$.
By Lemma~\ref{lem:lessfree},
$\le$ is decidable, $\le$ is a quasi-order on $S$, and 
$\le$ is a well-quasi-order on 
$\uc{\{c\}}$ for every AC-PTPN configuration $c$.
Therefore $\lessfree$ is a well-quasi-order on $\uc{C}$, because $C$
is finite.

Let ${\it init} := \ac_\threshold(\aptpn(\initconf))$
and $F := \bigcup_{0 \le y \le \threshold} 
\ac_y(\aptpn(\finalconfs))$. In particular, $F$ is upward-closed w.r.t.
$\lessfree$ and w.r.t. $\lessall$.
Thus conditions 1 and 2 of Def.~\ref{def:phase} are satisfied.

Let $\rightarrow_A$ be the transition relation induced by the discrete 
AC-PTPN transitions and the abstract timed AC-PTPN transitions of types 1 and 2.
These are monotone w.r.t. $\lessfree$.
Thus condition 4 of Def.~\ref{def:phase} is satisfied.

Let $\rightarrow_B$ be the transition relation induced by
abstract timed AC-PTPN transitions of types 3 and 4.
These are monotone w.r.t. $\lessfree$, but only enabled
in $\uc{C}$, because otherwise the cost would be too high.
(Remember that every AC-PTPN configuration stores the remaining allowed cost,
which must be non-negative.)
Moreover, timed AC-PTPN transitions of types 3 and 4 do not change the
number or type of the tokens in a configuration, and thus 
$\rightarrow_B \subseteq \uc{C} \times \uc{C}$.
So we have condition 5.a of Def.~\ref{def:phase}.
Condition 5.b is satisfied, because there are only finitely many token ages
$\le \maxval$ and the number and type of tokens is unchanged.

Condition 3 is satisfied, because 
$\rightarrow = \rightarrow_A \cup \rightarrow_B$ by the definition of AC-PTPN.

Now we show the conditions 6.a and 6.b.
$F$ is upward-closed w.r.t. $\lessall$ and $\rightarrow_A$
is monotone w.r.t. $\lessall$ (not only w.r.t $\lessfree$).
By Lemma~\ref{lem:lessfree}, $\lessall$ is a decidable wqo on the set of 
AC-PTPN configurations. 
Therefore, ${\it Pre}^*_{\rightarrow_A}(F)$ is upward-closed w.r.t. $\lessall$
and effectively constructible (i.e., its finitely many minimal elements 
w.r.t. $\lessall$), because the sequence
${\it Pre}^{\le i}_{\rightarrow_A}(F)$ for $i=1,2,\dots$ converges.
Let $K$ be this finite set of minimal (w.r.t. $\lessall$) 
elements of ${\it Pre}^*_{\rightarrow_A}(F)$.
We obtain condition 6.a., because $K$ is finite and $\lessall$ is decidable.
Moreover, ${\it Pre}^{*}_{\rightarrow_A}(F)$ is also 
upward-closed w.r.t. $\lessfree$. The set $C$ is a finite set of AC-PTPN 
configurations and $\uc{C}$ is the upward-closure of $C$ w.r.t. $\lessfree$.
Therefore ${\it Pre}^{*}_{\rightarrow_A}(F) \cap \uc{C}$ is upward closed
w.r.t. $\lessfree$. Now we show how to construct the finitely many minimal
(w.r.t. $\lessfree$) elements of 
${\it Pre}^{*}_{\rightarrow_A}(F) \cap \uc{C}$.
For every $k \in K$ let 
$\alpha(k) := \{k'\ |\ k' \in \uc{C}, k \lesscost k'\}$, i.e.,
those configurations which have the right control-state for $\uc{C}$,
but whose number of tokens on cost-places is bounded by $\threshold$,
and who are larger (w.r.t. $\lesscost$) than some base element in $K$.
In particular, $\alpha(k)$ is finite and constructible, because 
$\threshold$ is finite, and $\lesscost$ and $\lessfree$ are decidable.
Note that $\alpha(k)$ can be empty (if $k$ has the wrong control-state or too
many tokens on cost-places).
Let $K' := \bigcup_{k \in K} \alpha(k)$, which is finite and constructible.
We show that ${\it Pre}^{*}_{\rightarrow_A}(F) \cap \uc{C} = \uc{K'}$.
Consider the first inclusion. If $x \in \uc{K'}$ then 
$\exists k' \in K', k \in K.\, k \lesscost k' \lessfree x, k' \in \uc{C}$.
Therefore $k \lessall x$ and $x \in {\it Pre}^{*}_{\rightarrow_A}(F)$.
Also $k' \in \uc{C}$ and $k' \lessfree x$ and thus $x \in \uc{C}$.
Now we consider the other inclusion.
If $x \in {\it Pre}^{*}_{\rightarrow_A}(F) \cap \uc{C}$ then
there is a $k \in K$ s.t. $k \lessall x$. Moreover, 
the number of tokens on cost-places in $x$ is bounded by $\threshold$
and the control-state is of the form required by $\uc{C}$, because
$x \in \uc{C}$. Since, $k \lessall x$, the same holds for $k$
and thus there is some $k' \in \alpha(k)$ s.t.
$k' \lessfree x$. Therefore $x \in \uc{K'}$.
To summarize, $K'$ is the finite set of minimal (w.r.t. $\lessfree$)
elements of ${\it Pre}^{*}_{\rightarrow_A}(F) \cap \uc{C}$ and thus condition
6.b holds.

Conditions 7.a and 7.b are satisfied by Lemma~\ref{lem:7ab},
(which will be proven in Section~\ref{sec:oracle}).

So we have that our instantiation satisfies the abstract phase
construction requirements of Def.~\ref{def:phase}.
Therefore, Theorem~\ref{thm:phase} yields the decidability of the 
reachability problem ${\it init} \rightarrow^* F$, i.e.,
$\ac_\threshold(\aptpn(\initconf)) \compto{*} 
\bigcup_{0 \le y \le \threshold} 
\ac_y(\aptpn(\finalconfs))$.
\end{proof}

The remaining Lemma~\ref{lem:7ab} will be shown in Section~\ref{sec:oracle}.
Its proof uses the simultaneous-disjoint transfer
nets of Section~\ref{sec:sdtn}. 

\newpage
\section{Simultaneous-Disjoint-Transfer Nets}\label{sec:sdtn}

\noindent
\begin{defi}
Simultaneous-disjoint-transfer nets (SD-TN) \cite{AMM:LMCS2007} 
are a subclass of transfer nets \cite{Ciardo:1994}.
SD-TN subsume ordinary Petri nets.
A SD-TN $N$ is described by a tuple 
$\sdtntuple$.
\begin{enumerate}[$\bullet$]
\item $\states$ is a finite set of control-states
\item $\places$ is a finite set of places
\item $\transitions$ is a finite set of ordinary transitions.
Every transition $\transition \in \transitions$ has the form
$\transition = (q_1,q_2,I,O)$ where $q_1,q_2 \in \states$ and
$I,O \in \msetsover{P}$.
\item ${\it Trans}$ describes the set of simultaneous-disjoint transfer
transitions. Although these transitions can have different control-states 
and input/output places, they all share the {\em same transfer}
(thus the `simultaneous'). 
The transfer is described by the relation 
$\st \subseteq \places \times \places$, which is global for the SD-TN $N$.
Intuitively, for $(p,p') \in \st$, in a transfer every token in $p$ is moved
to $p'$.
The transfer transitions in ${\it Trans}$ have the form
$(q_1,q_2,I, O, \st)$ where 
$q_1,q_2 \in \states$ are the source and target control-state, 
$I,O \in \msetsover{\places}$ are like in a normal Petri net transition, 
and $\st \subseteq \places \times \places$ is the same global transfer
relation for all these transitions.
For every transfer transition $(q_1,q_2,I, O, \st)$ the 
following `disjointness' restrictions must be satisfied:
\begin{enumerate}[-]
\item Let $(sr,tg), (sr',tg') \in \st$. Then either $(sr,tg)=(sr',tg')$
or $|\{sr,sr',tg,tg'\}|=4$. Furthermore, $\{sr,tg\} \cap (I \cup O) = \emptyset$.
\end{enumerate}
\end{enumerate}
Let $(q,\marking) \in \states\times\msetsover{\places}$ be a configuration of $N$. 
The firing of normal transitions
$\transition \in \transitions$ is defined just as for ordinary Petri nets.
A transition $\transition = (q_1,q_2,I,O) \in \transitions$ is enabled at
configuration $(q,\marking)$ iff
$q=q_1$ and $\marking \ge I$.
Firing $\transition$ yields the new
configuration $(q_2,\marking')$ where
$\marking' = \marking - I + O$.

A transfer transition $(q_1,q_2,I, O, \st) \in {\it Trans}$ is 
enabled at $(q,\marking)$ iff
$q=q_1$ and $\marking \ge I$. Firing it yields the new
configuration $(q_2,\marking')$ where 
\[
\begin{array}{ll}
\marking'(p) = \marking(p)-I(p)+O(p)  & \mbox{if $p \in I \cup O$} \\
\marking'(p) = 0       & \mbox{if $\exists p'.\, (p,p') \in \st$}\\
\marking'(p) = \marking(p)+\marking(p')  & \mbox{if $(p',p) \in \st$} \\
\marking'(p) = \marking(p)    & \mbox{otherwise}
\end{array}
\]
The restrictions above ensure that these cases are disjoint.
Note that after firing a transfer transition all source places of transfers are empty,
since, by the restrictions defined above, a place that is a source of a
transfer can neither be the target of another transfer, nor receive any tokens
from the output of this transfer transition.
\end{defi}

\begin{thm}\label{thm:SD-TN-reachability}
The reachability problem for SD-TN is decidable, and has the same complexity
as the reachability problem for Petri nets with one inhibitor arc.
\end{thm}
\begin{proof}
We show that the reachability problem for SD-TN is polynomial-time
reducible to the reachability problem for Petri nets with 
one inhibitor arc (see Def.~\ref{def:ian}), and vice-versa.

For the first direction consider an SD-TN $N = \sdtntuple$,
with 
initial configuration $(q_0,\marking_0)$ and
final configuration $(q_f, \marking_f)$.
We construct a Petri net with one inhibitor arc
$N' = {\tuple{\states',\places',\transitions',\inhibitor}}$
with initial configuration $(q_0',\marking_0')$ and
final configuration $(q_f', \marking_f')$ s.t.
$(q_0,\marking_0) \trans{*} (q_f, \marking_f)$ in $N$ iff
$(q_0',\marking_0') \trans{*} (q_f', \marking_f')$ in $N'$.

Let $S := \{sr\ |\ (sr,tg) \in \st\}$ be the set of source-places of
transfers. We add a new place $\inhibitorplace$ to $\places'$ and modify the transitions
to obtain the invariant that for all reachable configurations 
$(q,\marking)$ in $N'$ we have 
$\marking(\inhibitorplace) = \sum_{sr \in S} \marking(sr)$.
Thus for every transition $\transition = (q_1,q_2,I,O) \in \transitions$ in $N$
we have a transition $\transition' = (q_1,q_2,I',O') \in \transitions'$ in $N'$
where $I'(\inhibitorplace) = \sum_{sr \in S} I(sr)$
and $O'(\inhibitorplace) = \sum_{sr \in S} O(sr)$.
For all other places $\place$ we have $I'(\place) = I(\place)$
and $O'(\place)=O(\place)$. This suffices to ensure the invariant, because
no place in $S$ is the target of a transfer.

To simulate a transfer transition $(q_1,q_2,I, O, \st) \in {\it Trans}$,
we add another control-state $q^i$ to $\states'$,
another place $p(q_2)$ to $\places'$ and a transition
$(q_1,q^i,I',O'+\{p(q_2)\})$ to $\transitions'$,  
where $I',O'$ are derived from $I,O$ as above.
Moreover, for every pair $(sr, tg) \in \st$ 
we add a transition $(q^i,q^i,\{sr,\inhibitorplace\}, \{tg\})$.
This allows to simulate the transfer by moving the tokens from the source to
the target step-by-step. The transfer is complete when all source places are
empty, i.e., when $\inhibitorplace$ is empty.
Finally, we add a transition $\inhibitortransition = (q^i, q_2, \{p(q_2)\}, \{\})$
and let the inhibitor arc be $\inhibitor$. I.e., we can only return to $q_2$
when $\inhibitorplace$ is empty and the transfer is complete. We return to
the correct control-state $q_2$ for this transition, because the last step is
only enabled if there is a token on $p(q_2)$.

So we have $\states' = \states \cup \{q^i\}$,
$\places' = \places \cup \{\inhibitorplace\} \cup \{p(q)\,|\, q \in \states\}$
and $\transitions'$ is derived from $\transitions$ as described above.
We let $q_0' = q_0$, $q_f' = q_f$ and 
$\marking_0'(\inhibitorplace) = \sum_{p \in S} \marking_0(p)$,
$\marking_f'(\inhibitorplace) = \sum_{p \in S} \marking_f(p)$
and $\marking_0'(p) = \marking_0(p)$ and $\marking_f'(p) = \marking_f(p)$
for all places $p \in \places$ and $\marking_0'(p(q)) = \marking_f'(p(q)) = 0$.
Note that, by definition of SD-TN, source-places and target-places 
of transfers are disjoint. Therefore, the condition on the inhibitor arc enforces 
that all transfers are done completely (i.e., until $\inhibitorplace$ is 
empty, and thus all places in $S$ are empty) and therefore the simulation is faithful.
Thus we obtain $(q_0,\marking_0) \trans{*} (q_f, \marking_f)$ in $N$ iff
$(q_0',\marking_0') \trans{*} (q_f', \marking_f')$ in $N'$, as required.
Since the reachability problem for Petri nets with one inhibitor arc
is decidable \cite{Reinhardt:inhibitor,Bonnet-mfcs11}, we obtain the decidability of the 
reachability problem for SD-TN.

Now we show the reverse reduction.
Consider a Petri net with one inhibitor arc
$N = \iantuple$
with initial configuration $(q_0,\marking_0)$ and
final configuration $(q_f, \marking_f)$.
We construct an SD-TN
$N' = {\tuple{\states',\places',\transitions',\transfer}}$
with initial configuration $(q_0',\marking_0')$ and
final configuration $(q_f', \marking_f')$ s.t.
$(q_0,\marking_0) \trans{*} (q_f, \marking_f)$ iff
$(q_0',\marking_0') \trans{*} (q_f', \marking_f')$.

Let $\states' = \states$, $\places' = \places \cup \{\place_x\}$
where $\place_x$ is a new place, and
$\transitions' = \transitions - \{\inhibitortransition\}$.
Let $\inhibitortransition = (q_1,q_2,I,O)$.
In $N'$, instead of $\inhibitortransition$, we have the 
${\it Trans} = \{(q_1,q_2,I, O, \st)\}$ where 
$\st = \{(\inhibitorplace, \place_x)\}$.
Unlike in $N$, in $N'$ the inhibited transition can fire even if
$\inhibitorplace$ is nonempty. However, in this case the contents of
$\inhibitorplace$ are moved to $\place_x$ where they stay forever.
I.e., we can detect an unfaithful simulation by the fact that $\place_x$ 
is nonempty.
Let $q_0' = q_0$, $q_f' = q_f$, $\marking_0'(p_x) = 0$,
$\marking_f(p_x) = 0$ and $\marking_0'(p) = \marking_0(p)$
and $\marking_f'(p) = \marking_f(p)$ for all other places $p$.
Thus we get $(q_0,\marking_0) \trans{*} (q_f, \marking_f)$ in $N$ iff
$(q_0',\marking_0') \trans{*} (q_f', \marking_f')$ in $N'$, as required.
Therefore, the reachability problem for SD-TN is polynomially equally hard as the
reachability problem for Petri nets with one inhibitor arc. 
\end{proof}

The following corollary shows decidability of a slightly generalized
reachability problem for SD-TN, which we will need in the proofs of the
following sections.

\begin{cor}\label{cor:SD-TN-reachability}
Let $N$ be an SD-TN and $F$ a set of SD-TN configurations, which is defined by a boolean 
combination of finitely many constraints of the following forms.
\begin{enumerate}[\em(1)]
\item control-state = $q$ (for some state $q \in Q$)
\item exactly $k$ tokens on place $p$ (where $k \in \nat$)
\item at least $k$ tokens on place $p$ (where $k \in \nat$)
\end{enumerate}
Then the generalized reachability problem 
$(q_0,\marking_0) \strans F$ is decidable.
\end{cor}
\begin{proof}
First, the boolean formula can be transformed into disjunctive normal form
and solved separately for each clause. Every clause is a conjunction of
constraints of the types above. This problem can then be reduced to the basic
reachability problem for a modified SD-TN $N'$ and then solved by 
Theorem~\ref{thm:SD-TN-reachability}.
One introduces a new final control-state $q'$ and adds a construction that
allows the transition from $F$ to $(q',\{\})$ if and only if the constraints
are satisfied. For type (2) one adds a transition that consumes exactly $k$
tokens from place $p$. For type (3) one adds a transition that consumes exactly $k$
tokens from place $p$, followed by a loop which can consume arbitrarily many
tokens from place $p$.
We obtain $(q_0,\marking_0) \strans F$ in $N$ iff 
$(q_0,\marking_0) \strans (q',\{\})$ in $N'$. 
Decidability follows from Theorem~\ref{thm:SD-TN-reachability}.
\end{proof}

\section{Encoding AC-PTPN Computations by SD-TN}\label{sec:oracle}
\noindent
In this section, we fix
an AC-PTPN $\tpn$, described by the tuple $\ptpntuple$
and the cost-threshold $\threshold$.
We use the partial order $\le := \lessfree$ on AC-PTPN
configurations; see Def.~\ref{def:lessfree}.
We describe an encoding of the configurations of $\tpn$ as words over 
an alphabet $\Sigma$.
We define 
$\Sigma:=
\left(\places\times\zeroto{\maxval +1}\right)
\cup
\left(\states\times\setcomp{y}{0\leq y\leq v}\right)
\cup
\set{\#,\$}
$, i.e., the members of
$\Sigma$ are elements of
$\places\times\zeroto{\maxval +1}$, the control-states of $\tpn$,
and the two ``separator'' symbols $\#$ and $\$$.
For a multiset $\mset=[a_1,\ldots,a_n]\in\msetsover{\left(\places\times\zeroto{\maxval +1}\right)}$,
we define the encoding $\wordencodingof{\mset}$ to be the word
$a_1\cdots a_n\in\left(\places\times\zeroto{\maxval +1}\right)^*$.
For a word $\word=\mset_1\cdots\mset_n\in
\wordsover{\left(\msetsover{\left(\places\times\zeroto{\maxval +1}\right)}\right)}$, 
we define
$\wordencodingof{\word}:=\wordencodingof{\mset_n}\#\cdots\#\wordencodingof{\mset_1}$,
i.e., it consists of the reverse concatenation of the
encodings of the individual multisets, separated by $\#$.
For a  marking $\marking=\tuple{\word_1,\mset,\word_2}$,
we define
$\wordencodingof{\marking}:=\wordencodingof{\word_2}\$\wordencodingof{\mset}\$\wordencodingof{\word_1}$.
In other words, we concatenate the encoding of the components in reverse order:
first $\word_2$ then $\mset$ and finally $\word_1$, separated by $\$$. 
Finally for a configuration $\conf=\tuple{\tuple{\state,y},\marking}$,
we define $\wordencodingof{\conf}:=\tuple{\state,y} \wordencodingof{\marking}$, 
i.e., we append the pair $(\state,y)$ in front of the encoding of $\marking$.
The function $\wordencodingof{}$ is extended from configurations to sets of
configurations in the standard way.
We call a finite automaton $\automaton$ over $\Sigma$ 
a {\em configuration-automaton} 
if whenever $w \in L(\automaton)$ then $w = \wordencoding(\conf)$ for
some AC-PTPN configuration $\conf$.

\begin{lem}\label{lem:build_aut_uc_finite}
Given any finite set $\confs$ of AC-PTPN configurations,
one can construct a configuration-automaton $\automaton$ s.t.
$L(\automaton) = \wordencodingof{\uc{\confs}}$.
\end{lem}
\begin{proof}
For every $\conf \in \confs$ we construct an automaton 
$\automaton_\conf$ s.t. $L(\automaton_\conf) = \wordencodingof{\uc{\{\conf\}}}$.
Remember that here the upward-closure is taken w.r.t. $\lessfree$.
Let $\conf = ((q,y),b_{-m}\dots b_{-1}, b_0, b_1 \dots b_n)$.
We have $b_i = [b_i^1, \dots, b_i^{j(i)}]$ where
$b_i^k \in \places\times\zeroto{\maxval +1}$.
Let $\Sigma_1 = \freeplaces\times\zeroto{\maxval +1}$, i.e., only tokens on
free-places can be added in the upward-closure.
Let $L_1 = (\Sigma_1^+ \#)^*$. The language $L_1$ describes encodings of sets
of tokens on free-places. Many such sets of tokens can be added during
the upward closure w.r.t. $\lessfree$.
Let $w_i = b_i^1 \dots b_i^{j(i)}$ , i.e.,
$w_i$ is an encoding of $b_i$.
Let $L_2 = L_1 w_{-m}\Sigma_1^* \# L_1
w_2 \Sigma_1^* \# L_1 \dots w_{-1}\Sigma_1^* (\# L_1)^*$.
So $L_2$ encodes the upward closure w.r.t. $\lessfree$ of the
part $b_{-m}\dots b_{-1}$ of the configuration $\conf$.
Let $L_3 = w_{-0}\Sigma_1^*$.
So $L_3$ encodes the upward closure w.r.t. $\lessfree$ of the
part $b_0$ of the configuration $\conf$.
Let $L_4 = L_1 w_{1}\Sigma_1^* \# L_1
w_2 \Sigma_1^* \# L_1 \dots w_{n}\Sigma_1^* (\# L_1)^*$.
So $L_4$ encodes the upward closure w.r.t. $\lessfree$ of the
part $b_{1}\dots b_{n}$ of the configuration $\conf$.
Then $L(\automaton_\conf) = (q,y) L_2 \$ L_3 \$ L_4 = \wordencodingof{\uc{\{\conf\}}}$.

Finally, $L(\automaton) = \bigcup_{\conf \in \confs} L(\automaton_\conf)
= \wordencodingof{\uc{\confs}}$.
\end{proof}

\begin{lem}\label{lem:build_aut_universal}
We can construct a configuration-automaton $\automaton$ s.t.
$L(\automaton) = \wordencoding(S)$, where $S$ is the set of all
configurations of a given AC-PTPN.
\end{lem}
\begin{proof}
Let $\Sigma_1 = \{(q,y)\,|\, q\in Q, 0 \le y \le \threshold\}$
and
$\Sigma_2 = \places\times\zeroto{\maxval +1}$.
Let $L_1 = \Sigma_2^*$ and
$L_2 = L_1 (\# \Sigma_2^+)^*$
and 
$L_3 = L_2 \$ L_1 \$ L_2$.
Then the language of $\automaton$ is
$\Sigma_1 L_3$, which is a regular language over $\Sigma$.
\end{proof}

\begin{lem}\label{lem:7ab}
Consider an instance of the PTPN cost problem, a given threshold 
$\threshold\in\nat$, and a structure 
$(S,C,\le,\rightarrow, \rightarrow_A, \rightarrow_B,{\it init}, F)$,
instantiated as in Lemma~\ref{lem:AC-PTPN-reach}. 

Then conditions 7.a and 7.b. of Def.~\ref{def:phase} are decidable.
\end{lem}
\proof~
\begin{enumerate}[\hbox to8 pt{\hfill}]  
\item\noindent{\hskip-12 pt\bf 7.a:}\
Consider a configuration $\conf$.
We can trivially
construct a configuration-automaton
$\automaton$ s.t. $L(\automaton) =\{\wordencodingof{\conf}\}$.
Thus the question $\conf\in {\it Pre}_{\rightarrow_A}^*(\uc{U})$
can be decided by applying Lemma~\ref{lem:automata_oracle}
to $\automaton$ and $U$.
\item\noindent{\hskip-12 pt\bf 7.b:}\
Consider finite sets of AC-PTPN configurations 
$U,X \subseteq \uc{C}$. 
By Lemma~\ref{lem:build_aut_uc_finite},
we can construct configuration-automata $\automaton_1, \automaton_2$ with
$L(\automaton_1) = \wordencodingof{\uc{X}}$ and
$L(\automaton_2) = \wordencodingof{\uc{C}}$.
Furthermore, by Lemma~\ref{lem:build_aut_universal},
we can construct a configuration-automaton $\automaton_3$
with $L(\automaton_3) = \wordencodingof{S}$.
Therefore, by elementary operations on finite automata,
we can construct a configuration-automaton $\automaton_4$
with $L(\automaton_4) = \overline{L(\automaton_1)} \cap L(\automaton_3) \cap L(\automaton_2)$,
and we obtain that $L(\automaton_4) = \wordencodingof{\overline{\uc{X}} \cap \uc{C}}$.
Note that the complement operation on words is not the same as the 
complement operation on the set of AC-PTPN configurations.
Thus the need for intersection with $\automaton_3$.
The question 
$\exists z \in (\overline{\uc{X}} \cap \uc{C}).\, z \rightarrow_A^* \uc{U}$ of
7.b can be decided by applying Lemma~\ref{lem:automata_oracle} to 
$\automaton_4$ and $U$.
\qed
\end{enumerate}

\begin{lem}\label{lem:automata_oracle}
Given a configuration-automaton $\automaton$,
$C$ as in Lemma~\ref{lem:AC-PTPN-reach},
and a finite set $U\subseteq \uc{C}$,
it is decidable if there exists some
AC-PTPN configuration 
$\initconf \in \wordencoding^{-1}(L(\automaton))$ s.t. $\initconf \rightarrow_A^* \uc{U}$.
\end{lem}
\begin{proof}
The idea is to translate the AC-PTPN into an SD-TN which simulates its
computation. The automaton $\automaton$ is also encoded into the SD-TN and
runs in parallel. $\automaton$ outputs an encoding of $\initconf$, a nondeterministically
chosen initial AC-PTPN configuration from $L(\automaton)$.
Since the SD-TN cannot encode sequences, it cannot store the order information in
the sequences which are AC-PTPN configurations. Instead this is encoded into
the behavior of $\automaton$, which outputs parts of the configuration $\initconf$
`just-in-time' before they are used in the computation (with exceptions;
see below).
Several abstractions are used to unify groups of tokens with different
fractional parts, whenever the PTPN is unable to distinguish them.
AC-PTPN timed transitions of types 1 and 2 are encoded as SD-TN transfer
transitions, e.g., all tokens with integer age advance to an age with a small
fractional part. Since this operation must affect all tokens, it cannot
be done by ordinary Petri net transitions, but requires the
simultaneous-disjoint transfer of SD-TN. Another complication is that the
computation of the AC-PTPN might use tokens (with high fractional part) from
$\initconf$, which the automaton $\automaton$ has not yet produced.
This is handled by encoding a `debt' on future outputs of $\automaton$ in
special SD-TN places. These debts can later be `paid back' by outputs of
$\automaton$ (but not by tokens created during the computation).
At the end, the computation must reach an encoding of a configuration in
$\uc{U}$ and all debts must be paid. This yields a reduction to a reachability
problem for the constructed SD-TN, which is decidable by
Theorem~\ref{thm:SD-TN-reachability}.

We devote the rest of the section to give the details of the proof.

We show the lemma for the case where $U$ is a singleton
$\set{\finalconf}$.
The result follows from the fact that $U$ is finite and that
$\uc{U}=\cup_{\conf\in U}\uc{\conf}$.
We will define an SD-TN 
$\sdtn=\tuple{\states^{\sdtn},\places^{\sdtn},\transitions^{\sdtn},\transfer^{\sdtn}}$,
a finite set $\initsdtnconfs$ of (initial) configurations, 
and a finite set of (final) $\omega$-configurations $\finalsdtnconfs$ such that
$\exists\initsdtnconf\in\initsdtnconfs.\;
\exists\finalsdtnconf\in\finalsdtnconfs.\;
\initsdtnconf\movesto{*}\finalsdtnconf$ in $\sdtn$ iff
there is
a $\initconf \in \wordencoding^{-1}(L(\automaton))$ s.t. $\initconf \rightarrow_A^* \uc{U}$.
The result then follows immediately from Theorem~\ref{thm:SD-TN-reachability}
(and Corollary~\ref{cor:SD-TN-reachability}).
Let
$\finalconf=\tuple{\tuple{\state_{\it fin},y_{\it fin}},\finalmarking}$ where
$\finalmarking$ is of the form
$\tuple{\mset_{-m}\cdots\mset_{-1},\mset_0,\mset_1\cdots\mset_n}$ and
$\mset_i$ is of the form
$\tuple{\tuple{\place_{i1},k_{i1}},\ldots,\tuple{\place_{in_i},k_{in_i}}}$
for $i:-m\leq i\leq n$.
Let the finite-state automaton $\automaton$ be of the form 
$\tuple{\states^{\cal A},\transitions^{\cal A},\state^{\cal A}_0,\fstates^A}$
where $\states^{\cal A}$ is the set of states,
$\transitions^{\cal A}$ is the transition relation,
$\state^{\cal A}_0$ is the initial state, and
$\fstates^{\cal A}$ is the set of final states.
A transition in $\transitions^{\cal A}$ is of the form
$\tuple{\state_1,a,\state_2}$ where
$\state_1,\state_2\in\states^{\cal A}$ and 
$a\in\left(\places\times\zeroto{\maxval +1}\right)
\cup
(\states\times\setcomp{y}{0\leq v\leq y_{\it init}})
\cup
\set{\sep,\$}$.
We write $\state_1\movesto{a}\state_2$ to denote that
$\tuple{\state_1,a,\state_2}\in\transitions^{\cal A}$.
During the operation of $\sdtn$, we will run the automaton
$\automaton$ ``in parallel'' with $\tpn$.
%
%
During the course of the simulation, the automaton
$\automaton$ will generate the encoding of a configuration
$\initconf$.
We know that such an encoding consists of 
a control-state $\tuple{\initstate,y_{\it init}}$ followed by the
encoding of a  marking
$\initmarking$, say of the form
$\tuple{\cmset_{-m'}\cdots\cmset_{-1},\cmset_0,\cmset_1\cdots\cmset_{n'}}$.
Notice that $\automaton$ may output the encoding of any marking
in its language, and therefore the values of $m'$ and $n'$
are not a priori known.

To simplify the presentation, we introduce a number of conventions for the description
of $\sdtn$.
First we define a set $\vars$ of variables (defined below), 
where each variable $\var\in\vars$ ranges over a finite domain $\domof{\var}$.
A control-state $\state$ then is a mapping that assigns, to each variable
$\var\in\vars$, a value in $\domof{\var}$, i.e., $\state(\var)\in\domof{\var}$.
Consider a state $\state$, variables $\var_1,\ldots,\var_k$ where $\var_i\neq\var_j$ if $i\neq j$,
and values $\val_1,\ldots,\val_k$ where $\val_i\in\domof{\var_i}$ for all $i:1\leq i\leq k$.
We use $\state[\var_1\assigned\val_1,\ldots,\var_k\assigned\val_k]$ to denote that
state $\state'$ such that $\state'(\var_i)=\val_i$ for all $i:1\leq i\leq k$,
and $\state'(\var)=\state(\var)$ if $\var\not\in\set{\var_1,\ldots,\var_k}$.
Furthermore, we introduce a set of \emph{transition generators}, where each
transition generator $\transgen$ characterizes a (finite) set $\denotationof{\transgen}$
of transitions in $\sdtn$.
A transition generator $\transgen$ is a tuple
$\tuple{\precondof{\transgen},
\postcondof{\transgen},\inputsof{\transgen},\outputsof{\transgen}}$,
where
\begin{iteMize}{$\bullet$}
\item
$\precondof{\transgen}$ is a set 
$\set{\var_1=\val_1,\ldots,\var_k=\val_k}$,
where $\var_i\in\vars$ and $\val_i\in\domof{\var_i}$ for all $i:1\leq i\leq k$.
\item
$\postcondof{\transgen}$ is a set 
$\set{\var'_1\assigned\val'_1,\ldots,\var'_\ell\assigned\val'_\ell}$,
where $\var'_i\in\vars$ and $\val'_i\in\domof{\var'_i}$ for all $i:1\leq i\leq \ell$.
\item
$\inputsof{\transgen},\outputsof{\transgen}\in \msetsover{\left(\places^\sdtn\right)}$.
\end{iteMize}
The set $\denotationof\transgen$ contains all transitions of the form
$\tuple{\state_1,\state_2,I,O}$ where
\begin{iteMize}{$\bullet$}
\item
$\state_1(\var_i)=\val_i$ for all $i:1\leq i\leq k$.
\item
$\state_2=\state_1[\var'_1\assigned\val'_1,\ldots,\var'_\ell\assigned\val'_\ell]$.
\item
$I=\inputsof{\transgen}$, and $O=\outputsof{\transgen}$.
\end{iteMize}
In the constructions we will define a set $\transgens$ of transition generators
and define $\transitions^\sdtn:=\cup_{\transgen\in\transgens}\denotationof\transgen$.

Below we will define the components
$\states^{\sdtn}$, $\places^{\sdtn}$, $\transitions^{\sdtn}$, and $\transfer^{\sdtn}$
in the definition of $\sdtn$, together with
the set $\initsdtnconfs$ and
configuration $\finalsdtnconf$.

\medskip
\noindent{\bf The set $\states^{\sdtn}$} is,
as mentioned above,  defined in terms of a set $\vars$ of variables.
The set $\vars$ contains the following elements:
\begin{iteMize}{$\bullet$}
\item
$\mode$ indicates the \emph{mode} of the simulation.
More precisely, a computation of $\sdtn$ will consist of three phases, namely an
{\it initialization}, a {\it simulation}, and a {\it final phase}.
Each phase is divided into a number of sub-phases referred
to as {\it modes}.
\item
A variable $\nstate$, with $\domof{\nstate}=\states$, that
stores the current control-state $\state_\tpn$.
\item
A variable $\astate$, with $\domof{\astate}=\states_{\automaton}$, that
stores the current state of $\automaton$.
\item
A variable $\fstate{i,j}$ with $\domof{\fstate{i,j}}=\set{\true,\false}$,
for each $i:-m\leq i\leq n$ and $1\leq j \leq n_i$.
During the simulation phase, the systems tries to cover all the tokens
in the multisets of $\finalmarking$.
Intuitively, $\fstate{i,j}$ is a flag that indicates whether the token
$\tuple{\place_{i,j},k_{i,j}}$ has been covered.
\item
A variable $\coverflag$ that has one of the values $\on$ or $\off$.
The covering of tokens in $\finalmarking$ occurs only during certain phases 
of the simulation.
This is controlled by the value of the variable $\coverflag$.

\item
A variable $\coverindex$ with
$-m\leq\coverindex\leq n$ gives the next multiset whose tokens 
are to be covered.
\item
For each $\place\in\places$ and $k:0\leq k\leq\maxval+1$, we have a variable
$\rdebt{\place}{k}$, whose use and domain are explained below.
During the simulation, we will need to use tokens that
have still not been generated by $\automaton$.
To account for these tokens, we will implement a ``debt scheme''
in which tokens are used first, and then ``paid back'' by
tokens that are later generated by $\automaton$.
The variable $\rdebt{\place}{k}$ keeps track of the number of tokens
$\tuple{\place,k}$ that have been used on read arcs
(the debt on tokens consumed in input operations are managed
through specific places described later.)
For a place $\place$ and a transition $\transition$,
let $\maxread(\place,\transition)$ be the number of read arcs
between $\place$ and $\transition$.
Define $\maxread:=\max_{\place\in\places,\transition\in\transitions}\maxread(\place,\transition)$.
Then, $\domof{\rdebt{\place}{k}}=\set{0,\ldots,\maxread}$.
The  definition of the domain reflects the fact the largest amount of debt
that we will generate due to tokens raveling through read arcs is bounded by
$\maxread$.
\end{iteMize}

\medskip
\noindent{\bf The set $\places^{\sdtn}$}
contains the following places:
\begin{iteMize}{$\bullet$}
\item
For each 
$\place\in\places$ and $k:0\leq k\leq\maxval+1$, 
the set $\places^{\sdtn}$ contains the place
$\zeroplace{\place}{k}$.
The number of tokens in $\zeroplace{\place}{k}\in\places^{\sdtn}$ 
reflects (although it may be not exactly equal to)
the number of tokens
in $\place\in\places$ whose ages have zero fractional parts.
\item
For each 
$\place\in\places$ and $k:0\leq k\leq\maxval+1$, 
the set $\places^{\sdtn}$ contains the places
$\lowplace{\place}{k}$ and $\highplace{\place}{k}$.
These places play the same roles as above
for tokens with ages that have \emph{low} (close to $0$) resp. \emph{high}
(close to $1$) fractional parts.
\item
For each $\place\in\places$ and $0\leq k\leq\maxval+1$,
the set $\places^{\sdtn}$ contains the place
$\inputdebtplace{\place}{k}$.
The place represents the amount of debt due to tokens
$\tuple{\place,k}$ traveling through input arcs.
There is a priori no bound on the amount of debt on such tokens.
Hence, this amount is stored in places (rather than in variables
as is the case for read tokens.)
\end{iteMize}

\medskip
\noindent{\bf The Set $\initsdtnconfs$}
contains all configurations $\tuple{\initstate^\sdtn,\initmarking^\sdtn}$
satisfying the following conditions:
\begin{iteMize}{$\bullet$}
\item
$\initstate^\sdtn(\mode)=\initmode$.
The initial mode is $\initmode$.

\item
$\initstate^\sdtn(\astate)=\state^{\cal A}_0$.
The automaton $\automaton$
is simulated starting from its initial state $\state^{\cal A}_0$.
\item
$\initstate^\sdtn(\fstate{i,j})=\false$ 
for all $i:-m\leq i\leq n$ and $1\leq j \leq n_i$.
Initially we have not covered any tokens in $\finalmarking$.
\item
$\initstate^\sdtn(\rdebt{\place}{k})=0$
for all $\place\in\places$ and $k:0\leq k\leq\maxval+1$.
Initially, we do not have any debts due to read tokens.
\item
$\initmarking^\sdtn(\place)=0$ for all places $\place\in\places^\sdtn$.
Initially, all the places of $\sdtn$ are empty.
\end{iteMize}
Notice that the variables $\coverflag$ and $\coverindex$ 
are not restricted so $\coverflag$ may be $\on$ or $\off$
and  $\coverindex$ may have any value $-m\leq\coverindex\leq n$.
Although $\nstate$ is not restricted either, 
its value will be defined in the first step of the simulation
(see below.)

Next, we explain how $\sdtn$ works.
In doing that, we also introduce all the members of the set $\transitions^\sdtn$.

\medskip
\noindent{\bf Initialization}
In the initialization phase
the SD-TN $\sdtn$ reads the initial control-state and then fills in the places according to $\initmarking$.
From the definition of the encoding of a configuration, we know that the automaton
$\automaton$ outputs a pair $\tuple{\state,y}$
 in its first transition.
The first move of $\sdtn$ is to store this pair in its control-state.
Thus, for each transition
$\state_1\movesto{\tuple{\state,y}}\state_2$ in $\automaton$ 
where $\state\in\states$ and $1\leq y \leq y_{\it init}$,
the set $\transgens$ contains $\transgen$ where:
\begin{iteMize}{$\bullet$}
\item
$\precondof{\transgen}=\set{\mode=\initmode,\astate=\state_1}$.
\item
$\postcondof{\transgen}=\set{\mode\assigned\initlowmode,\nstate\assigned\tuple{\state,y},\astate\assigned\state_2}$.
\item
$\inputsof{\transgen}=\emptyset$.
\item
$\outputsof{\transgen}=\set{\lowplace{\place}{k}}$.
\end{iteMize}
In other words, once $\sdtn$ has input the initial control-state, 
it enters a new mode $\initlowmode$.
In mode $\initlowmode$, we read the multisets
$\cmset_{1}\cdots\cmset_{m}$ that represent tokens with low fractional parts.
The system starts running $\automaton$ one step at a time, generating the elements
of $\cmset_m$ (that are provided by $\automaton$.)
When it has finished generating all the tokens in $\cmset_m$,
it moves to the next multiset, generating
the multisets one by one in the reverse order
finishing with $\cmset_{1}$.
We distinguish between two types of such tokens depending on
how they will be used in the construction.
More precisely, such a token is 
either {\it consumed} when firing transitions during the simulation phase
or  used for {\it covering} the multisets in $\finalmarking$.
A token (of the form $\tuple{\place,k}$), used for consumption,
is put in a place $\lowplace{\place}{k}$.
Recall that the relation $\amovesto{}$  in $\tpn$ is insensitive to the order of 
the fractional parts that are small (fractional parts of the tokens
in $\cmset_1,\ldots,\cmset_{n'}$.)
Therefore, tokens in  $\cmset_1,\ldots,\cmset_{n'}$
that have identical places $\place$ and identical integer parts $k$
will all be put in the same place $\lowplace{\place}{k}$.
Formally, for each transition
$\state_1\movesto{\tuple{\place,k}}\state_2$ in $\automaton$,
the set $\transgens$ contains $\transgen$ where:
\begin{iteMize}{$\bullet$}
\item
$\precondof{\transgen}=\set{\mode=\initlowmode,\astate=\state_1}$.
\item
$\postcondof{\transgen}=\set{\astate\assigned\state_2}$.
\item
$\inputsof{\transgen}=\emptyset$.
\item
$\outputsof{\transgen}=\set{\lowplace{\place}{k}}$.
\end{iteMize}
Each time a new multiset $\cmset_j$ is read from $\automaton$, 
the system decides whether 
it may be (partially) used for covering the next multiset $\mset_i$
in $\finalmarking$.
This decision is made by checking the value
of the component $\coverflag$.
if $\coverflag=\off$ then
the tokens are only used for consumption during the simulation phase.
However, if $\coverflag=\on$ then 
the tokens generated by $\automaton$ can also be used to cover those in $\finalmarking$.
The multiset currently covered is given by the value
of the component  $\coverindex$.
More precisely, if $\coverindex=i$ for some $i:1\leq i \leq n$ then
(part of) the multiset $\cmset_j$ that is
currently being generated by $\automaton$ ($j:1\leq j \leq n'$)
may be used to cover (part of) the multiset $\mset_i$.
At this stage, we only cover tokens with low fractional parts (those in the 
multisets $\mset_1,\ldots,\mset_n$.)
When using tokens for covering, the order on the fractional parts 
of tokens \emph{is} relevant.
The construction takes into consideration different aspects
of this order as follows:
\begin{iteMize}{$\bullet$}
\item
According to the definition of the ordering $\lessfree$,
the tokens in a given multiset $\cmset_j$ 
may only be used to cover those in one and the same multiset (say $\mset_i$.)
This also agrees with the observation that
the tokens represented in $\cmset_j$ correspond to tokens in the original TPN 
that have identical fractional parts (the same applies to $\mset_i$.)
In fact, if this was not case, 
then we would be using tokens with identical fractional parts (in $\cmset_j$)
to cover tokens with different fractional parts.
Analogously, the multiset $\mset_i$ can be covered only by the elements
of one multiset $\cmset_j$.
\item
If $i'<i$ then the fractional parts of the tokens represented by 
$\mset_{i'}$ are smaller than those represented by $\mset_{i}$.
The same applies to $\cmset_{j'}$ and $\cmset_{j}$ if $j'<j$.
Therefore, if $\cmset_{j}$ is used to cover $\mset_{i}$ and $j'<j$ then
 $\cmset_{j'}$ should be used to cover $\mset_{i'}$ for some $i'<i$.
Furthermore, a multiset $\cmset_{j}$ is not necessarily used to 
cover any multiset, i.e.,
all the tokens represented by $\cmset_{j}$ 
may be used for consumption during the simulation
(none of them being used for covering.)
Similarly, it can be the case that a given  $\mset_{i}$ is not covered
by any multiset  $\cmset_{j}$ (all its tokens are covered by
tokens that are generated during the simulation.)
Also, a multiset $\cmset_{j}$ may only be partially used to cover
$\mset_{i}$, i.e., some of its tokens may be used for covering
$\mset_{i}$ while some are consumed during the simulation.
Finally,  $\mset_{i}$ may only be partially covered by $\cmset_{j}$, 
i.e., some of its tokens are
covered by $\cmset_{j}$ while the rest of tokens are covered by 
tokens generated during the simulation.
\end{iteMize}
Formally, for each
$\state_1\movesto{\tuple{\place,k}}\state_2$ in $\automaton$,
$1\leq i\leq n$, $1\leq j\leq n_i$ with $\tuple{\place_{i,j},k_{i,j}}=\tuple{\place,k}$,
we add $\transgen$ to $\transgens$, where:
\begin{iteMize}{$\bullet$}
\item
$\precondof{\transgen}=\{\mode=\initlowmode,\astate=\state_1,$ $\coverflag=\on,\coverindex=i\}$.
\item
$\postcondof{\transgen}=\set{\astate\assigned\state_2,\fstate{i,j}\assigned\true}$.
\item
$\inputsof{\transgen}=\emptyset$.
\item
$\outputsof{\transgen}=\emptyset$.
\end{iteMize}
The transition sets the flag $\fstate{i,j}$ to $\true$ 
indicating that the token has now been covered.
A transition
$\state_1\movesto{\sep}\state_2$ in $\automaton$ indicates
that we have finished generating the elements of the current multiset $\cmset_j$.
If $\coverflag=\on$ then we have also finished covering tokens
in the multiset $\mset_i$.
Therefore, we decide the next multiset $i'<i$ in which which to cover tokens.
Recall that not all multisets have to be covered and 
hence $i'$ need not be equal to
$i-1$ (in fact the multisets $\mset_{i''}$ for $i'<i''<i$ 
will not be covered by the multisets in $\initmarking$.)
We also decide whether to use  $\initmset_{j-1}$ to cover  $\mset_{i'}$ or not.
In the former case, we set $\coverflag$ to $\on$, 
while in the latter case we
 set $\coverflag$ equal to $\off$.
Also, if $\coverflag=\off $ then we decide whether to use
$\cmset_{j-1}$ for covering
$\mset_{i}$ or not.
We cover these four possibilities by adding the following transition generators to $\transgens$.

(i) For each transition
$\state_1\movesto{\sep}\state_2$ in $\automaton$,
$i:1\leq i\leq n$, and $i':-m\leq i'<i$,
we add $\transgen$ where:
\begin{iteMize}{$\bullet$}
\item
$\precondof{\transgen}=\{\mode=\initlowmode,\astate=\state_1,\coverflag=\on,\coverindex=i\}$.
\item
$\postcondof{\transgen}=\{\astate\assigned\state_2,\coverindex\assigned i'\}$.
\item
$\inputsof{\transgen}=\emptyset$.
\item
$\outputsof{\transgen}=\emptyset$.
\end{iteMize}
This is the case where $\coverflag$ is $\on$ and continues to be $\on$.
Notice that no covering takes place if $\coverindex\leq 0$,
and that the new value of $\coverindex$ is made strictly smaller than the current one.

(ii)
For each transition
$\state_1\movesto{\sep}\state_2$ in $\automaton$,
and each $i,i':1\leq i'<i\leq n$,
we add $\transgen$ where:
\begin{iteMize}{$\bullet$}
\item
$\precondof{\transgen}=\{\mode=\initlowmode,\astate=\state_1,$ 
$\coverflag=\on,\coverindex=i\}$.
\item
$\postcondof{\transgen}=\{\astate\assigned\state_2,\coverflag\assigned\off,\coverindex\assigned i'\}$.
\item
$\inputsof{\transgen}=\emptyset$.
\item
$\outputsof{\transgen}=\emptyset$.
\end{iteMize}
This is the case where $\coverflag$ is $\on$ but it is turned $\off$
for the next step.

(iii)
For each transition
$\state_1\movesto{\sep}\state_2$ in $\automaton$,
we add $\transgen$ where:
\begin{iteMize}{$\bullet$}
\item
$\precondof{\transgen}=\{\mode=\initlowmode,\astate=\state_1,$ 
$\coverflag=\off\}$.
\item
$\postcondof{\transgen}=\{\astate\assigned\state_2\}$.
\item
$\inputsof{\transgen}=\emptyset$.
\item
$\outputsof{\transgen}=\emptyset$.
\end{iteMize}
This is the case where $\coverflag$ is $\off$ and continues to be $\off$.

(iv) 
For each transition
$\state_1\movesto{\sep}\state_2$ in $\automaton$,
we add $\transgen$ where:
\begin{iteMize}{$\bullet$}
\item
$\precondof{\transgen}=\{\mode=\initlowmode,\astate=\state_1,$ 
$\coverflag=\off\}$.
\item
$\postcondof{\transgen}=\{\astate\assigned\state_2,\coverflag\assigned\on\}$.
\item
$\inputsof{\transgen}=\emptyset$.
\item
$\outputsof{\transgen}=\emptyset$.
\end{iteMize}
This is the case where $\coverflag$ is $\off$ but it is turned $\on$
for the next step.

The process of generating tokens with low fractional parts
continues until we encounter a transition of the form
$\state_1\movesto{\$}\state_2$ in $\automaton$.
According to the encoding of markings, this indicates 
that we have finished generating the elements of the multisets $\cmset_1,\ldots,\cmset_n$.
Therefore, we change mode from $\initlowmode$ to
$\initzeromode$ (where we scan the multiset $\mset_0$.)
We have also to consider changing
the variables $\coverflag$ and $\coverindex$ in the 
same way as above.
Therefore, we add the following transition generators:

(i) 
For each transition $\state_1\movesto{\$}\state_2$ in $\automaton$,
$i:1\leq i\leq n$, and $i':-m\leq i'<i$,
we add $\transgen$ where:
\begin{iteMize}{$\bullet$}
\item
$\precondof{\transgen}=\{\mode=\initlowmode,\astate=\state_1,$ 
$\coverflag=\on,\coverindex=i\}$.
\item
$\postcondof{\transgen}=\{\mode\assigned\initzeromode,\astate\assigned\state_2,\coverindex\assigned i'\}$.
\item
$\inputsof{\transgen}=\emptyset$.
\item
$\outputsof{\transgen}=\emptyset$.
\end{iteMize}

(ii) 
For each transition $\state_1\movesto{\$}\state_2$ in $\automaton$,
$i:1\leq i\leq n$, and $i':-m\leq i'<i$,
we add $\transgen$ where:
\begin{iteMize}{$\bullet$}
\item
$\precondof{\transgen}=\{\mode=\initlowmode,\astate=\state_1,$ 
$\coverflag=\on,\coverindex=i\}$.
\item
$\postcondof{\transgen}=\{\mode\assigned\initzeromode,\astate\assigned\state_2,\coverflag\assigned\off,\coverindex\assigned i'\}$.
\item
$\inputsof{\transgen}=\emptyset$.
\item
$\outputsof{\transgen}=\emptyset$.
\end{iteMize}

(iii)
For each transition
$\state_1\movesto{\$}\state_2$ in $\automaton$,
we add $\transgen$ where:

\begin{iteMize}{$\bullet$}
\item
$\precondof{\transgen}=\{\mode=\initlowmode,\astate=\state_1,$ 
$\coverflag=\off\}$.
\item
$\postcondof{\transgen}=\set{\mode\assigned\initzeromode,\astate\assigned\state_2}$.
\item
$\inputsof{\transgen}=\emptyset$.
\item
$\outputsof{\transgen}=\emptyset$.
\end{iteMize}

(iv)
For each transition
$\state_1\movesto{\$}\state_2$ in $\automaton$,
we add $\transgen$ where:

\begin{iteMize}{$\bullet$}
\item
$\precondof{\transgen}=\{\mode=\initlowmode,\astate=\state_1,$ 
$\coverflag=\off\}$.
\item
$\postcondof{\transgen}=\{\mode\assigned\initzeromode,\astate\assigned\state_2,\coverflag\assigned\on\}$.
\item
$\inputsof{\transgen}=\emptyset$.
\item
$\outputsof{\transgen}=\emptyset$.
\end{iteMize}

In $\initzeromode$ the places are filled according to $\cmset_0$.
The construction is similar to the previous mode.
The only differences are that the tokens to be consumed will
be put in places $\zeroplace{\place}{k}$ 
and that no tokens are covered in $\finalmarking$.

For each transition
$\state_1\movesto{\tuple{\place,k}}\state_2$ in $\automaton$,
the set $\transgens$ contains $\transgen$ where:
\begin{iteMize}{$\bullet$}
\item
$\precondof{\transgen}=\set{\mode=\initzeromode,\astate=\state_1}$.
\item
$\postcondof{\transgen}=\set{\astate\assigned\state_2}$.
\item
$\inputsof{\transgen}=\emptyset$.
\item
$\outputsof{\transgen}=\set{\zeroplace{\place}{k}}$.
\end{iteMize}
Since the tokens are not used at this stage
for covering the multisets of $\finalmarking$,
no transition generators are added for that purpose.
Also, in contrast to tokens belonging to 
$\cmset_0,\ldots,\cmset_{n'}$
we cannot generate tokens 
belonging to $\cmset_{-m'},\ldots,\cmset_{-1}$
during the initialization phase.
The reason is that, in the former case,
we only need to keep track of the order of multisets whose tokens are used for covering
(the ordering of the fractional parts in tokens used for consumption is not relevant.)
Since the number $n$ is given a priori in the construction (the marking
$\finalmarking$ is a parameter of the problem), we need only to keep track
of tokens belonging to at most $n$ different multisets.
This does not hold in the case of the latter tokens, 
since the order of the multisets to which the tokens belong is
relevant also in the case of tokens that will be consumed.
Since $m'$ 
is not a priori bounded, we postpone the generation of 
these tokens to the simulation phase, where we generate these tokens from
$\automaton$ ``on demand'': 
each time we perform a timed transition, we allow the 
$\highplace{\place}{k}$ tokens with the highest 
fractional part to be generated.
This construction is made more 
precise in the description of the simulation phase.

The mode $\initzeromode$ is concluded when we the next transition of 
$\automaton$ is labeled with $\$$.
This means that we have finished inputting the last multiset
$\mset_0$.
We now move on to the simulation phase.

For each transition of the form
$\state_1\movesto{\$}\state_2$ in $\automaton$, we add
$\transgen$ to $\transgens$ where:
\begin{iteMize}{$\bullet$}
\item
$\precondof{\transgen}=\set{\mode=\initzeromode,\astate=\state_1}$.
\item
$\postcondof{\transgen}=\set{\mode\assigned\simmode,\astate\assigned\state_2}$.
\item
$\inputsof{\transgen}=\emptyset$.
\item
$\outputsof{\transgen}=\emptyset$.
\end{iteMize}

\medskip
\noindent{\bf Simulation.}
The simulation phase consists of simulating a sequence of
transitions, each of which is either discrete, of type $1$, or of type $2$.
Each type $2$ transition is preceded by at least one type $1$ transition.
Therefore, from $\simmode$ we next perform a discrete or a type $1$  transition.
The (non-deterministic) choice is made
using the transition generators $\transgen_1$ and $\transgen_2$ 
where:
\begin{iteMize}{$\bullet$}
\item
$\precondof{\transgen_1}=\set{\mode=\simmode}$.
\item
$\postcondof{\transgen_1}=\set{\mode\assigned\discmode}$.
\item
$\inputsof{\transgen_1}=\emptyset$.
\item
$\outputsof{\transgen_1}=\emptyset$.
\end{iteMize}
\begin{iteMize}{$\bullet$}
\item
$\precondof{\transgen_2}=\set{\mode=\simmode}$.
\item
$\postcondof{\transgen_2}=\set{\mode\assigned\typeoneamode}$.
\item
$\inputsof{\transgen_2}=\emptyset$.
\item
$\outputsof{\transgen_2}=\emptyset$.
\end{iteMize}

\medskip
\noindent{\bf Discrete Transitions.}
A discrete transition $\transition=\tuple{\state_1,\state_2,\inputs,\reads,\outputs}$ 
in $\tpn$ is simulated by a set of transitions in $\sdtn$.
In defining this set, we take into consideration several aspects
of the simulation procedure as follows:
\begin{iteMize}{$\bullet$}
\item
Basically, an interval $\interval$ 
on an arc leading from an input place $\place\in\inputs$ to
$\transition$   induces a set of transitions in $\transitions^\sdtn$; namely
transitions where there are arcs from places $\zeroplace{\place}{k}$ with $k\in\interval$, and
from places $\lowplace{\place}{k}$ and  $\highplace{\place}{k}$
with $(k+\epsilon)\in\interval$ for some
$\epsilon:0<\epsilon<1$.
An analogous construction is made for output and read places of $\transition$.
Since a read arc does not remove the token from the place, 
there is both an input arc and output arc to the corresponding transition in $\sdtn$.
\item
We recall that the tokens belonging to $\cmset_{-m'},\ldots,\cmset_{-1}$
are not generated during the initial phase, and that these tokens are gradually
introduced during the simulation phase.
Therefore, a transition may need to be fired before the 
required $\highplace{\place}{k}$-tokens have been produced by $\automaton$.
Such tokens are needed for performing both input and read operations.
In order to cover for tokens that are needed for input arcs, 
we use the set of places $\inputdebtplace{\place}{k}$
for $\place\in\places$ and $0\leq k\leq\maxval+1$.
Then, consuming a token from a place  $\highplace{\place}{k}$
may be replaced by putting a token in $\inputdebtplace{\place}{k}$.
The ``debt'' can be paid back using tokens 
that are later generated by $\automaton$.
When $\sdtn$ terminates, we require all the debt places to be empty
(all the debt has been paid back.)
Also, we need an analogous (but different) scheme for the 
read arcs.
The difference is due to the fact that
the same token may be read several times (without being consumed.)
Hence, once the debt has been introduced by the first read operation, it will not
be increased by the subsequent read operations.
Furthermore, several read operations may be covered by a (single)
input operation (a token in a place may be read several times before 
it is finally consumed through an input operations.)
To implement this, we use the variables $\rdebt{\place}{k}$.
Each time a number $r$ of tokens $\tuple{\place,k}$ are ``borrowed'' for a read operation,
we increase the value of $\rdebt{\place}{k}$ to $r$ (unless it already has a higher value.)
Furthermore, each debt taken on a token $\tuple{\place,k}$ in an input operation
subsumes a debt performed on the same token $\tuple{\place,k}$ in a read operation.
Therefore, the value of an old read debt is decreased by the amount
of the input debt taken during the current transition.
In a similar manner to input debts, the read debt
is later paid back.
When $\sdtn$ terminates, we require all $\rdebt{\place}{k}$
variables to be equal to $0$
(all the read debts have been paid back.)
\item
The transition also changes the control-state of $\tpn$.
\end{iteMize}

To formally define the set of transitions in $\sdtn$ induced by discrete transitions, 
we use a number of definitions.
We define $x\monus y:=\max(y-x,0)$.
For $k\in\nat$ and an interval $\interval$, we write
$k\inn\interval$ to denote that $(k+\epsilon)\in\interval$ for some (equivalently all) $\epsilon:0<\epsilon<1$.
During the simulation phase, there are two mechanisms for simulating
the effect of a token traveling through an (input, read, or output) arc in $\tpn$,
namely, (i)  by letting a token travel from (or to) a corresponding place;  and
(ii) by ``taking debt''.
Therefore, we define a number of ``transformers'' that translate tokens 
in $\tpn$ to corresponding ones in $\sdtn$ as follows:
\begin{iteMize}{$\bullet$}
\item
$\zeroplacetransform{\place,\interval}:=\setcomp{\zeroplace{\place}{k}}{(0\leq k\leq\maxval+1) \wedge(k\in\interval)}$.
The $\tpn$-token is simulated
by a $\sdtn$-token in a
place that represent tokens
with zero fractional parts.

\item
$\lowplacetransform{\place,\interval}:=\setcomp{\lowplace{\place}{k}}{(0\leq k\leq\maxval+1) \wedge (k\inn\interval)}$.
The $\tpn$-token is simulated
by a $\sdtn$-token in a
place that represent tokens
with low fractional parts.
Notice that we use the relation $\inn$ since the fractional part of the token is not zero.
\item
$\highplacetransform{\place,\interval}:=\setcomp{\highplace{\place}{k}}{(0\leq k\leq\maxval+1) \wedge (k\inn\interval)}$.
The $\tpn$-to\-ken is simulated
by a $\sdtn$-token in a
place that represent tokens
with high fractional parts.
\item
$\inputdebttransform{\place,\interval}:=\setcomp{\inputdebtplace{\place}{k}}{(0\leq k\leq\maxval+1) \wedge (k\inn\interval)}$. %
The $\tpn$-token is simulated by taking debt on an input token.
\item
$\readdebttransform{\place,\interval}:=\setcomp{\readdebtplace{\place}{k}}{(0\leq k\leq\maxval+1) \wedge (k\inn\interval)}$.
The $\tpn$-token is simulated by taking debt on a read token.
\end{iteMize}
We extend the transformers to multisets,
so for a multiset $\mset=\lst{\tuple{\place_1,\interval_1},\ldots,\tuple{\place_{\ell},\interval_{\ell}}}$,
let 
$\zeroplacetransform{\mset}:=\!\!
\setcompdiv{\lst{\tuple{\place_1,k_1},\ldots,\tuple{\place_{\ell},k_{\ell}}}\!}
{\!\!\!\forall i:\!1\!\leq\! i\leq\ell\!:\!\tuple{\place_i,k_i}\!\in\!\!}{\zeroplacetransform{\!\place_i,\interval_i}}
$.
We extend the other definitions to  multisets analogously.

\noindent
An $\rdebtz$-mapping $\rdebtmapping$ is a function that maps each
$\rdebt{\place}{k}$ to a value in $\set{0,\ldots,\maxread}$.
In other words, the function describes the state of the debt on read
tokens.

Now, we are ready to define the transitions in $\sdtn$ that are
induced by discrete transitions in $\tpn$.
Each such a transition is induced by a number of objects, namely:
\begin{iteMize}{$\bullet$}
\item
A transition $\transition = \tuple{\state_1,\state_2,\inputs,\reads,\outputs}\in\transitions$.
This is the transition in $\tpn$ that is to be simulated in $\sdtn$.
\item
The current remaining cost 
$y:\costof\transition\leq y\leq y_{\it init}$.
The remaining cost has to be at least as large as the
cost of the transition to be fired.
\item
An $\rdebtz$-mapping $\rdebtmapping$ describing the current debt on read tokens.
\item
Multisets 
$\inputs^{\it Zero},\inputs^{\it Low},\inputs^{\it High},\inputs^{\it Debt}$
where 
\[\inputs=\inputs^{\it Zero}+\inputs^{\it Low}+\inputs^{\it High}+\inputs^{\it Debt}.\]
Intuitively, the tokens traveling through arcs of $\transition$
are covered by four types of tokens:
\begin{iteMize}{$-$}
\item
$\inputs^{\it Zero}$: $\tpn$-tokens that will be transformed
into $\sdtn$-tokens in places encoding ages with zero fractions parts.
\item
$\inputs^{\it Low}$: $\tpn$-tokens that will be transformed
into $\sdtn$-tokens in places encoding ages with low fractions parts.
\item
$\inputs^{\it High}$: $\tpn$-tokens that will be transformed
into $\sdtn$-tokens in places encoding ages with high fractions parts.
\item
$\inputs^{\it Debt}$: $\tpn$-tokens that will be covered by taking debt.
\end{iteMize}
\item{\sloppy
Multisets $\reads^{\it Zero},\reads^{\it Low},\reads^{\it High},\reads^{\it Debt}$
where 
\[\reads=\reads^{\it Zero}+\reads^{\it Low}+
\reads^{\it High}+\reads^{\it Debt}.\]
The roles of these multisets are similar to the above.}
\item
Multisets $\outputs^{\it Zero},\outputs^{\it Low},\outputs^{\it High}$
where 
\[\outputs=\outputs^{\it Zero}+\outputs^{\it Low}+\outputs^{\it High}.\]
The roles of the multisets $\outputs^{\it Zero},\outputs^{\it Low},\outputs^{\it High}$
are similar to their counter-parts above.
\end{iteMize}
For each such collection of objects
(i.e., for each $\transition$, $0\!\leq\! y\!\leq\! y_{\it init}$, $\alpha$,
$\inputs^{\it Zero},\inputs^{\it Low},\inputs^{\it High},\inputs^{\it Debt}$,
$\reads^{\it Zero},\reads^{\it Low},\reads^{\it High},\reads^{\it Debt}$,
$\outputs^{\it Zero},\outputs^{\it Low},\outputs^{\it High}$),
we add the transition generator $\transgen$ where:
\begin{iteMize}{$\bullet$}
\item
$\precondof{\transgen}=\set{\mode=\discmode,\nstate=\tuple{\state_1,y}}\cup\alpha$,
i.e., the current mode is $\discmode$, the current state of 
$\tpn$ is $\tuple{\state_1,y}$, and the current debt on read tokens is given
by $\alpha$.
\item
$\postcondof{\transgen}=\\
\set{\mode\assigned\simmode,\nstate\assigned\tuple{\state_2,y-\costof\transition}}\cup$\\
$\setcompdiv{\rdebt{\place}{k}
\assigned\max(\alpha\monus\inputs^{\it Debt'},\reads^{\it Debt'})(\place,k)}{}
{\!\!\!(\place\in\places)\wedge (0\!\leq\! k\!\leq\!\maxval\!+\!1)}$,
where 
\begin{iteMize}{$-$}
\item
$\inputs^{\it Debt'}=\inputdebttransform{\inputs^{\it Debt}}$.
\item
$\reads^{\it Debt'}=\readdebttransform{\reads^{\it Debt}}$.
\end{iteMize}
In other words,
we change the mode back to $\simmode$, and change the control-state
of $\tpn$ to $\tuple{\state_2,y-\costof\transition}$.
The new read debts are defined as follows:
We reduce the current debt $\alpha$ 
using the new debt on input tokens $\inputs^{\it Debt'}$,
then we update the amount again using the new debt $\reads^{\it Debt'}$.
\item
$\inputsof{\transgen}=
\inputs^{\it Zero'}+\inputs^{\it Low'}+\inputs^{\it High'}+
\reads^{\it Zero'}+\reads^{\it Low'}+\reads^{\it High'}$, where
\begin{iteMize}{$-$}
\item
$\inputs^{\it Zero'}=\zeroplacetransform{\inputs^{\it Zero}}$.
\item
$\inputs^{\it Low'}=\lowplacetransform{\inputs^{\it Low}}$.
\item
$\inputs^{\it High'}=\highplacetransform{\inputs^{\it High}}$.
\item
$\reads^{\it Zero'}=\zeroplacetransform{\reads^{\it Zero}}$.
\item
$\reads^{\it Low'}=\lowplacetransform{\reads^{\it Low}}$.
\item
$\reads^{\it High'}=\highplacetransform{\reads^{\it High}}$.
\end{iteMize}
The multisets 
$\inputs^{\it Zero},\inputs^{\it Low},\inputs^{\it High}$
represent tokens that will be consumed
due to input arcs.
These tokens are distributed among places according to 
whether their fractional parts are zero, low, or high.
A similar reasoning holds for the multisets
$\reads^{\it Zero}$, $\reads^{\it Low}$, $\reads^{\it High}$.
\item
$\outputsof{\transgen}=
\outputs^{\it Zero'}+\outputs^{\it Low'}+\outputs^{\it High'}+\outputs^{\it Debt'}+
\reads^{\it Zero'}+\reads^{\it Low'}+\reads^{\it High'}$, where
\begin{iteMize}{$-$}
\item
$\outputs^{\it Zero'}=\zeroplacetransform{\outputs^{\it Zero}}$.
\item
$\outputs^{\it Low'}=\lowplacetransform{\outputs^{\it Low}}$.
\item
$\outputs^{\it High'}=\highplacetransform{\outputs^{\it High}}$.
\item
$\outputs^{\it Debt'}=\highplacetransform{\inputs^{\it Debt}}$.
\item
$\reads^{\it Zero'}=\zeroplacetransform{\reads^{\it Zero}}$.
\item
$\reads^{\it Low'}=\lowplacetransform{\reads^{\it Low}}$.
\item
$\reads^{\it High'}=\highplacetransform{\reads^{\it High}}$.
\end{iteMize}
The read multisets are defined as in the previous item.
The multisets $\outputs^{\it Zero}$, $\outputs^{\it Low}$, 
and $\outputs^{\it High}$ play the same roles
as their input and read counterparts.
The multiset $\outputs^{\it Debt'}$ represents the increase in the debt on 
read tokens.
\end{iteMize}

\medskip
\noindent{\bf Transitions of Type $1$.}
The simulation of a type $1$ transition is started when the mode is
$\typeoneamode$.
We recall that a type $1$ transition encodes that time passes so that
all  tokens of integer age in $\mset_0$ will now have a positive fractional part,
but no tokens reach an integer age.
This phase is performed in two steps.
First, in $\typeoneamode$ (that is repeated an arbitrary number of times), 
some of these tokens are used for 
covering the multisets of $\finalmarking$
in a similar manner to the previous phases.
In the second step we
change mode to $\typeonebmode$, at the same time
switching $\on$ or $\off$ the component $\coverflag$
in a similar manner to the initialization phase.
In $\typeonebmode$, the (only set) of transfer transitions
encodes the effect of passing time.
More precisely all tokens in a 
place  $\zeroplace{\place}{k}$ will be moved to the place
$\lowplace{\place}{k}$, for $k:1\leq k\leq\maxval+1$.
From $\typeonebmode$ the mode will be changed
to $\typetwoamode$.

To describe $\typeoneamode$ formally we add,
for each $i:1\leq i\leq n$, $j:1\leq j\leq n_i$,
$\place\in\places$, $k:0\leq k\leq \maxval+1$
with $\tuple{\place,k}=\tuple{\place_{i,j},k_{i,j}}$,
a transition generator $\transgen$ where:
\begin{iteMize}{$\bullet$}
\item
$\precondof{\transgen}=\{\mode=\typeoneamode,\coverflag=\on,\coverindex=i\}$.
\item
$\postcondof{\transgen}=\set{\fstate{i,j}\assigned\true}$.
\item
$\inputsof{\transgen}=\set{\zeroplace{\place}{k}}$.
\item
$\outputsof{\transgen}=\emptyset$.
\end{iteMize}

On switching to $\typeonebmode$, we  change
the variables $\coverflag$ and $\coverindex$ in 
a similar manner to the previous phases.
Therefore, we add the following transition generators:

(i) 
For each $i:1\leq i\leq n$, and $i':-m\leq i'<i$,
we add $\transgen$ where:
\begin{iteMize}{$\bullet$}
\item
$\precondof{\transgen}=\{\mode=\typeoneamode,$ 
$\coverflag=\on,\coverindex=i\}$.
\item
$\postcondof{\transgen}=\{\mode\assigned\typeonebmode,\coverflag\assigned\off,\coverindex\assigned i'\}$.
\item
$\inputsof{\transgen}=\emptyset$.
\item
$\outputsof{\transgen}=\emptyset$.
\end{iteMize}

(ii) 
For each 
$i:1\leq i\leq n$, and $i':-m\leq i'<i$,
we add $\transgen$ where
\begin{iteMize}{$\bullet$}
\item
$\precondof{\transgen}=\{\mode=\typeoneamode,$ 
$\coverflag=\on,\coverindex=i\}$.
\item
$\postcondof{\transgen}=\{\mode\assigned\typeonebmode,$
$\coverindex\assigned i'\}$.
\item
$\inputsof{\transgen}=\emptyset$.
\item
$\outputsof{\transgen}=\emptyset$.
\end{iteMize}

(iii)
We add $\transgen$ where:
\begin{iteMize}{$\bullet$}
\item
$\precondof{\transgen}=\{\mode=\typeoneamode,$
$\coverflag=\off\}$.
\item
$\postcondof{\transgen}=\{\mode\assigned\typeonebmode\}$.
\item
$\inputsof{\transgen}=\emptyset$.
\item
$\outputsof{\transgen}=\emptyset$.
\end{iteMize}

(iv)
We add $\transgen$ where:
\begin{iteMize}{$\bullet$}
\item
$\precondof{\transgen}=\{\mode=\typeoneamode,$
$\coverflag=\off\}$.
\item
$\postcondof{\transgen}=\{\mode\assigned\typeonebmode,$
$\coverflag\assigned\on\}$.
\item
$\inputsof{\transgen}=\emptyset$.
\item
$\outputsof{\transgen}=\emptyset$.
\end{iteMize}

The set of transfer transitions is defined by the transfer
transition generator $\transgen$

\begin{iteMize}{$\bullet$}
\item
$\precondof{\transgen}=\set{\mode=\typeonebmode}$.
\item
$\postcondof{\transgen}=\set{\mode\assigned\typetwoamode}$.
\item
$\inputsof{\transgen}=\emptyset$.
\item
$\outputsof{\transgen}=\emptyset$.
\item
$\stof{\transgen}=\setcompdiv{\tuple{\zeroplace{\place}{k},\lowplace{\place}{k}}}{}
{(\place\in\places)\wedge(0\leq k\leq\maxval+1)}
$.
\end{iteMize}

\medskip
\noindent{\bf Transitions of Type $2$.}
Recall that transitions of type $2$ 
encode what happens to tokens with the largest fractional parts
when an amount of time passes sufficient
for making these ages
equal to the next integer (but not larger.)
There are two sources of such tokens.
The generation of tokens according to these two sources 
divides the phase into two steps.
The first source are tokens that are currently in places
of the form $\highplace{\place}{k}$.
In $\typetwoamode$, (some of) these tokens reach the next integer, 
and are therefore moved to the corresponding
places encoding tokens with zero fractional parts.
As mentioned above,
only some (but not all) of these tokens reach the next integer.
The reason is that they are generated during the computation (not
by $\automaton$), and hence they have arbitrary fractional parts.

The second source are tokens that are provided by the automaton $\automaton$
(recall that these tokens are not generated during the initialization phase.)
In $\typetwobmode$, we run the automaton $\automaton$ one step at a time.
At each step we generate the next token
by taking a transition $\state_1\movesto{\tuple{\place,k}}\state_2$.
In fact, such a token $\tuple{\place,k}$ is used in two ways:
either it moves to the place $\zeroplace{\place}{k}$, or
it is used to pay the debt we have taken on tokens.
The debt is paid back either
(i) by removing
a token from $\inputdebtplace{\place}{k}$; or
(ii) by decrementing the value of the variable $\rdebt{\place}{k}$.
A transition $\state_1\movesto{\sep}\state_2$ means that we have 
read the last element of the current multiset.
This finishes simulating the transitions of type $1$ and $2$  and the mode is moved back
to $\simmode$ starting another iteration of the simulation phase.

Formally, we describe the movement of tokens in $\typetwoamode$
by adding,
for each $\place\in\places$ and $k:0\leq k\leq\maxval+1$, 
a transition generator $\transgen$ where:
\begin{iteMize}{$\bullet$}
\item
$\precondof{\transgen}=\set{\mode=\typetwoamode}$.
\item
$\postcondof{\transgen}=\emptyset$.
\item
$\inputsof{\transgen}=\set{\highplace{\place}{k}}$.
\item
$\outputsof{\transgen}=\set{\zeroplace{\place}{\max(k+1,\maxval+1)}}$.
\end{iteMize}
At any time, we can change mode from $\typetwoamode$ to $\typetwobmode$:
\begin{iteMize}{$\bullet$}
\item
$\precondof{\transgen}=\set{\mode=\typetwoamode}$.
\item
$\postcondof{\transgen}=\set{\mode\assigned\typetwobmode}$.
\item
$\inputsof{\transgen}=\emptyset$.
\item
$\outputsof{\transgen}=\emptyset$.
\end{iteMize}
We can also move back from $\typetwoamode$ to
$\simmode$ without letting the automaton generate any tokens:
\begin{iteMize}{$\bullet$}
\item
$\precondof{\transgen}=\set{\mode=\typetwoamode}$.
\item
$\postcondof{\transgen}=\set{\mode\assigned\simmode}$.
\item
$\inputsof{\transgen}=\emptyset$.
\item
$\outputsof{\transgen}=\emptyset$.
\end{iteMize}
We simulate $\typetwobmode$ as follows.
To describe the movement of tokens places representing
tokens with zero fractional parts we add,
for each transition
$\state_1\movesto{\tuple{\place,k}}\state_2$ in $\automaton$,
a transition generator $\transgen$ where:
\begin{iteMize}{$\bullet$}
\item
$\precondof{\transgen}=\set{\mode=\typetwobmode,\astate=\state_1}$.
\item
$\postcondof{\transgen}=\set{\astate\assigned\state_2}$.
\item
$\inputsof{\transgen}=\emptyset$.
\item
$\outputsof{\transgen}=\set{\zeroplace{\place}{k}}$.
\end{iteMize}
To describe the payment of debts on input tokens we add,
for each transition
$\state_1\movesto{\tuple{\place,k}}\state_2$ in $\automaton$,
a transition generator $\transgen$ where:
\begin{iteMize}{$\bullet$}
\item
$\precondof{\transgen}=\set{\mode=\typetwobmode,\astate=\state_1}$.
\item
$\postcondof{\transgen}=\set{\astate\assigned\state_2}$.
\item
$\inputsof{\transgen}=\set{\inputdebtplace{\place}{k}}$.
\item
$\outputsof{\transgen}=\emptyset$.
\end{iteMize}
To describe the payment of debts on read tokens we add,
for each transition
$\state_1\movesto{\tuple{\place,k}}\state_2$ in $\automaton$,
and $r:1\leq r\leq\maxread$,
a transition generator $\transgen$ where:
\begin{iteMize}{$\bullet$}
\item
$\precondof{\transgen}=\set{\mode=\typetwobmode,\astate=\state_1,\rdebt{\place}{k}=r}$.
\item
$\postcondof{\transgen}=\set{\astate\assigned\state_2,\rdebt{\place}{k}\assigned r-1}$.
\item
$\inputsof{\transgen}=\emptyset$.
\item
$\outputsof{\transgen}=\emptyset$.
\end{iteMize}

As usual, transition
$\state_1\movesto{\sep}\state_2$ in $\automaton$ indicates
means that we have 
read the last element of the current multiset.
We can now move back to the mode $\simmode$,
changing
the variables $\coverflag$ and $\coverindex$ in 
a similar manner to the previous phases.

(i)
For each transition of the form
$\state_1\movesto{\sep}\state_2$ in $\automaton$ ,
$i:1\leq i\leq n$, and $i':-m\leq i'<i$,
we add $\transgen$ where:
\begin{iteMize}{$\bullet$}
\item
$\precondof{\transgen}=\{\mode=\typetwobmode,\astate=\state_1,$ 
$\coverflag=\on,\coverindex=i\}$.
\item
$\postcondof{\transgen}=\{\mode\assigned\simmode,\astate\assigned\state_2,$
$\coverflag\assigned\off,\coverindex\assigned i'\}$.
\item
$\inputsof{\transgen}=\emptyset$.
\item
$\outputsof{\transgen}=\emptyset$.
\end{iteMize}

(ii)
For each transition $\state_1\movesto{\sep}\state_2$ in $\automaton$,
$i:1\leq i\leq n$, and $i':-m\leq i'<i$,
we add $\transgen$ where:
\begin{iteMize}{$\bullet$}
\item
$\precondof{\transgen}=\{\mode=\typetwobmode,\astate=\state_1,$ 
$\coverflag=\on,\coverindex=i\}$.
\item
$\postcondof{\transgen}=\{\mode\assigned\simmode,\astate\assigned\state_2,$
$\coverflag\assigned\on,\coverindex\assigned i'\}$.
\item
$\inputsof{\transgen}=\emptyset$.
\item
$\outputsof{\transgen}=\emptyset$.
\end{iteMize}

(iii)
For each transition
$\state_1\movesto{\sep}\state_2$ in $\automaton$,
we add $\transgen$ where:
\begin{iteMize}{$\bullet$}
\item
$\precondof{\transgen}=\{\mode=\typetwobmode,\astate=\state_1,$ 
$\coverflag=\off\}$.
\item
$\postcondof{\transgen}=\set{\mode\assigned\simmode,\astate\assigned\state_2}$.
\item
$\inputsof{\transgen}=\emptyset$.
\item
$\outputsof{\transgen}=\emptyset$.
\end{iteMize}

(iv)
For each transition
$\state_1\movesto{\sep}\state_2$ in $\automaton$,
we add $\transgen$ where:

\begin{iteMize}{$\bullet$}
\item
$\precondof{\transgen}=\{\mode=\typetwobmode,\astate=\state_1,$ 
$\coverflag=\off\}$.
\item
$\postcondof{\transgen}=\{\mode\assigned\simmode,\astate\assigned\state_2,$
$\coverflag\assigned\on\}$.
\item
$\inputsof{\transgen}=\emptyset$.
\item
$\outputsof{\transgen}=\emptyset$.
\end{iteMize}

\medskip
\noindent{\bf The Final Phase.}
From the simulation mode we can at any time enter the final
mode.
\begin{iteMize}{$\bullet$}
\item
$\precondof{\transgen}=\set{\mode=\simmode}$.
\item
$\postcondof{\transgen}=\set{\mode\assigned\finalonemode}$.
\item
$\inputsof{\transgen}=\emptyset$.
\item
$\outputsof{\transgen}=\emptyset$.
\end{iteMize}
The main tasks of the final phase are 
(i) to cover the multisets in $\finalmarking$;
and 
(ii) to continue paying back the \emph{debt} tokens (recall that the debt was partially
paid back in the simulation of type $2$ transitions.)
At the end of the final phase, we expect all tokens in $\finalmarking$ to have been
covered and all debt to have been paid back.
The final phase consists of two modes.
In $\finalonemode$ we cover the multisets in $\finalmarking$
using the tokens that have already been generated.
In $\finaltwomode$, we resume running $\automaton$ one step at a time.
The tokens generated from $\automaton$ are used both
(i) for paying back debt; and (ii) for covering the multisets 
$\mset_{-1},\ldots,\mset_{-m}$ (in that order.)

Formally, we add the following transition generators.
First, 
we continue covering the multisets $\mset_1,\ldots,\mset_n$.
For each $\place\in\places$, $1\leq i\leq n$, and $1\leq j\leq n_i$ with
$\tuple{\place_{i,j},k_{i,j}}=\tuple{\place,k}$, we add $\transgen$ where:
\begin{iteMize}{$\bullet$}
\item
$\precondof{\transgen}=\set{\mode=\finalonemode}$.
\item
$\postcondof{\transgen}=\set{\fstate{i,j}\assigned\true}$.
\item
$\inputsof{\transgen}=\lowplace{\place}{k}$.
\item
$\outputsof{\transgen}=\emptyset$.
\end{iteMize}
We cover the multiset $\mset_0$
by moving tokens from places of the form $\zeroplace{\place}{k}$.
For each $\place\in\places$ and $1\leq j\leq n_0$ with
$\tuple{\place_{0,j},k_{0,j}}=\tuple{\place,k}$, we add $\transgen$ where:
\begin{iteMize}{$\bullet$}
\item
$\precondof{\transgen}=\set{\mode=\finalonemode}$.
\item
$\postcondof{\transgen}=\set{\fstate{0,j}\assigned\true}$.
\item
$\inputsof{\transgen}=\zeroplace{\place}{k}$.
\item
$\outputsof{\transgen}=\emptyset$.
\end{iteMize}
We cover the multisets $\mset_{-1},\ldots,\mset_{-m}$
by moving tokens from places of type $\highplace{\place}{k}$.
For each $\place\in\places$, $-m\leq i\leq -1$, $1\leq j\leq n_i$ with
$\tuple{\place_{i,j},k_{i,j}}=\tuple{\place,k}$, we add $\transgen$ where:
\begin{iteMize}{$\bullet$}
\item
$\precondof{\transgen}=\set{\mode=\finalonemode}$.
\item
$\postcondof{\transgen}=\set{\fstate{i,j}\assigned\true}$.
\item
$\inputsof{\transgen}=\highplace{\place}{k}$.
\item
$\outputsof{\transgen}=\emptyset$.
\end{iteMize}
We can change mode to $\finaltwomode$:
\begin{iteMize}{$\bullet$}
\item
$\precondof{\transgen}=\set{\mode=\finalonemode}$.
\item
$\postcondof{\transgen}=\set{\mode\assigned\finaltwomode}$.
\item
$\inputsof{\transgen}=\emptyset$.
\item
$\outputsof{\transgen}=\emptyset$.
\end{iteMize}

In $\finaltwomode$, we start running $\automaton$.
The tokens can be used for paying input debts.
For each transition
$\state_1\movesto{\tuple{\place,k}}\state_2$ in $\automaton$,
we add $\transgen$ where:
\begin{iteMize}{$\bullet$}
\item
$\precondof{\transgen}=\set{\mode=\finaltwomode,\astate=\state_1}$.
\item
$\postcondof{\transgen}=\set{\astate\assigned\state_2}$.
\item
$\inputsof{\transgen}=\set{\inputdebtplace{\place}{k}}$.
\item
$\outputsof{\transgen}=\emptyset$.
\end{iteMize}
The tokens can also be used for paying read debts.
For each transition
$\state_1\movesto{\tuple{\place,k}}\state_2$ in $\automaton$,
and $k:1\leq r\leq\maxread$,
we add $\transgen$ where:
\begin{iteMize}{$\bullet$}
\item
$\precondof{\transgen}=\set{\mode=\finaltwomode,\astate=\state_1,\rdebt{\place}{k}= r}$.
\item
$\postcondof{\transgen}=\set{\astate\assigned\state_2,\rdebt{\place}{k}\assigned r-1}$.
\item
$\inputsof{\transgen}=\emptyset$.
\item
$\outputsof{\transgen}=\emptyset$.
\end{iteMize}
Finally,
the tokens can  be used for covering.
For each transition
$\state_1\movesto{\tuple{\place,k}}\state_2$ in $\automaton$,
$i:-m\leq i\leq -1$, $j:1\leq j\leq n_i$,
$\place\in\places$, $k:0\leq k\leq \maxval+1$
with $\tuple{\place,k}=\tuple{\place_{i,j},k_{i,j}}$,
we have $\transgen$ where:
\begin{iteMize}{$\bullet$}
\item
$\precondof{\transgen}=\{\mode=\finaltwomode,\coverflag=\on,\coverindex=i\}$.
\item
$\postcondof{\transgen}=\set{\fstate{i,j}\assigned\true}$.
\item
$\inputsof{\transgen}=\emptyset$.
\item
$\outputsof{\transgen}=\emptyset$.
\end{iteMize}
A transition
$\state_1\movesto{\sep}\state_2$ in $\automaton$ indicates
that we have 
read the last element of the current multiset.
We now let $\automaton$ generate the next multiset.
We change
the variables $\coverflag$ and $\coverindex$ in 
a similar manner to the previous phases.

(i)
For each transition of the form
$\state_1\movesto{\sep}\state_2$ in $\automaton$ ,
$i:-m\leq i\leq -1$, and $i':-m\leq i'<i$,
we add $\transgen$ where:
\begin{iteMize}{$\bullet$}
\item
$\precondof{\transgen}=\{\mode=\finaltwomode,\astate=\state_1,$ 
$\coverflag=\on,\coverindex=i\}$.
\item
$\postcondof{\transgen}=\{\astate\assigned\state_2,$
$\coverflag\assigned\off,\coverindex\assigned i'\}$.
\item
$\inputsof{\transgen}=\emptyset$.
\item
$\outputsof{\transgen}=\emptyset$.
\end{iteMize}

(ii)
For each transition $\state_1\movesto{\sep}\state_2$ in $\automaton$,
$i:1\leq i\leq n$, and $i':-m\leq i'<i$,
we add $\transgen$ where:
\begin{iteMize}{$\bullet$}
\item
$\precondof{\transgen}=\{\mode=\finaltwomode,\astate=\state_1,$ 
$\coverflag=\on,\coverindex=i\}$.
\item
$\postcondof{\transgen}=\{\astate\assigned\state_2,\coverindex\assigned i'\}$.
\item
$\inputsof{\transgen}=\emptyset$.
\item
$\outputsof{\transgen}=\emptyset$.
\end{iteMize}

(iii)
For each transition
$\state_1\movesto{\$}\state_2$ in $\automaton$,
we add $\transgen$ where:
\begin{iteMize}{$\bullet$}
\item
$\precondof{\transgen}=\{\mode=\finaltwomode,\astate=\state_1,$ 
$\coverflag=\off\}$.
\item
$\postcondof{\transgen}=\{\astate\assigned\state_2\}$.
\item
$\inputsof{\transgen}=\emptyset$.
\item
$\outputsof{\transgen}=\emptyset$.
\end{iteMize}

(iv)
For each transition
$\state_1\movesto{\$}\state_2$ in $\automaton$,
we add $\transgen$ where:

\begin{iteMize}{$\bullet$}
\item
$\precondof{\transgen}=\{\mode=\finaltwomode,\astate=\state_1,$ 
$\coverflag=\off\}$.
\item
$\postcondof{\transgen}=\{\astate\assigned\state_2,\coverflag\assigned\on\}$.
\item
$\inputsof{\transgen}=\emptyset$.
\item
$\outputsof{\transgen}=\emptyset$.
\end{iteMize}

\medskip
\noindent{\bf The Set $\finalsdtnconfs$}
contains all configurations $\tuple{\finalstate^\sdtn,\finalmarking^\sdtn}$
satisfying the following conditions:
\begin{iteMize}{$\bullet$}
\item
$\finalstate^\sdtn(\nstate)=\finalstate$.
The AC-PTPN is in its final control-state.
\item
$\finalstate^\sdtn(\fstate{i,j})=\true$ 
for all $i:-m\leq i\leq n$ and $1\leq j \leq n_i$.
We have covered all tokens in $\finalmarking$.
\item
$\finalstate^\sdtn(\rdebt{\place}{k})=0$
for all $\place\in\places$ and $k:0\leq k\leq\maxval+1$.
We have paid back all debts on read tokens.
\item
$\finalmarking(\inputdebtplace{\place}{k})=0$
for all $\place\in\places$ and $0\leq k\leq\maxval+1$.
We have paid back all debts on input tokens.\qedhere
\end{iteMize}
\end{proof}

\section{Undecidability for Negative Costs}\label{sec:undecidability}

The cost threshold coverability problem 
for PTPN is undecidable if negative
transition costs are allowed.

In fact, as we see from the proof of 
Theorem~\ref{thm:neg_trans_undecidability} below, the undecidability proof holds
even if the costs of  places are restricted to be non-negative
integers, and the costs of transitions are restricted to be non-positive.
Moreover, the undecidability proof does not require real-valued clocks,
but works even if clock values are natural numbers, i.e., in the
discrete-time case.

\begin{thm}\label{thm:neg_trans_undecidability}
The cost threshold problem for PTPN $\ptpn=\ptpntuple$ is undecidable.
\end{thm}

\begin{proof}
\input negproof
\end{proof}

\section{Conclusion and Extensions}\label{sec:conclusion}

\noindent
We have shown that the infimum of the costs to reach a given control-state
is computable in priced timed Petri nets with continuous time. 
This subsumes the corresponding
results for less expressive models such as priced timed automata
\cite{BouyerBBR:PTA} and priced discrete-timed Petri nets
\cite{AM:FOSSACS2009}.

For simplicity of presentation, we have used a one-dimensional cost model,
i.e., with a cost $\in \nnreals$, but our result on decidability of the
Cost-Threshold problem can trivially be generalized to a multidimensional
cost model (provided that the cost is linear in the elapsed time).
However, in a multidimensional cost model, the Cost-Optimality problem 
is not defined, because the infimum of the costs does
not exist, due to trade-offs between different components. 
E.g., one can construct a PTPN (and even a priced timed automaton) with a
2-dimensional cost where the feasible costs are 
$\{(x, 1-x)\,|\, x \in \nnreals, 0 < x \le 1\}$, i.e., with uncountably many
incomparable values.

Another simple generalization is to make token storage costs on places
dependent on the current control-state, e.g., storing one token on place $p$
for one time unit costs $2$ if in control-state $q_1$, but $3$ if in
control-state $q_2$. Our constructions can trivially be extended to handle
this.

Other extensions are much harder. If the token storage costs are not linear 
in the elapsed time then the infimum of the costs is not necessarily an
integer. In particular, the corner-point abstraction of
Section~\ref{sec:abstract_ptpn} and our translation 
to an A-PTPN problem would not work. It is an open question
how to compute optimal costs in such cases.

Finally, some extensions make the cost-problems undecidable.
As shown in Section~\ref{sec:undecidability}, reachability of a given
control-state with cost zero becomes undecidable if general integer costs 
(including negative costs, i.e., rewards) are allowed.
This negative result holds even for the simpler case of discrete-time PTPN, 
i.e., when clocks values are natural numbers.
 
If one considers the reachability problem (instead of our control-state
reachability problem) then the question is undecidable for TPN
\cite{PCVC99:TPN:Nondecidability}, even without considering costs.

\section*{Acknowledgments}
We thank the anonymous referees for their detailed comments.

\input references

\vspace{-20 pt}
\end{document}

%% file: comm.tex
\newtheorem{remark}{Remark}

\newcommand{\hide}[1]{}

\newcommand{\set}[1]{\left\{#1\right\}}

\newcommand{\tuple}[1]{\left(#1\right)}
\newcommand{\setcomp}[2]{\left\{#1\mid\; #2\right\}}
\newcommand{\setcompdiv}[3]{\left\{#1\mid\; #2\right.$ $\left. #3\right\}}

\newcommand{\denotationof}[1]{[\![#1]\!]}

\newcommand{\zeroto}[1]{[{#1}]}

\newcommand{\fract}{{\it frac}}
\newcommand{\fractof}[1]{{\it frac}\left(#1\right)}

\newcommand{\nat}{{\mathbb N}}

\newcommand{\preals}{{\mathbb R}_{> 0}}
\newcommand{\nnreals}{{\mathbb R}_{\geq 0}}

\newcommand{\states}{Q}
\newcommand{\state}{q}
\newcommand{\conf}{c}
\newcommand{\confs}{C}

\newcommand{\comp}{\pi}

\newcommand{\match}{{\it match}}
\newcommand{\Intervals}{{\it Intrv}}
\newcommand{\interval}{{\mathcal I}}

\newcommand{\oointrvl}[2]{({#1}:{#2})}

\newcommand{\mset}{b}
\newcommand{\cmset}{c}
\newcommand{\initmset}{\mset^{{\it init}}}

\newcommand{\aset}{{\it A}}

\newcommand{\lst}[1]{[{#1}]}
\newcommand{\mleq}{\leq}
\newcommand{\emptystring}{\varepsilon}

\newcommand{\tpn}{{\cal N}}
\newcommand{\ptpn}{{\cal N}}
\newcommand{\ptpntuple}{\tuple{\states,\places,\transitions,\costfun}}

\newcommand{\places}{P}
\newcommand{\costplaces}{\places_c}
\newcommand{\freeplaces}{\places_f}
\newcommand{\lessfree}{\le^f}
\newcommand{\lessplus}{{\oplus}}
\newcommand{\lesscost}{\le^c}
\newcommand{\lessall}{\le^{fc}}
\newcommand{\place}{p}

\newcommand{\transitions}{T}
\newcommand{\transition}{t}

\newcommand{\inputs}{{\it In}}

\newcommand{\outputs}{{\it Out}}
\newcommand{\reads}{{\it Read}}

\newcommand{\cost}{C}
\newcommand{\costfun}{{\it Cost}}
\newcommand{\optcostfun}{{\it OptCost}}
\newcommand{\costof}[1]{\costfun\left(#1\right)}
\newcommand{\optcostof}[1]{\optcostfun\left(#1\right)}

\newcommand{\marking}{M}

\newcommand{\trans}[1]{\stackrel{#1}{\longrightarrow}}
\newcommand{\movesto}[1]{\stackrel{#1}{\longrightarrow}}
\newcommand{\amovesto}[1]{{\stackrel{#1}{\longrightarrow}}_A}
\newcommand{\compto}[1]{\stackrel{#1}{\longrightarrow}}

\newcommand{\timedmovesto}{\longrightarrow_{\it Time}}
\newcommand{\xtimedmovesto}[1]{\stackrel{#1}\longrightarrow_{\it Time}}
\newcommand{\disc}{{\it Disc}}
\newcommand{\discmovesto}{\longrightarrow_\disc}

\newcommand{\ttrans}{\longrightarrow_{\transition}}

\newcommand{\pre}{\textit{Pre}}

\newcommand{\word}{w}

\newcommand{\maxval}{{\it cmax}}

\newcommand{\initmarking}{\marking_{\it init}}
\newcommand{\finalmarking}{\marking_{\it fin}}

\newcommand{\w}{{\it w}}
\newcommand{\z}{{\it z}}

\newcommand{\msets}[1]{{#1}^\odot}

\newcommand{\uc}[1]{{{#1}\!\uparrow}}

\newcommand{\st}{{\it ST}}

\newcommand{\strans}{\trans{*}}

\newcommand{\wleq}{\leq^\word}

\newcommand{\vecdef}{\vec}

\newcommand{\automaton}{{\cal A}}

\newcommand{\counter}{x}
\newcommand{\cmachine}{M}
\newcommand{\initstate}{\state_{\it init}}
\newcommand{\finalstate}{\state_{\it fin}}
\newcommand{\ctransition}{r}
\newcommand{\op}{{\it op}}
\newcommand{\ztest}[1]{{#1}=0?}

\newcommand{\inc}[1]{{#1}{++}}
\newcommand{\dec}[1]{{#1}{--}}

\newcommand{\initconf}{\conf_{\it init}}
\newcommand{\finalconf}{\conf_{\it fin}}
\newcommand{\finalconfs}{\confs_{\it fin}}

\newcommand{\threshold}{v}
\newcommand{\costnum}{v}
\newcommand{\aptpn}{{\it aptpn}}
\newcommand{\ac}{{\it ac}}
\newcommand{\wordencoding}{{\it enc}}
\newcommand{\wordencodingof}[1]{{\it enc}\left({#1}\right)}

\newcommand{\msetsover}[1]{\msets{#1}}
\newcommand{\wordsover}[1]{{#1}^*}

\newcommand{\fstates}{F}

\newcommand{\sep}{\#}

\newcommand{\sdtn}{{\cal T}}
\newcommand{\transfer}{{\it Trans}}
\newcommand{\inhibitorplace}{\place^i}
\newcommand{\inhibitortransition}{\transition^i}
\newcommand{\inhibitor}{(\inhibitorplace, \inhibitortransition)}
\newcommand{\sdtntuple}{\tuple{\states,\places,\transitions,\transfer}}
\newcommand{\iantuple}{\tuple{\states,\places,\transitions,\inhibitor}}
\newcommand{\annotate}{{\it annotate}}
\newcommand{\initsdtnconf}{\conf^{\sdtn}_{{\it init}}}
\newcommand{\initsdtnconfs}{\confs^{\sdtn}_{{\it init}}}

\newcommand{\finalsdtnconf}{\conf^{\sdtn}_{{\it final}}}
\newcommand{\finalsdtnconfs}{\confs^{\sdtn}_{{\it final}}}

\newcommand{\lowplace}[2]{{\tt LowPlace}\left(#1,#2\right)}

\newcommand{\highplace}[2]{{\tt HighPlace}\left(#1,#2\right)}

\newcommand{\zeroplace}[2]{{\tt ZeroPlace}\left(#1,#2\right)}

\newcommand{\off}{{\tt off}}
\newcommand{\on}{{\tt on}}

\newcommand{\inputdebtplace}[2]{{\it InputDebt}\left(#1,#2\right)}
\newcommand{\readdebtplace}[2]{{\it ReadDebt}\left(#1,#2\right)}

\newcommand{\aptpnof}[1]{{\it aptpn}\left(#1\right)}
\newcommand{\highword}{\word^{{\it high}}}
\newcommand{\lowword}{\word^{{\it low}}}
\newcommand{\zmset}{\mset_0}

\newcommand{\astate}{{\tt AState}}
\newcommand{\nstate}{{\tt NState}}
\newcommand{\fstate}[1]{{\tt FState}\left({#1}\right)}
\newcommand{\rdebtz}{{\tt RDebt}}
\newcommand{\rdebt}[2]{\rdebtz\left({#1,#2}\right)}
\newcommand{\mode}{{\tt Mode}}
\newcommand{\coverflag}{{\tt CoverFlag}}
\newcommand{\coverindex}{{\tt CoverIndex}}

\newcommand{\true}{{\it true}}
\newcommand{\false}{{\it false}}

\newcommand{\initmode}{{\tt Init}}
\newcommand{\initlowmode}{{\tt InitLow}}
\newcommand{\initzeromode}{{\tt InitZero}}
\newcommand{\simmode}{{\tt Sim}}
\newcommand{\discmode}{{\tt Disc}}
\newcommand{\typeoneamode}{{\tt Type1.1}}
\newcommand{\typeonebmode}{{\tt Type1.2}}
\newcommand{\typetwoamode}{{\tt Type2.1}}
\newcommand{\typetwobmode}{{\tt Type2.2}}
\newcommand{\finalonemode}{{\tt Final1}}
\newcommand{\finaltwomode}{{\tt Final2}}

\newcommand{\var}{{\tt x}}
\newcommand{\vars}{{\tt X}}
\newcommand{\val}{{\tt v}}
\newcommand{\domof}[1]{{\it dom}\left(#1\right)}
\newcommand{\transgen}{\theta}
\newcommand{\transgens}{\Theta}
\newcommand{\precondof}[1]{{\tt PreCond}\left(#1\right)}
\newcommand{\postcondof}[1]{{\tt PostCond}\left(#1\right)}
\newcommand{\assigned}{\leftarrow}
\newcommand{\inputsof}[1]{{\tt In}\left(#1\right)}
\newcommand{\outputsof}[1]{{\tt Out}\left(#1\right)}

\newcommand{\stof}[1]{\st\left(#1\right)}

\newcommand{\zeroplacetransform}[1]{{\it ZeroPlaceTransf}\left(#1\right)}
\newcommand{\lowplacetransform}[1]{{\it LowPlaceTransf}\left(#1\right)}
\newcommand{\highplacetransform}[1]{{\it HighPlaceTransf}\left(#1\right)}

\newcommand{\inputdebttransform}[1]{{\it InputDebtTransf}\left(#1\right)}
\newcommand{\readdebttransform}[1]{{\it ReadDebtTransf}\left(#1\right)}

\newcommand{\inn}{\rightModels}

\newcommand{\maxread}{{\it Rmax}}
\newcommand{\rdebtmapping}{\alpha}
\newcommand{\monus}{\stackrel{\bullet}{-}}

%% file: negproof.tex
\newcommand{\ctransitions}{\delta}

We prove that it is undecidable whether a given control-state can be reached 
with cost $\le 0$, through a reduction from the 
control-state reachability problem for Minsky $2$-counter machines \cite{Minsky:book}.

We recall that a $2$-counter machine $\cmachine$, operating
on two counters $\counter_0$ and $\counter_1$, is a triple
$\tuple{\states,\ctransitions,\initstate}$, 
where $\states$ is a finite set of {\it control
states}, $\ctransitions$ is a finite set of {\it transitions},
and $\initstate\in\states$ is the {\it initial control state}.
A transition $\ctransition\in\ctransitions$  is a triple
$(\state_1,\op,\state_2)$, where $\op$ is of one of the 
three forms
(for $i=0,1$):
(i) $\inc{\counter_i}$ which increments the value of $\counter_i$ by one;
(ii) $\dec{\counter_i}$ which decrements the value of $\counter_i$ by one; or
(iii) $\ztest{\counter_i}$ which checks whether the value of 
$\counter_i$ is equal to zero.
A {\it configuration} $\conf$ of $\cmachine$ is a triple $(\state,y_0,y_1)$, where
$\state\in\states$ and $y_0,y_1\in\mathbb{N}$.
The configuration gives the control-state together with the
values of the counters $\counter_0$ and $\counter_1$.
The initial configuration $\initconf$  is 
$(q_{\it init},0,0)$.
The operational semantics of $M$ is defined in the standard manner.
In the control-state reachability problem, we are given a counter machine
$\cmachine$ and a (final) control-state $\finalstate$, and check
whether we can reach a configuration of the form
$(\finalstate,y_0,y_1)$ (for arbitrary $y_0$ and $y_1$) from $\initconf$.

Given an instance of the control-state reachability  problem for
$2$-counter machines, 
with $\cmachine=\tuple{\states,\ctransitions,\initstate}$
and $\finalstate\in\states$,
we derive an instance of the cost threshold
coverability problem for a PTPN where we have only non-negative costs on places
and only non-positive costs on transitions; and where the threshold 
vector is given by $0$.

We define the PTPN $\ptpn=\tuple{\states',\places,\transitions,\costfun}$ 
as follows.
Each control state $\state\in\states$ has a copy in $\states'$.
For each counter $\counter_i$ we have a place
$\counter_i\in\places$ with $C(\counter_i)=1$.
The number of tokens in place $\counter_i$ gives the value of the
corresponding counter.
To simplify the notation, we adopt the convention that, unless otherwise
stated, the time intervals on transitions in our PTPN are $[0:\infty)$.

Increment and decrement transitions
(on a counter $\counter_i$) are simulated straightforwardly by
changing control state and adding/removing a token
from the corresponding place.
More precisely, for a transition 
$\ctransition=\tuple{\state_1,\inc{\counter_i},\state_2}\in\ctransitions$,
there is a transition
$\ctransition\in\transitions$ such that
$\inputs(\ctransition)=\set{\state_1}$
$\outputs(\ctransition)=\set{\state_2,\counter_i}$,
and $\cost(\ctransition)=0$.
For a transition 
$\ctransition=\tuple{\state_1,\dec{\counter_i},\state_2}\in\ctransitions$ 
there is a transition
$\ctransition\in\transitions$ such that
$\inputs(\ctransition)=\set{\state_1,\counter_i}$, 
$\outputs(\ctransition)=\set{\state_2}$,
and $\cost(\ctransition)=0$.
The details of simulating a zero testing transition
$\ctransition=\tuple{\state_1,\ztest{\counter_i},\state_2}\in\ctransitions$
are shown in Figure~\ref{ztest:fig}.
The main idea
is to put a positive cost on the counter places
$\counter_i$ and $\counter_{1-i}$.
If the system `cheats' and takes the transition from
a configuration where the counter value $\counter_i$ 
is positive (the corresponding place
is not empty), then the transition will impose a cost which 
cannot be compensated
in the remainder of the computation.
On the other hand, since the other counter $\counter_{1-i}$ also has a positive cost,
we will also pay an extra (unjustified)
price corresponding to the number
of tokens in $\counter_{1-i}$.
Therefore, we use a number of auxiliary places and transitions
which make it possible to reimburse
unjustified cost for tokens on counter $\counter_{1-i}$.
The reimbursement is carried out
(at most completely, but possibly just partially)
by cycling around the tokens in $\counter_{1-i}$.
For technical reasons (see below), the construction
ensures that the only parts of a computation in which time elapses
are those that simulate the zero testing of counters (the other
parts have zero  duration).
Concretely,  
we have seven transitions 
$\transition^\ctransition_1,\ldots,\transition^\ctransition_7\in\transitions$.
Furthermore, we have four extra control states $\state^\ctransition_1,\ldots,\state^\ctransition_4$, and
two extra places
$\place^\ctransition,z$.
The place $z$ is shared among all transitions in $\transitions$ used to
simulate transitions in $\ctransitions$ that zero-test
the value of a counter.
The operation of $\ptpn$ starts by putting a single token of
age zero in $z$.
The costs are given by 
$\cost(\transition^\ctransition_3)=-2$,
$\cost(\place^\ctransition)=2$,
$\cost(\counter_{1-i})=1$; while the cost of the
other places and transitions are all equal to $0$.
\begin{figure}[tbp]
\begin{center}
\begin{tikzpicture}[label distance=-0.5mm,node distance=15mm]
  \node[place,draw=blue,fill=blue,label=above:$\state_1$] (s1){};
\node[transition,right of = s1,fill=red,label=above:$\transition^\ctransition_1$](t1){};
\node[place,draw=blue,fill=blue,right of = t1,label=above:$\state^\ctransition_1$] (qr1){};
\node[transition,right of = qr1,fill=red,label=above:$\transition^\ctransition_2$](t2){};
\node[place,draw=blue,fill=blue,right of = t2,label=above:$\state^\ctransition_2$] (qr2){};
\node[transition,right of = qr2,label=right:$-2$,label=above:$\transition^\ctransition_3$](t3){};

\node[place, below of = t3,xshift=10mm,,label=below:$2$,label= left:$\place^\ctransition$] (pr){};
\node[place, right of = pr,label=below:$1$,label= above:$\counter_{1-i}$] (c2){};
\node[place, below of = qr1] (ref){$z$};

\node[transition,below of = qr2 ,label=left :$\transition^\ctransition_7$](t7){};

\node[place,draw=blue,fill=blue,below of = s1,yshift=-15mm,label=below:$\state_2$] (s2){};
\node[transition,right of = s2,fill=red,label=below:$\transition^\ctransition_6$](t6){};
\node[place,draw=blue,fill=blue,right of = t6,label=below:$\state^\ctransition_4$] (qr4){};
\node[transition,right of = qr4,fill=red,label=below:$\transition^\ctransition_5$](t5){};
\node[place,draw=blue,fill=blue,right of = t5,label=below:$\state^\ctransition_3$] (qr3){};
\node[transition,right of = qr3,fill=red,label=below:$\transition^\ctransition_4$](t4){};

\draw[->,thick]
         (s1) edge (t1)
         (t1) edge (qr1)
         (qr1) edge (t2)
         (t2) edge (qr2)
         (qr2) edge[bend left] (t3)
         (t3) edge[bend left] (qr2)
         (c2) edge (t3)
         (t3) edge (pr)
         (t4) edge (c2)
         (pr) edge (t4)
         (qr2) edge (t7)
         (t7) edge (qr3)
         (qr3) edge[bend left] (t4)
         (t4) edge[bend left] (qr3)
         (qr3) edge (t5)
         (t5) edge (qr4)
         (qr4) edge (t6)
         (t6) edge (s2)
         (t1) edge[bend left] node [left] {\tiny $[0\!\!:\!\!0]$}  (ref)
         (ref) edge [bend left] node [left] {\tiny $[0\!\!:\!\!0]$} (t1)
         (t6) edge[bend left]  node [left] {\tiny $[0\!\!:\!\!0]$} (ref)
         (ref) edge [bend left] node [left] {\tiny $[1\!\!:\!\!1]$}  (t6)
         (t2) edge[bend left] node [right] {\tiny $[0\!\!:\!\!0]$}  (ref)
         (ref) edge [bend left] node [right] {\tiny $[1\!\!:\!\!1]$} (t2)
         (t5) edge[bend left]  node [right] {\tiny $[0\!\!:\!\!0]$} (ref)
         (ref) edge [bend left] node [right] {\tiny $[0\!\!:\!\!0]$}  (t5)
 
;
\end{tikzpicture}
\end{center}
\caption{Simulating zero testing in a TPTN. Timed transitions are filled.
The cost of a place or transition is shown only if it is different from $0$.}
\label{ztest:fig}
\end{figure}
Intuitively, $\tpn$ simulates the transition 
$\ctransition=(\state_1,\ztest{\counter_i},\state_2)$ by first
firing the transition $\transition^\ctransition_1$ moving to control state
$\state_1^\ctransition$
which signals the start of the simulation procedure.
The transition checks whether the token inside $z$ has still age
zero (by removing and putting back a token of age zero).
The computation stays in $\state_1^\ctransition$ and transition
$\transition^\ctransition_2$ will be fired after exactly one time unit
(this is 
ensured by checking that age of the token in $z$ is exactly equal to one).
Furthermore, the age of the token in $z$ is reset to zero, and the new control
state will be $\state_2^\ctransition$.
The delay will add a cost 
which is equal to the number of tokens in $\counter_i$ and
$\counter_{1-i}$.
More precisely, if $n$ is the number of tokens in place $\counter_{1-i}$,
then the mentioned unit time delay will add a cost of $n$.
We observe also that $\transition^\ctransition_2$ will move the computation to
$\state^\ctransition_2$.
This enables the next phase which will make it possible
to reclaim the (unjustified) cost we have for the tokens in the
place $\counter_{1-i}$.
We can now fire the transition $\transition^\ctransition_3$
$m$ times, where $m\leq n$, thus moving $m$ tokens
from $\counter_{1-i}$ 
to $\place^\ctransition$ and gaining $2m$ (i.e., paying $-2m$).
Eventually, $\transition^\ctransition_7$ will fire, moving control to
 $\state^\ctransition_3$.
From  $\state^\ctransition_3$, transition $\transition^\ctransition_4$
can fire $k$ times (where $k\leq m$) moving $k$
tokens back to $\counter_{1-i}$.
The places  $\place^\ctransition$  and
$\counter_{1-i}$ will now contain $m-k$ resp.\ $n+k-m$ tokens.
Eventually, transition $\transition^\ctransition_5$ will fire,
moving control to $\state^\ctransition_4$ and ensuring that no 
time has elapsed since the firing of $\transition^\ctransition_2$
(and hence no extra costs added due to delays).
Finally, transition $\transition^\ctransition_6$ will fire, ensuring a delay of one time unit (thus
costing $2(m-k)+(n+k-m)=n+m-k$), moving control
to $\state_2$, and resuming the
``normal'' simulation of $\cmachine$.
The total cost $\ell$ for the whole sequence of transitions is 
$\ell=n-2m+n+m-k=2n-m-k$.
This means $0\leq\ell$ and that $\ell=0$ only if
$k=m=n$, i.e., all the tokens of $\counter_{1-i}$ are moved
to $\place^\ctransition$ and back to  $\counter_{1-i}$.
In other words, we can reimburse all the unjustified cost
(but not more than that).

This implies correctness of the construction as follows.
Suppose that the given
instance of the control-state reachability problem has a positive
answer.
Then, there is a faithful simulation in $\tpn$
(which will eventually put a token in the place $\state_F$).
In particular, each time we perform a transition which tests the value
of counter $\counter_i$, the corresponding place is indeed empty and hence
we pay no cost for it.
We can also choose to reimburse all the unjustified cost
paid for counter $\counter_{1-i}$.
Notice that, letting time pass in the parts of the computation
that are not part of the simulation of a zero-testing transition, may
only contribute with non-negative costs, and that 
we can always choose not to have any delays in those parts.
It follows that the computation has an
accumulated cost equal to $0$.
On the other hand, suppose that the given
instance of the control-state reachability problem has a negative
answer.
Then the only way for $\tpn$ to put a token in $\state_F$ is to 
`cheat' during the simulation of a zero testing
transition (as described above).
However, in such a case  the accumulated cost for simulating the transition
is positive.
Since simulations of increment and decrement transitions have zero costs,
and simulations of zero testing transitions have non-negative costs,
the extra cost paid for cheating cannot be recovered later in the
computation.
This means that the accumulated cost for the whole computation 
will be strictly positive, and thus we have a negative instance of the
cost threshold coverability problem (with our chosen threshold of $0$).

%% file: references.tex
\newcommand{\etalchar}[1]{$^{#1}$}